\documentclass[11pt,a4paper]{article}

\usepackage[pagebackref,plainpages=false,colorlinks,linkcolor=blue,anchorcolor=teal,citecolor=blue]{hyperref}
\usepackage{amsmath} % AMS Math Package
\usepackage{amsthm} % Theorem Formatting
\usepackage{amssymb} % Math symbols such as \mathbb
\usepackage{graphicx} % Allows for eps images
\usepackage{multicol} % Allows for multiple columns
\usepackage{multirow}
\usepackage{xcolor}
\usepackage[dvips,margin=1in,bottom=1in]{geometry}
\usepackage[capitalize,noabbrev]{cleveref}

\allowdisplaybreaks[1]

\usepackage[utf8]{inputenc}
\usepackage[english]{babel}

\usepackage{diagbox}
\usepackage{mathtools}

\usepackage{xspace}
\usepackage[T1]{fontenc}
\AtBeginDocument{%
  \DeclareFontShape{T1}{cmr}{m}{scit}{<->ssub*cmr/m/sc}{}%
}
\usepackage{comment}
\usepackage{authblk}

\usepackage{enumerate}
\usepackage{enumitem}

\usepackage{makecell}
\usepackage{float}

\usepackage{adjustbox}
\usepackage{footnotehyper} 
\makesavenoteenv{table}  

\usepackage{optidef}

\usepackage{thmtools}
\declaretheorem[style=plain,numberwithin=section]{theorem}
\declaretheorem[style=plain,numberlike=theorem]{lemma,corollary}
\declaretheorem[style=remark,numberlike=theorem]{remark}
\declaretheorem[style=plain,numberlike=theorem]{definition}

\declaretheorem[style=definition,numberlike=theorem]{problem,conjecture}

\DeclarePairedDelimiter\rbra{\lparen}{\rparen}
\DeclarePairedDelimiter\sbra{\lbrack}{\rbrack}
\DeclarePairedDelimiter\cbra{\{}{\}}
\DeclarePairedDelimiter\abs{\lvert}{\rvert}
\DeclarePairedDelimiter\Abs{\lVert}{\rVert}
\DeclarePairedDelimiter\ceil{\lceil}{\rceil}
\DeclarePairedDelimiter\floor{\lfloor}{\rfloor}
\DeclarePairedDelimiter\ket{\lvert}{\rangle}
\DeclarePairedDelimiter\bra{\langle}{\rvert}
\DeclarePairedDelimiter\ave{\langle}{\rangle}

\newcommand{\ketbra}[2]{\ensuremath{\ket{#1}\!\bra{#2}}}

\newcommand{\tr} {\operatorname{tr}}
\newcommand{\poly} {\operatorname{poly}}

\newcommand{\polylog} {\operatorname{polylog}}
\newcommand{\rank} {\operatorname{rank}}

\newcommand{\yes}{{\rm yes}}
\newcommand{\no}{{\rm no}}

\newcommand{\BQP}{\textnormal{\textsf{BQP}}\xspace}

\newcommand{\PP}{\textnormal{\textsf{PP}}\xspace}
\newcommand{\QSZK}{\textnormal{\textsf{QSZK}}\xspace}
\newcommand{\NIQSZK}{\textnormal{\textsf{NIQSZK}}\xspace}

\newcommand{\NISZK}{\textnormal{\textsf{NISZK}}\xspace}
\newcommand{\QIPtwo}{\textnormal{\textsf{QIP(2)}}\xspace}
\newcommand{\PSPACE}{\textnormal{\textsf{PSPACE}}\xspace}
\newcommand{\qqQAM}{\textnormal{\textsf{qq}-\textsf{QAM}}\xspace}

\newcommand{\Real} {\operatorname{Re}}
\newcommand{\Imag} {\operatorname{Im}}

\renewcommand{\H}{\mathrm{H}}
\newcommand{\Hq}{\mathrm{H}_q}
\renewcommand{\S}{\mathrm{S}}
\newcommand{\Sq}{\mathrm{S}_q}
\newcommand{\D}{\mathrm{D}}
\newcommand{\Dq}{\mathrm{D}_q}
\newcommand{\QJS}{\textnormal{\textrm{QJS}}\xspace}
\newcommand{\measQJS}{{\mathrm{QJS}^{\mathrm{meas}}}}

\newcommand{\QJTq}{\texorpdfstring{\textnormal{QJT}\textsubscript{\textit{q}}}\xspace}
\newcommand{\QJT}[1]{{\mathrm{QJT}_{#1}}}
\newcommand{\measQJTq}{{\mathrm{QJT}_q^{\mathrm{meas}}}}
\newcommand{\JS}{\mathrm{JS}}
\newcommand{\JTq}{{\mathrm{JT}_{q}}}
\newcommand{\JT}[1]{{\mathrm{JT}_{#1}}}
\newcommand{\td}{\mathrm{T}}
\newcommand{\TV}{\mathrm{TV}}

\newcommand{\QSD}{\textnormal{\textsc{QSD}}\xspace}
\newcommand{\PureQSD}{\textnormal{\textsc{PureQSD}}\xspace}
\newcommand{\QSCMM}{\textnormal{\textsc{QSCMM}}\xspace}

\newcommand{\QED}{\textnormal{\textsc{QED}}\xspace}

\newcommand{\QEA}{\textnormal{\textsc{QEA}}\xspace}

\newcommand{\TsallisQED}{\texorpdfstring{\textnormal{\textsc{TsallisQED}\textsubscript{\textit{q}}}}\xspace}
\newcommand{\TsallisQEA}{\texorpdfstring{\textnormal{\textsc{TsallisQEA}\textsubscript{\textit{q}}}}\xspace}
\newcommand{\TsallisQEAnoq}{\texorpdfstring{\textnormal{\textsc{TsallisQEA}}}\xspace}

\newcommand{\ConstRankTsallisQED}{\texorpdfstring{\textnormal{\textsc{ConstRankTsallisQED}\textsubscript{\textit{q}}}}\xspace}
\newcommand{\ConstRankTsallisQEA}{\texorpdfstring{\textnormal{\textsc{ConstRankTsallisQEA}\textsubscript{\textit{q}}}}\xspace}

\newcommand{\binset}{\{0,1\}}
\newcommand{\innerprod}[2]{\left\langle #1 | #2 \right\rangle}
\newcommand{\innerprodF}[2]{\langle #1 , #2 \rangle}
\newcommand{\dx}{\mathrm{d}x}

\newcommand{\dn}{\mathrm{d}n}
\newcommand{\dd}{\mathrm{d}}

\newcommand{\bbC}{\mathbb{C}}
\newcommand{\bbN}{\mathbb{N}}
\newcommand{\bbR}{\mathbb{R}}
\newcommand{\bbZ}{\mathbb{Z}}
\newcommand{\rmL}{\mathrm{L}}
\newcommand{\sfA}{\mathsf{A}}
\newcommand{\sfB}{\mathsf{B}}

\newcommand{\sfE}{\mathsf{E}}
\newcommand{\sfF}{\mathsf{F}}
\newcommand{\sfO}{\mathsf{O}}
\newcommand{\ttU}{\mathtt{U}}
\newcommand{\calA}{\mathcal{A}}
\newcommand{\calB}{\mathcal{B}}
\newcommand{\calE}{\mathcal{E}}
\newcommand{\calF}{\mathcal{F}}
\newcommand{\calH}{\mathcal{H}}
\newcommand{\calI}{\mathcal{I}}

\newcommand{\calM}{\mathcal{M}}
\newcommand{\calO}{\mathcal{O}}
\newcommand{\calP}{\mathcal{P}}

\newcommand{\CNOT}{\textnormal{\textsc{CNOT}}\xspace}
\newcommand{\Iin}{\calI_\mathrm{inner}}
\newcommand{\Iout}{\calI_\mathrm{outer}}
\newcommand{\BiSearch}{{\texttt{BiSearch}}}
\newcommand{\TsallisQEAalgo}{{\texttt{TsallisQEA}_q}}

\usepackage{latexsym}

\usepackage{enumerate}

\usepackage{tcolorbox}

\usepackage{algorithm}
\usepackage{algpseudocode}
\renewcommand{\algorithmicrequire}{\textbf{Input:}} %Use Input in the format of Algorithm
\renewcommand{\algorithmicensure}{\textbf{Output:}} %Use Output in the format of Algorithm

\usepackage{tabularx}
\usepackage{booktabs}

\usepackage{subfig}
\usepackage{tikz}
\usetikzlibrary{arrows}
\usepackage{quantikz}

\begin{document}
% Reduce the line spacing between equation and text
\setlength{\abovedisplayskip}{6pt}
\setlength{\belowdisplayskip}{6pt}

\title{On estimating the trace of quantum state powers}

\author[1,2]{Yupan Liu\thanks{Email: \url{yupan.liu@epfl.ch}}}
\author[3,4,2]{Qisheng Wang\thanks{Email: \url{QishengWang1994@gmail.com}}}

\affil[1]{School of Computer and Communication Sciences, \'Ecole Polytechnique F\'ed\'erale de Lausanne}
\affil[2]{Graduate School of Mathematics, Nagoya University}
\affil[3]{School of Computer Science, Shanghai Jiao Tong University}
\affil[4]{School of Informatics, University of Edinburgh}
\date{}

\maketitle
\pagenumbering{roman}
\thispagestyle{empty}

\begin{abstract}
    We investigate the computational complexity of estimating the trace of quantum state powers $\mathrm{tr}(\rho^q)$ for an $n$-qubit mixed quantum state $\rho$, given its state-preparation circuit of size $\operatorname{poly}(n)$. 
    This quantity is closely related to and often interchangeable with the Tsallis entropy $\mathrm{S}_q(\rho) = \frac{1 - \mathrm{tr}(\rho^q)}{q-1}$, where $q = 1$ corresponds to the von Neumann entropy. 
    For any non-integer $q \geq 1 + \Omega(1)$, we provide a quantum estimator for $\mathrm{S}_q(\rho)$ with time complexity $\operatorname{poly}(n)$, \textit{exponentially} improving the prior best results of $\exp(n)$ due to \hyperlink{cite.AISW20}{Acharya, Issa, Shende, and Wagner (ISIT 2019)}, \hyperlink{cite.WGL+22}{Wang, Guan, Liu, Zhang, and Ying (TIT 2024)}, \hyperlink{cite.WZL24}{Wang, Zhang, and Li (TIT 2024)}, and \hyperlink{cite.WZ24b}{Wang and Zhang (TIT 2025)}. 
    Our speedup is achieved by introducing efficiently computable \textit{uniform approximations} of positive power functions into quantum singular value transformation.

    Our quantum algorithm reveals a sharp phase transition between the case of $q=1$ and constant $q>1$ in the computational complexity of the \textsc{Quantum $q$-Tsallis Entropy Difference Problem} ($\textsc{TsallisQED}_q$), particularly deciding whether the difference $\mathrm{S}_q(\rho_0) - \mathrm{S}_q(\rho_1)$ is at least $0.001$ or at most $-0.001$:
    \begin{itemize}
        \item For any $1+\Omega(1) \leq q \leq 2$, $\textsc{TsallisQED}_q$ is $\mathsf{BQP}$-complete, which implies that \textsc{Purity Estimation} is also $\mathsf{BQP}$-complete. 
        \item For any $1 \leq q \leq 1 + \frac{1}{n-1}$, $\textsc{TsallisQED}_q$ is $\mathsf{QSZK}$-hard, leading to hardness of approximating the von Neumann entropy because $\mathrm{S}_q(\rho) \leq \mathrm{S}(\rho)$, as long as $\mathsf{BQP} \subsetneq \mathsf{QSZK}$.
    \end{itemize}
    The hardness results are derived from reductions based on new inequalities for the quantum $q$-Jensen--(Shannon--)Tsallis divergence with $1\leq q \leq 2$, which are of independent interest.
\end{abstract}

\newpage
\tableofcontents
\thispagestyle{empty}

%%%%%%%%%%%%%%%%%%%%%%%%%%%%%%%%%%%%%%%%%%%%%%%%%%%%%%%
\newpage
\pagenumbering{arabic}
\section{Introduction}

In recent years, the development of quantum devices has posed an intriguing challenge of verifying their intended functionality. Typically, a quantum device is designed to prepare an $n$-qubit (mixed) state $\rho$. The problem of (tolerant) quantum state testing aims to design algorithms that can efficiently test whether a quantum state approximately has a certain property, assuming the state either nearly has the property or is somehow ``far'' from having it.  
This problem is a quantum (non-commutative) generalization of classical (tolerant) distribution testing (see~\cite{Canonne20}) and classical property testing in general (see~\cite{Goldreich17}). Furthermore, this problem is an instance of the emerging field of quantum property testing (see~\cite{MdW16}), which focuses on devising (efficient) quantum testers for properties of quantum objects. 

The general upper bound for (tolerant) quantum state testing depends (at least) linearly on the dimension (e.g.,~\cite[Section 4.2]{MdW16}), whereas some properties of quantum states can be tested significantly more efficiently than the general case. 
A simple and interesting example is the property \textsc{Purity}, where $\rho$ satisfies the property if and only if it is a pure state. This example is essentially an instance of estimating the trace of quantum state powers, specifically $\tr(\rho^2)$. A natural approach to test \textsc{Purity} is to apply the SWAP test~\cite{BCWdW01} to two copies of $\rho$, and this algorithm accepts with probability $(1+\tr(\rho^2))/2$, which is equal to $1$ if and only if $\rho$ is pure. Further analysis deduces that \textsc{Purity} can be tolerantly tested with $O(1/\epsilon^2)$ copies of $\rho$.\footnote{The sample (or query) complexity for \textsc{Purity} differs between one-sided or two-sided error scenarios. Our upper bound applies to the latter, while the sample complexity for the former is $O(1/\epsilon)$~\cite[Section 4.2]{MdW16}.} 
Meanwhile, Ekert et al.~\cite{EAO+02} presented an efficient quantum algorithm for estimating $\tr(\rho^q)$ where $q>1$ is an integer. These two fundamental works raise two interesting questions:  

\begin{enumerate}[label={\upshape(\alph*)}, topsep=0.33em, itemsep=0.33em, parsep=0.33em]
    \item Is there an efficient quantum algorithm for estimating the trace of quantum state powers $\tr(\rho^q)$ for any non-integer $q > 1$? \label{Qitem:efficient-algorithm}
    \item Can estimating the trace of quantum state powers, e.g., $\tr(\rho^2)$, fully capture the computational power of quantum computing, namely \BQP{}-complete?  \label{Qitem:purity-BQPhard}
\end{enumerate}

Notably, the trace of quantum state powers $\tr(\rho^q)$ is closely related to the \textit{power} quantum entropy of order $q$. Particularly, the quantum $q$-Tsallis entropy $\Sq(\rho)$, which is a non-additive (but still concave) generalization of the von Neumann entropy $\S(\rho)$, with the von Neumann entropy being the limiting case of the quantum $q$-Tsallis entropy as $q$ approaches $1$:
\[ \Sq(\rho) = \frac{1 - \tr\rbra*{\rho^q}}{q-1} \quad \text{and} \quad \lim_{q \rightarrow 1} \Sq(\rho) = \S(\rho) = -\tr\rbra*{\rho \log(\rho)}. \]
As a consequence, $\Sq(\rho)$ can naturally provide a lower bound for $\S(\rho)$ when considering $\Sq(\rho)$ with $q=1+\epsilon$, where $\epsilon$ can be a small constant, such as $q=1.0001$.
This observation serves as the first reason motivating Question \ref{Qitem:efficient-algorithm}.

The study of power entropy dates back to Havrda and Charv{\'a}t~\cite{HC67}. Since then, it has been rediscovered independently by Dar{\'o}czy~\cite{Daroczy70}, and finally popularized by Tsallis~\cite{Tsallis88}. Raggio~\cite{Raggio95} expanded on this study by introducing the quantum Tsallis entropy. 
Tsallis entropy has been particularly useful in physics for describing systems with non-extensive properties, such as long-range interactions, in statistical mechanics (see~\cite{Tsallis01}).

A notable example is the Tsallis entropy $\Hq(p)$ with $q=3/2$, which is useful for modeling systems where both frequent and rare events matter.\footnote{In contrast, the Tsallis entropy with $q=2$ (Gini impurity) is very sensitive to rare events.} For instance, in fluid dynamics, the distribution that maximizes $\H_{3/2}$ helps model velocity changes in turbulent flows~\cite{Beck02}.
This example provides the second reason motivating Question \ref{Qitem:efficient-algorithm}, as existing efficient quantum algorithms~\cite{BCWdW01,EAO+02} are designed only for integer $q \geq 2$. Estimating $\Sq(\rho)$ for non-integer $q$ between $1$ and $2$, therefore, appears to be computationally challenging. 

\vspace{1em}
In this paper, we focus on estimating the trace of quantum state powers, or equivalently, the \textsc{Quantum $q$-Tsallis Entropy Difference Problem} (\TsallisQED{}) and the \textsc{Quantum $q$-Tsallis Entropy Approximation Problem} (\TsallisQEA{}). These two problems constitute the (white-box) quantum state testing problem with respect to the quantum $q$-Tsallis entropy. For \TsallisQED{}, we consider two polynomial-size quantum circuits (devices), denoted as $Q_0$ and $Q_1$, which prepare $n$-qubit quantum states $\rho_0$ and $\rho_1$, respectively, with access to the descriptions of these circuits. 
Our goal is to decide whether the difference $\Sq(\rho_0)-\Sq(\rho_1)$ is at least $0.001$ or at most $-0.001$.\footnote{It is noteworthy that $0.001$ is just an arbitrary constant for the precision parameter, which can be replaced by any inverse polynomial function in general. See \Cref{def:TsallisQED,def:TsallisQEA} for formal definitions.} The setting of \TsallisQEA{} is similar to \TsallisQED{}, except that we only consider a single $n$-qubit quantum state $\rho$, and the task is to decide whether the difference $\Sq(\rho)-t(n)$ is at least $0.001$ or at most $-0.001$, where $t(n)$ is a known threshold.

Next, we will state our main results and then provide justifications for their significance.

\subsection{Main results}

We begin by presenting our first main result, which provides a positive answer to Question \ref{Qitem:efficient-algorithm} for the regime $q \geq 1+\Omega(1)$:\footnote{We implicitly assume that $q$ satisfies $1 + \Omega\rbra{1} \leq q \leq O\rbra{1}$. 
Since $\Sq\rbra{\rho} \leq o\rbra{1}$ when $q=\omega(1)$, it is reasonable to consider constantly large $q$.}

\begin{theorem} [Quantum estimator for $q$-Tsallis entropy]
    \label{thm:TsallisQE-containment-informal} 
    Given quantum query access to the state-preparation circuit of an $n$-qubit quantum state $\rho$, for any $q \geq 1 + \Omega\rbra{1}$, there is a quantum algorithm for estimating $\Sq\rbra{\rho}$ to additive error $0.001$ with query complexity $O\rbra{1}$.
    Moreover, if the description of the state-preparation circuit is of size $\poly\rbra{n}$, then the time complexity of the quantum algorithm is $\poly\rbra{n}$.
    Consequently, for any $q \geq 1+\Omega(1)$, \TsallisQED{} and \TsallisQEA{} are in \BQP{}.
\end{theorem}

More specifically, when the desired additive error is set to $\epsilon$, the explicit query complexity of \cref{thm:TsallisQE-containment-informal} becomes $O\rbra{1/\epsilon^{1+\frac{1}{q-1}}}$, or expressed as $\poly\rbra{1/\epsilon}$ (see \cref{thm:tr-power-constant-queries}). 
Moreover, if the state-preparation circuit of $\rho$ is of size $L(n)=\poly\rbra{n}$, \cref{thm:TsallisQE-containment-informal} provides a quantum algorithm with time complexity $O\rbra{L/\epsilon^{1+\frac{1}{q-1}}}$, or equivalently, $\poly\rbra{n, 1/\epsilon}$. 
Using the same idea, we can also derive an upper bound $\widetilde O\rbra{1/\epsilon^{3+\frac{2}{q-1}}}$, or expressed as $\poly\rbra{1/\epsilon}$, for the sample complexity needed to estimate $\Sq\rbra{\rho}$ (see \cref{thm:tr-power-constant-samples}).
This is achieved by applying the \textit{samplizer} from \cite{WZ24b}, which allows a quantum query-to-sample simulation.

\vspace{1em}
There are several quantum algorithms for estimating the $q$-Tsallis entropy of an $n$-qubit mixed quantum state $\rho$ for non-integer constant $q > 1$ proposed in \cite{AISW20,WGL+22,WZL24,WZ24b}, all of which turn out to have time complexity $\exp\rbra{n}$ in the setting that $\rho$ is given by its state-preparation circuit of size $\poly\rbra{n}$.
\begin{itemize}
    \item In \cite[Theorem 3]{AISW20} and \cite[Theorem 1.2]{WZ24b}, for non-integer constant $q > 1$, they proposed quantum algorithms for estimating the $q$-R{\'e}nyi entropy of an $n$-qubit quantum state $\rho$ by using $\mathsf{S} = \poly\rbra{1/\epsilon} \cdot \exp\rbra{n}$ samples of $\rho$ and $\mathsf{T} = \poly\rbra{1/\epsilon} \cdot \exp\rbra{n}$ quantum gates.\footnote{The explicit sample complexities of the approaches of \cite[Theorem 3]{AISW20} and \cite[Theorem 2]{WZ24b} are $O\rbra{2^{2n}/\epsilon^2}$ and $O\rbra{2^{\rbra{\frac{4}{q}-2}n}/\epsilon^{1+\frac{4}{q}} \cdot \poly\rbra{n, \log\rbra{1/\epsilon}}}$, respectively, both of which are $\poly\rbra{1/\epsilon} \cdot \exp\rbra{n}$.
    The number of quantum states in the approach of \cite[Theorem 3]{AISW20} was mentioned in \cite{WZ24b} to be $O\rbra{ \rbra{2^{2n}/\epsilon^2}^3 \cdot \polylog\rbra{2^{n}, 1/\epsilon}} = \poly\rbra{1/\epsilon} \cdot \exp\rbra{n}$ by using the weak Schur sampling in \cite[Section 4.2.2]{MdW16} and the quantum Fourier transform over symmetric groups \cite{KS16}.
    Another possible implementation noted in \cite{Hayashi24} is to use the Schur transform in \cite{Nguyen23}, resulting in $O\rbra{2^{2n}/\epsilon^2 \cdot 2^{4n} \cdot \polylog\rbra{2^n,1/\epsilon}} = \poly\rbra{1/\epsilon} \cdot \exp\rbra{n}$.}
    Their result implies an estimator for $\Sq\rbra{\rho}$ with the same complexity,
    because any estimator for $q$-R{\'e}nyi entropy implies an estimator for $q$-Tsallis entropy with the same parameter for $q > 1$ (as noted in \cite[Appendix A]{AOST17}).
    By preparing each sample of $\rho$ using its state-preparation circuit of size $\poly\rbra{n}$, one can estimate $\Sq\rbra{\rho}$ by using their estimators with overall time complexity $\mathsf{S} \cdot \poly\rbra{n} + \mathsf{T} = \poly\rbra{1/\epsilon} \cdot \exp\rbra{n}$.
    \item In \cite[Theorem III.9]{WGL+22}, for non-integer constant $q > 1$, they proposed a quantum algorithm for estimating $\Sq\rbra{\rho}$ with query complexity $\widetilde O\rbra{r^{1/\cbra{\frac{q-1}{2}}}/\epsilon^{1+1/\cbra{\frac{q-1}{2}}}} = \poly\rbra{r, 1/\epsilon}$, where $r$ is (an upper bound on) the rank of $\rho$ and $\cbra{x} \coloneqq x - \floor{x}$ denotes the fractional part of $x$.
    In \cite[Corollary 5]{WZL24}, for non-integer constant $q > 1$, they proposed a quantum algorithm for estimating the $q$-R{\'e}nyi entropy of a quantum state with query complexity $\widetilde O\rbra{r/\epsilon^{1+\frac{1}{q}}} = \poly\rbra{r, 1/\epsilon}$, which also implies a quantum algorithm for estimating $\Sq\rbra{\rho}$ with query complexity $\poly\rbra{r, 1/\epsilon}$ (the reason has been discussed in the last item).
    For $n$-qubit quantum state $\rho$ without prior knowledge, by taking $r = 2^n$, their query complexity is then $\poly\rbra{2^n, 1/\epsilon} = \poly\rbra{1/\epsilon} \cdot \exp\rbra{n}$, which is exponentially larger than our $\poly\rbra{n, 1/\epsilon}$.
\end{itemize}

\vspace{1em}
Our efficient quantum estimator for $\Sq(\rho)$ where $q \geq 1+\Omega(1)$ (\Cref{thm:TsallisQE-containment-informal}), combined with our hardness results for $\TsallisQED{}$ and $\TsallisQEA{}$ (\Cref{thm:TsallisQE-hardness-informal}), indicates a sharp phase transition between the case of $q=1$ and constant $q>1$ and answers to Question \ref{Qitem:efficient-algorithm} and \ref{Qitem:purity-BQPhard}. For clarity, we summarize our main results in \Cref{table:computational-hardness-TsallisQED-TsallisQEA}. 

\begin{table}[H]
\centering
\adjustbox{max width=\textwidth}{
\begin{tabular}{ccccc}
    \toprule
    & $q=1$ & $1 < q \leq 1\!+\!\frac{1}{n-1}$ & $1 \!+\! \Omega(1) \leq q \leq 2$ & $q > 2$\\
    \midrule
    \TsallisQED{} 
    & \makecell{\QSZK{}-complete\\ \footnotesize{\cite{BASTS10}}}
    & \makecell{\QSZK{}-hard\\ \footnotesize{\Cref{thm:TsallisQE-hardness-informal}\ref{thmitem:TsallisQE-QSZKhard}}} 
    & \makecell{\BQP{}-complete\\ \footnotesize{\Cref{thm:TsallisQE-containment-informal,thm:TsallisQE-hardness-informal}\ref{thmitem:TsallisQE-BQPhard}}} 
    & \makecell{in \BQP{}\\ \footnotesize{\Cref{thm:TsallisQE-containment-informal}}}\\
    \midrule
    \TsallisQEA{} 
    & \makecell{\NIQSZK{}-complete\\ \footnotesize{\cite{BASTS10,CCKV08}}}
    & \makecell{\NIQSZK{}-hard*\\ \footnotesize{\Cref{thm:TsallisQE-hardness-informal}\ref{thmitem:TsallisQE-QSZKhard}}} 
    & \makecell{\BQP{}-complete\\ \footnotesize{\Cref{thm:TsallisQE-containment-informal,thm:TsallisQE-hardness-informal}\ref{thmitem:TsallisQE-BQPhard}}} 
    & \makecell{in \BQP{}\\ \footnotesize{\Cref{thm:TsallisQE-containment-informal}}}\\
    \bottomrule
\end{tabular}
}
\caption{Computational hardness of \TsallisQED{} and \TsallisQEA{}.}
\label{table:computational-hardness-TsallisQED-TsallisQEA}
\end{table}

Here, \QSZK{} and \NIQSZK{} are the classes of promise problems possessing quantum statistical zero-knowledge and non-interactive quantum statistical zero-knowledge, respectively, as introduced by~\cite{Wat02,Wat09} and~\cite{Kobayashi03}. The asterisk in \Cref{table:computational-hardness-TsallisQED-TsallisQEA} indicates that \TsallisQEA{} is \NIQSZK{}-hard for a specific $q(n)=1+\frac{1}{n-1}$, as detailed in \Cref{thm:TsallisQE-hardness-informal}\ref{thmitem:TsallisQE-QSZKhard}.

For the case of $q=1$, \TsallisQED{} and \TsallisQEA{} coincide with the \textsc{Quantum Entropy Difference Problem} (\QED{}) and the \textsc{Quantum Entropy Approximation Problem} (\QEA{}) introduced in~\cite{BASTS10}, respectively. Moreover, \QED{} is complete for the class \QSZK{}~\cite{BASTS10}, whereas \QEA{} is complete for the class \NIQSZK{}~\cite{BASTS10,CCKV08}. These two classes contain \BQP{} and are seemingly much harder than \BQP{}.\footnote{Following the oracle separation between \NISZK{} and \PP{}~\cite{BCHTV19}, it holds that $\NIQSZK^{\calO} \not\subseteq \PP^{\calO}$ and likewise $\QSZK^{\calO} \not\subseteq \PP^{\calO}$ for some classical oracle $\calO$.} Meanwhile, the best known upper bound for \QSZK{} is \QIPtwo{} with a quantum linear-space honest prover~\cite{LGLW23}, and the best known upper bound for \NIQSZK{} is \qqQAM{}~\cite{KLGN19}, both of which are contained in $\QIPtwo{} \subseteq \PSPACE$~\cite{JUW09}. 

In terms of \textit{quantitative} bounds on quantum query and sample complexities, \QSZK{}-hard or \NIQSZK{}-hard in the white-box setting correspond to rank-dependent complexities in black-box settings. Specifically, we establish lower bounds for both the easy regime $q \geq 1+\Omega(1)$ and the hard regime $1 < q \leq 1+\frac{1}{n-1}$, with the upper bounds for the hard regime derived from those for estimating quantum R\'enyi entropy, as detailed in \Cref{table:query-and-sample-complexity}. 

\begin{table}[!htp]
\centering
\adjustbox{max width=\textwidth}{
\begin{tabular}{ccccc}
\toprule
\multirow{2}{*}{Regime of $q$}       & \multicolumn{2}{c}{Query Complexity} & \multicolumn{2}{c}{Sample Complexity} \\ \cmidrule{2-5}
 & Upper Bound & Lower Bound & Upper Bound & Lower Bound \\ \midrule
\multirow{2}{*}{$q \geq 1 + \Omega\rbra{1}$} &
 $O\rbra{1/\epsilon^{1+\frac{1}{q-1}}}$  &  $\Omega\rbra{1/\sqrt{\epsilon}}$   & $\widetilde O\rbra{1/\epsilon^{3+\frac{2}{q-1}}}$  & $\Omega\rbra{1/\epsilon}$ \\ 
& \footnotesize{\cref{thm:tr-power-constant-queries}} & \footnotesize{\cref{thm:Tsallis-query-complexity-lower-bound-largeQ}} & \footnotesize{\cref{thm:tr-power-constant-samples}} & \footnotesize{\cref{thm:Tsallis-sample-complexity-lower-bound-largeQ}} \\ \midrule 
\multirow{2}{*}{$1 < q \leq 1 + \frac{1}{n-1}$}     & $\widetilde O\rbra{r/\epsilon^{2}}$ & $\Omega\rbra{r^{0.17-c}}$\footnote{In these bounds, $c > 0$ is a constant that can be made arbitrarily small, and we set $c'=3c$. \label{footnote:constant-in-lower-bounds}} & $\widetilde O\rbra{r^{2}/\epsilon^{5}}$\footnote{In the regime $1 \leq q \leq 1+\frac{1}{n-1}$, as the rank $r$ approaches $2^n$, a sample complexity upper bound of $O\rbra{4^n/\epsilon^2}$ with better dependence on $\epsilon$ was given in~\cite{AISW20}. \label{footnote:fullrank-sample-upper-bounds}} & $\Omega\rbra{r^{0.51-c'}}$\footref{footnote:constant-in-lower-bounds} \\ 
& \footnotesize{\cite{WZL24}} & \footnotesize{\cref{thm:Tsallis-query-complexity-lower-bound-smallQ}} & \footnotesize{\cite{WZ24b}} & \footnotesize{\cref{thm:Tsallis-sample-complexity-lower-bound-smallQ}} \\ \midrule
\multirow{2}{*}{$q = 1$} & $\widetilde O\rbra{r/\epsilon^{2}}$\footnote{As the rank $r$ approaches $2^n$, a better query complexity upper bound of $\widetilde O\rbra{2^n/\epsilon^{1.5}}$ was shown in~\cite{GL20}.} & $\widetilde \Omega\rbra{\sqrt{r}}$ & $\widetilde O\rbra{r^2/\epsilon^5}$\footref{footnote:fullrank-sample-upper-bounds} & $\Omega\rbra{r/\epsilon}$ \\
& \footnotesize{\cite{WGL+22}} & \footnotesize{\cite{BKT20}} & \footnotesize{\cite{WZ24b}} & \footnotesize{\cite{WZ24b}} \\
\bottomrule
\end{tabular}
}
\caption{(Rank-dependent) bounds on query and sample complexities for estimating $\Sq\rbra{\rho}$.}
\label{table:query-and-sample-complexity}
\end{table}

\vspace{1em}
On the other hand, understanding why the regime $q \geq 1+\Omega(1)$ is computationally easy can be illustrated by the case of $q=2$ (\textsc{Purity Estimation}), particularly deciding whether $\tr(\rho^2)$ is at least $2/3$ or at most $1/3$. Let $\{\lambda_k\}_{k \in [2^n]}$ be the eigenvalues of an $n$-qubit quantum state $\rho$. For any quantum state $\hat{\rho}$ having eigenvalues at most $1/n$, it follows that $\tr(\hat{\rho}^2) = \sum_{k \in [2^n]} \lambda_k^2 \leq n \cdot n^{-2} = 1/n$, hence $0$ provides a good estimate of $\tr(\hat{\rho}^2)$ to within additive error $1/3$. This intuition implies that only sufficiently large eigenvalues contribute to estimating the value of $\tr(\rho^2)$. Consequently, the computational complexity of \textsc{Purity Estimation} is supposed to be independent of the rank $r$. 

However, this argument is just the first step towards establishing an efficient quantum estimator for $\Sq(\rho)$.\footnote{A similar argument also applies to the classical Tsallis entropy, see~\cite[Section III.C]{AOST17}. Nevertheless, this type of argument does not extend to von Neumann entropy ($q=1$), see~\cite[Section 7]{QKW22}.} We also need to estimate $\sum_{k\in\calI_{\rm large}} \lambda_k^q$, where $\calI_{\rm large}$ is the index set for sufficiently large eigenvalue $\lambda_k$. 
For the case of integer $q>1$, the approach of~\cite{BCWdW01,EAO+02} equipped with quantum amplitude estimation~\cite{BHMT02} provides a solution, whereas the case of non-integer $q\geq 1+\Omega(1)$ is more challenging and requires more sophisticated techniques. 
Notably, the task is finally resolved by our first main result (\Cref{thm:TsallisQE-containment-informal}). 

\vspace{1em}
Lastly, we provide our second main result, namely the computational hardness for \TsallisQED{} and \TsallisQEA{}, as stated in \Cref{thm:TsallisQE-hardness-informal}. 
Let \ConstRankTsallisQED{} and \ConstRankTsallisQEA{} denote restricted variants of \TsallisQED{} and \TsallisQEA{}, respectively, such that the ranks of the states of interest are at most $O(1)$. 

\begin{theorem}[Computational hardness for \TsallisQED{} and \TsallisQEA{}, informal]
    \label{thm:TsallisQE-hardness-informal}
    The promise problems \TsallisQED{} and \TsallisQEA{} capture the computational power of their respective complexity classes in the corresponding regimes of $q$\emph{:}\footnote{For detailed definitions of Karp reduction and Turing reduction, please refer to \Cref{subsubsec:input-models-and-reductions}.}
    \begin{enumerate}[label={\upshape(\arabic*)}, topsep=0.33em, itemsep=0.33em, parsep=0.33em]
        \item \emph{\textbf{Easy regimes:}} For any $q \in [1,2]$, \ConstRankTsallisQED{} is \BQP{}-hard under Karp reduction, and consequently, \ConstRankTsallisQEA{} is \BQP{}-hard under Turing reduction. As a corollary, \TsallisQED{} and \TsallisQEA{} are \BQP{}-complete{} for $1+\Omega(1) \leq q \leq 2$. \label{thmitem:TsallisQE-BQPhard}
        \item \emph{\textbf{Hard regimes:}} For any $q \in \left(1,1+\frac{1}{n-1}\right]$, \TsallisQED{} is \QSZK{}-hard under Karp reduction, and consequently, \TsallisQEA{} is \QSZK{}-hard under Turing reduction. Furthermore, for $q=1+\frac{1}{n-1}$, \TsallisQEA{} is \NIQSZK{}-hard under Karp reduction. \label{thmitem:TsallisQE-QSZKhard}   
    \end{enumerate}
\end{theorem}

It is noteworthy that \BQP{}-hardness under Turing reduction is as strong as \BQP{}-hardness under Karp reduction, due to the \BQP{} subroutine theorem~\cite{BBBV97}.\footnote{Once we have an efficient quantum algorithm $\calA$ for \TsallisQEA{}, any problem in \BQP{} can be solved using $\calA$ as a subroutine. The \BQP{} subroutine theorem, as stated in~\cite[Section 4]{BBBV97}, implies that $\BQP^{\calA} \subseteq \BQP$. } 
Moreover, \Cref{thm:TsallisQE-hardness-informal} implies a direct corollary, offering a positive answer to Question \ref{Qitem:purity-BQPhard}:
\begin{corollary}
    \label{corr:purity-BQPhard-informal}
    \textsc{Purity Estimation} is \BQP{}-hard. 
\end{corollary}
Interestingly, the \BQP{}-hardness for a similar problem, specifically deciding whether $\tr(\rho_0\rho_1)$ is at least $2/3$ or at most $1/3$, turns out to be not difficult to show.\footnote{For any \BQP{} circuit $C_x$, the acceptance probability $\| \ketbra{1}{1}_{\rm out} C_x \ket{\bar{0}}\|_2^2 = \tr\rbra*{\ketbra{1}{1}_{\rm out} C_x \ketbra{\bar{0}}{\bar{0}} C_x^{\dagger}} = \tr(\rho_0 \rho_1)$, where $\rho_0 \coloneqq \ketbra{1}{1}_{\rm out}$ and $\rho_1 \coloneqq \tr_{\overline{\rm out}}\rbra*{C_x \ketbra{\bar{0}}{\bar{0}} C_x^{\dagger}}$. Similar observations appeared in~\cite[Theorem 9]{Kobayashi03}.} However, this result does not imply \Cref{corr:purity-BQPhard-informal}.  

\subsection{Proof techniques: \BQP{} containment for \texorpdfstring{$q$}{q} constantly larger than \texorpdfstring{$1$}{1}}

The proof of \Cref{thm:TsallisQE-containment-informal} consists of an efficient quantum (query) algorithm for estimating the value of $\tr\rbra{\rho^q}$ for $q > 1$, given quantum query access to the state-preparation circuit $Q$ of the mixed quantum state $\rho$.
Our approach to estimating $\tr\rbra{\rho^q}$ is via one-bit precision phase estimation~\cite{Kitaev95}, also known as the Hadamard test \cite{AJL09}, equipped with the quantum singular value transformation (QSVT) \cite{GSLW19}. 
Our algorithm is sketched in the following four steps (see \cref{sec:algos} for more details):
\begin{enumerate}[topsep=0.33em, itemsep=0.33em, parsep=0.33em] 
    \item Find a good polynomial approximation of $x^{q-1}$.
    \item Implement a unitary block-encoding $U$ of $\rho^{q-1}$ using QSVT, with the state-preparation circuit $Q$.
    \item Perform the Hadamard test on $U$ and $\rho$ with outcome $b \in \cbra{0, 1}$.
    \item One can learn the value of $\tr\rbra{\rho^q}$ from a good estimate of $b$ via quantum amplitude estimation. 
\end{enumerate}

The idea is simple. 
Similar ideas were ever used to estimate the fidelity \cite{GP22}, trace distance \cite{WZ24,LGLW23}, and von Neumann entropy \cite{LGLW23,WZ24b}. 
However, all of the aforementioned quantum algorithms have query or time complexity polynomials in the rank $r$ of quantum states. 
Additionally, all these prior works rely on the quantum singular value transformation~\cite{GSLW19}, which is a technique for designing quantum algorithms by approximating the target functions.\footnote{For example, estimating the fidelity and trace distance requires to approximate the sign function; and estimating the von Neumann entropy requires to approximate the logarithmic function.}
The main technical reason is that the functions to be approximated in their key steps are not smooth in the whole range of $\sbra{0, 1}$, so they have to use the polynomial approximations of piece-wise smooth functions in \cite[Corollary 23]{GSLW19} to avoid the bad part (which is actually the regime of tiny eigenvalues);\footnote{These eigenvalues correspond to the inputs of the target function.}
this results in an estimation error dependent on $r$ because, technically, the error for each bad eigenvalue has to be bounded individually (there are at most $r$ bad eigenvalues), thereby introducing an (at least) linear $r$-dependence.
Specifically, in their approaches, a target function $f\rbra{x}$ is specified and the goal is to estimate the value of $\tr\rbra{\rho f\rbra{\rho}}$.
For example, $f\rbra{x} = -\log\rbra{x}$ for estimating the von Neumann entropy. 
The target function $f\rbra{x}$ is usually only approximated well in the range $x \in \sbra{\delta, 1}$ for some parameter $\delta$, while leaving the rest range of $x$ unspecified;
more precisely, $f\rbra{x}$ is approximated by a polynomial $P\rbra{x}$ by, e.g., \cite[Corollary 23]{GSLW19}, such that 
\begin{equation} \label{eq:bad-poly-condition}
    \max_{x\in\sbra{\delta, 1}} \abs{P\rbra{x} - f\rbra{x}} \leq \epsilon,~\max_{x\in\sbra{-1, 1}} \abs{P\rbra{x}} \leq 1, \text{ and } \deg\rbra{P} = O\rbra*{\frac{1}{\delta}\log\frac{1}{\epsilon}}.
\end{equation}
Then, they instead estimate the value of $\tr\rbra{\rho P\rbra{\rho}}$. 
The intrinsic error turns out to be 
\[ \abs{\tr\rbra{\rho f\rbra{\rho}} - \tr\rbra{\rho P\rbra{\rho}}} \leq \sum_{\lambda_j < \delta} \abs*{\lambda_j f(\lambda_j) \!-\! \lambda_j P(\lambda_j)} + \sum_{\lambda_j \geq \delta} \abs*{\lambda_j f(\lambda_j) \!-\! \lambda_j P(\lambda_j)} 
\leq r \cdot \poly\rbra{\delta} + O\rbra{\epsilon}. \]
Here, $\{\lambda_j\}_{1 \leq j \leq 2^n}$ are the eigenvalues of the state $\rho$, with each $\lambda_j$ satisfying $0 \leq \lambda_j \leq 1$. To make the intrinsic error bounded, $\delta$ must be sufficiently small, e.g., $\delta = 1/\poly\rbra{r}$.

The above standard method has drawbacks: the intrinsic error is $r \cdot \poly\rbra{\delta}$ for the small-eigenvalue part and $O\rbra{\epsilon}$ for the large-eigenvalue part. While the $\epsilon$-dependence in the approximation degree is logarithmic (and thus not the dominating term), the $\delta$-dependence is significant. 
This suggests the need for the following trade-off: Can we reduce the error caused by the small-eigenvalue part, at the cost of a possibly worse error caused by the large-eigenvalue part?

To make this trade-off possible for our purpose, we turn to find polynomials that uniformly approximate the positive power functions.
This is inspired by the Stone-Weierstrass theorem, stating that any continuous function (e.g., $x^q$) on a closed interval (e.g., $\sbra{0, 1}$) can be uniformly approximated by polynomials.
The study of the \textit{best uniform approximation} (by polynomials)\footnote{The best uniform approximation polynomial of a continuous function $f(x)$ on $[-1,1]$ is a degree-$d$ polynomial that minimizes $\max_{x\in[-1,1]} |f(x)-P_d(x)|$ over all degree-$d$ polynomials $P_d$. For a formal definition, see \Cref{subsubsec:best-uniform-poly-approx}.} of positive power functions was initiated by Bernstein~\cite{Bernstein14, Bernstein38} almost a century ago in an abstract manner.\footnote{Actually, the function $\abs{x}^q$ for $x \in \sbra{-1, 1}$ is commonly considered in the literature. Nevertheless, we are only interested in the non-negative part, i.e., the range $\sbra{0, 1}$.} 
The best uniform approximation polynomial of $x^q$ was shown with a non-constructive proof in \cite[Section 7.1.41]{Timan63}, stating that there is a family of polynomials $P_d\rbra{x}$ of degree $d$ such that 
\begin{equation} \label{eq:good-poly-condition}
    \max_{x\in\sbra{0, 1}} \abs{P_d \rbra{x} - x^q} \to \frac{1}{d^{q}}, \text{ as } d \to \infty,
\end{equation}
whose approximation range is in sharp contrast to that in \cref{eq:bad-poly-condition}.
However, the coefficients of the leading error terms and the explicit construction of these polynomial approximations seem still not fully understood (e.g.,~\cite{Ganzburg02}). Consequently, it is somewhat challenging to directly use such polynomial approximations (e.g.,~\cite[Section 7.1.41]{Timan63}) in a time-efficient manner. 

Inspired by the result of the best uniform approximation of positive power functions in \cite{Timan63}, we, instead, aim to find a good enough uniform approximation that is also efficiently computable.
This is achieved by employing the construction of asymptotically best uniform approximation via combining Chebyshev truncations and the de La Vall\'ee Poussin partial sum (cf.\ \cite[Chapter 3]{Rivlin90}).
Finally, we obtain a family of efficiently computable uniform approximation polynomials of (scaled) $x^q$ that are suitable for QSVT:
\begin{equation} \label{eq:good-enough-poly}
    \max_{x\in\sbra{0, 1}} \abs*{P\rbra{x} - \frac{1}{2}x^q} \leq \epsilon,~\max_{x \in \sbra{-1, 1}}\abs{P\rbra{x}} \leq 1, \text{ and } \deg\rbra{P} = O\rbra*{\frac{1}{\epsilon^{1/q}}}.
\end{equation}

Using these efficiently computable uniform approximation polynomials, we are able to give a quantum algorithm for estimating $\tr\rbra{\rho^q}$. 
First, we approximate the function $x^{q-1}$ in the range $\sbra{0, 1}$ to error $\epsilon$ by a polynomial of degree $O\rbra{1/\epsilon^{\frac{1}{q-1}}}$.
Then, we can apply the algorithm sketched at the very beginning of this subsection. 
With further analysis, we can estimate the value of $\tr\rbra{\rho^q}$ to additive error $\epsilon$ with quantum query complexity $O\rbra{1/\epsilon^{1+\frac{1}{q-1}}}$ (see \cref{thm:tr-power-constant-queries}).
Using the same idea, we can also estimate $\tr\rbra{\rho^q}$ to additive error $\epsilon$ by using $\widetilde O\rbra{1/\epsilon^{3+\frac{2}{q-1}}}$ copies of $\rho$ through the samplizer \cite{WZ24b} (see \cref{thm:tr-power-constant-samples}).

To conclude this subsection, it can be seen that our quantum algorithm for estimating $\tr\rbra{\rho^q}$ is naturally applicable to solving \TsallisQED{} and \TsallisQEA{}. 
Particularly for the precision in the regime $1/\poly\rbra{n} \leq \epsilon \leq 1$, the efficiently-computability of the uniform approximation polynomials in \cref{eq:good-enough-poly} ensures that the description of the quantum circuit of our algorithm can be computed by a classical deterministic Turing machine in $\poly\rbra{n}$ time, which is a significant step to show the $\BQP$-completeness of \TsallisQED{} and \TsallisQEA{} for $1 + \Omega\rbra{1} \leq q \leq 2$ and precision $1/\poly\rbra{n} \leq \epsilon \leq 1$.

\subsection{Proof techniques: Hardness via \QJTq{}-based reductions}

Before we proceed with the proof of \Cref{thm:TsallisQE-hardness-informal}, we start by defining the (white-box) quantum state testing problem with respect to the trace distance, which was first introduced in~\cite{Wat02}. 
Let $\rho_0$ and $\rho_1$ be $n$-qubit quantum states such that their purifications can be prepared by polynomial-size quantum circuits $Q_0$ and $Q_1$, respectively.
The \textsc{Quantum State Distinguishability Problem} (\QSD{}) is to decide whether the trace distance $\td(\rho_0,\rho_1)$ is at least $1-\epsilon(n)$ or at most $\epsilon(n)$. 
Furthermore, we need the other two restricted versions of \QSD{}, see \Cref{subsec:state-closeness-testing} for formal definitions:
\begin{itemize}[topsep=0.33em, itemsep=0.33em, parsep=0.33em] 
    \item \PureQSD{}: Both $\rho_0$ and $\rho_1$ are pure states. 
    \item \QSCMM{}: $\rho_1$ is fixed to be the $n$-qubit maximally mixed state.\footnote{Precisely speaking, the problem called \textsc{Quantum State Closeness to Maximally Mixed State} (\QSCMM{}) is to decide whether $\td\rbra{\rho,(I/2)^{\otimes n}}$ is at most $1/n$ or at least $1-1/n$, which is the complement of \QSD{} concerning the same states of interest.}
\end{itemize}

The proof of \Cref{thm:TsallisQE-hardness-informal}, particularly the hardness results under Karp reduction, utilizes reductions from the aforementioned variants of \QSD{} to \TsallisQED{} or \TsallisQEA{} for the respective ranges of $q$. Next, we will specify two main technical challenges related to the corresponding inequalities necessary for establishing \Cref{thm:TsallisQE-hardness-informal}: 
\begin{enumerate}[label={\upshape(\arabic*)}]
    \item For \ConstRankTsallisQED{} and \TsallisQED{}, the key ingredient of these reductions is the quantum $q$-Jensen--(Shannon--)Tsallis divergence (\QJTq{}, see \Cref{def:quantum-Jensen-Tsallis-divergence}), first introduced in~\cite{BH09}. 
    When $q=1$, the quantum Jensen--Shannon divergence (\QJS{}), defined in~\cite{MLP05} (see also \Cref{eq:QJS-decomposition-D} for its explicit form), can be expressed as the difference between the von Neumann entropy of two quantum states, due to the additivity of the von Neumann entropy:\footnote{Notably, combining \Cref{eq:QJS-to-QED} with the \QSZK{}-hardness of the quantum state testing problem with respect to \QJS{}, as established in~\cite[Lemma 4.15]{Liu23} (Lemma 4.13 in the last arXiv version), yields a simpler proof that \QED{} is \QSZK{}-hard, compared to the approaches in~\cite[Section 5.4]{BASTS10} and \cite[Corollary 4.3]{Liu23}.}
    \begin{equation}
        \label{eq:QJS-to-QED}
        2 \cdot \QJS(\rho_0,\rho_1) = \S\rbra*{ \rbra[\Big]{\frac{\rho_0+\rho_1}{2}} \otimes \rbra[\Big]{\frac{\rho_0+\rho_1}{2}} } - \S\rbra*{\rho_0\otimes \rho_1}.
    \end{equation}
    For $1 < q \leq 2$, the quantum $q$-Tsallis entropy satisfies only \textit{pseudo-additivity}~\cite{Raggio95}. Nevertheless, \QJTq{} can still be interpreted as a \textit{distance version} of the quantum $q$-Tsallis entropy difference for these values of $q$. Consequently, the reductions underlying the proof heavily rely on the inequalities between \QJTq{} and the trace distance. 
    However, such inequalities are currently known only in the case $q=1$~\cite{FvdG99,Holevo73,BH09}, presenting the first technical challenge. 
    \item For \TsallisQEA{}, the reduction essentially relies on the lower and upper bounds on the quantum $q$-Tsallis entropy of a quantum state $\rho$ in terms of the trace distance between the state and the maximally mixed state, when the trace distance is promised to be a fixed value. These bounds are also only known for the case of $q=1$~\cite{Vajda70,CCKV08,KLGN19}, leading to the second technical challenge. 
\end{enumerate}

For clarity, we summarize the correspondence between our reductions for establishing \Cref{thm:TsallisQE-hardness-informal} and the new inequalities in \Cref{table:hardness-proof-TsallisQED-TsallisQEA}, where the $q$-logarithm $\ln_q(x) \coloneqq \frac{1-x^{1-q}}{q-1}$.\footnote{As $q$ approaches $1$, the $q$-logarithm becomes the natural logarithm. For further details and references on $q$-logarithm, please refer to the beginning of \Cref{subsec:closeness-measures-states-and-distributions}.} 

\begin{table}[!ht]
\centering
\adjustbox{max width=\textwidth}{
\begin{tabular}{ccccc}
    \specialrule{0.1em}{2pt}{2pt}
    Problem & Regime of $q$ & Reduction from & New inequalities\\
    \specialrule{0.1em}{2pt}{2pt}
    \makecell{\textsc{ConstRank}\\ \TsallisQED{}\\ \footnotesize{\Cref{thm:TsallisQE-hardness-informal}\ref{thmitem:TsallisQE-BQPhard}}}
    & $1 \!\leq\! q \!\leq\! 2$
    & \makecell{\PureQSD{} is \BQP{}-hard\\ \footnotesize{adapted from~\cite{RASW23}}} 
    & \makecell{\small{$\Hq\rbra*{\frac{1}{2}} \!-\! \Hq\rbra*{\frac{1-\td}{2}} \leq \QJTq \leq \Hq\rbra*{\frac{1}{2}} \td^q$}\\ \footnotesize{\Cref{thm:QJTq-vs-td}}}\\
    \specialrule{0.05em}{2pt}{2pt}
    \makecell{\TsallisQED{}\\ \footnotesize{\Cref{thm:TsallisQE-hardness-informal}\ref{thmitem:TsallisQE-QSZKhard}}}
    & $1 \!\leq\! q \!\leq\! 1\!+\!\frac{1}{n-1}$
    & \makecell{\QSD{} is \QSZK{}-hard\\ \footnotesize{\cite{Wat02,Wat09}}} 
    & \makecell{\small{$\Hq\rbra*{\frac{1}{2}} \!-\! \Hq\rbra*{\frac{1-\td}{2}} \leq \QJTq$}\\ \footnotesize{\Cref{thm:QJTq-vs-td}}}\\
    \specialrule{0.05em}{2pt}{2pt}
    \makecell{\TsallisQEA{}\\ \footnotesize{\Cref{thm:TsallisQE-hardness-informal}\ref{thmitem:TsallisQE-QSZKhard}}}
    & $q=1\!+\!\frac{1}{n-1}$
    & \makecell{\QSCMM{} is \NIQSZK{}-hard \\ \footnotesize{\cite{Kobayashi03,BASTS10,CCKV08}}} 
    & \makecell{\small{$\rbra*{1\!-\!\td\!-\!\frac{1}{2^n}} \ln_q(2^n) \leq \Sq \leq \ln_q(2^n(1\!-\!\td))$}\\ \footnotesize{\Cref{lemma:inequality-uniformTV-TsallisEA}}}\\
    \specialrule{0.1em}{2pt}{2pt}
\end{tabular}
}
\caption{Reductions for \TsallisQED{} and \TsallisQEA{}, and the corresponding inequalities.}
\label{table:hardness-proof-TsallisQED-TsallisQEA}
\end{table}

Once we have established these new inequalities, together with our new bounds for the Tsallis binary entropy $\Hq(x) \leq \Hq\rbra*{\frac{1}{2}} \sqrt{4x(1-x)}$ (see \Cref{thm:Tsallis-binary-entropy-bounds}, where previously only the case of $q=1$ was known~\cite{Lin91,Topsoe01}), we can establish our three hardness results under Karp reduction in \Cref{thm:TsallisQE-hardness-informal} through relatively complicated and detailed analyses. The additional two hardness results for \ConstRankTsallisQEA{} and \TsallisQEA{} under Turing reduction in \Cref{thm:TsallisQE-hardness-informal} follow straightforwardly from a binary search for promise problems. 

\vspace{1em}
In the remainder of this subsection, we provide insights into proving the new inequalities in \Cref{table:hardness-proof-TsallisQED-TsallisQEA}. 
The first technical challenge involves establishing the inequalities between \QJTq{} and the trace distance. The main barrier is to provide the data-processing inequality $\QJTq(\Phi(\rho_0),\Phi(\rho_1)) \leq \QJTq(\rho_0,\rho_1)$ for $1 < q \leq 2$.\footnote{We generalize the approach in~\cite{BH09} for $q=1$. 
Using the data-processing inequality with a measurement channel, we can establish the lower bound via the measured version of \QJTq{} (see \Cref{eq:measured-f-divergences}) and the classical counterpart inequality for $\JTq$ in~\cite{BH09}. 
For the upper bound, we construct new states $\hat{\rho}_0$ and $\hat{\rho}_1$ with an ancillary qubit, making $\QJTq(\hat{\rho}_0,\hat{\rho}_1)$ related to the trace distance for $1 < q \leq 2$ (and coincide with the trace distance for $q=1$). 
Applying the data-processing inequality with the partial trace, we obtain $\QJTq(\rho_0,\rho_1) \leq \QJTq(\hat{\rho}_0,\hat{\rho}_1)$. 
} This implies that applying any quantum channel $\Phi$ on states $\rho_0$ and $\rho_1$ does not increase the divergence between them. For $q=1$, the quantum Jensen--Shannon divergence ($\QJS$) can be decomposed into a sum of quantum relative entropy $\D(\rho_0\|\rho_1)$: 
\begin{equation}
    \label{eq:QJS-decomposition-D}
    \QJS(\rho_0,\rho_1) \coloneqq \S\rbra*{\frac{\rho_0+\rho_1}{2}} - \frac{\S(\rho_0) + \S(\rho_1)}{2} = \frac{1}{2} \rbra*{\D\rbra*{\rho_0 \Big\| \frac{\rho_0+\rho_1}{2}} + \D\rbra*{\rho_1 \Big\| \frac{\rho_0+\rho_1}{2}} }. 
\end{equation}
Since the data-processing inequality (essentially, the joint convexity) for the quantum relative entropy was established decades ago~\cite{Lieb73,Uhlmann77}, and given the equality in \Cref{eq:QJS-decomposition-D}, it directly follows that the data-processing inequality also holds for \QJS{}. 
However, a similar decomposition does not apply to the quantum $q$-Tsallis entropy when $q \neq 1$. Fortunately, the joint convexity of $\QJTq$ for $1 \leq q \leq 2$, specifically,
\[ \QJTq\rbra*{(1-\lambda) \rho_0 + \lambda \rho'_0, (1-\lambda) \rho_1 + \lambda \rho'_1} \leq (1-\lambda)\QJTq\rbra*{\rho_0,\rho_1) + \lambda \QJTq(\rho'_0,\rho'_1}, \]
was established few years ago~\cite{CT14,Virosztek19}, where $0 < \lambda < 1$. As a consequence, once we establish the data-processing inequality for \QJTq{}, we can then generalize the inequalities between \QJS{} and the trace distance to \QJTq{} for $1 \leq q \leq 2$, using the same approach applied to \QJS{}. 

For the second technical challenge, specifically the bounds for $\Sq(\rho)$ when $\td\rbra*{\rho,(I/2)^{\otimes n}}=\gamma$ is fixed, it suffices to focus on the classical counterpart,\footnote{More specifically, let $p$ denote the distribution of the eigenvalues of $\rho$, and let $\nu$ be the uniform distribution over $2^n$ items. This task is exactly equivalent to proving the bounds for $\Hq(p)$ when $\TV\rbra*{p,\nu}=\gamma$ is fixed.} as the maximally mixed state commutes with any state $\rho$. 
The lower bound can be established by following the approach in~\cite{KLGN19} for $q=1$. 
On the other hand, the upper bound for $q=1$ can be derived using Vajda's inequality~\cite{Vajda70}, but similar results for $q \neq 1$ are unknown. However, by assuming an appropriate condition between $q$ and the fixed distance $\gamma$, we can deduce an upper bound analogous to the $q=1$ case. 

\subsection{Discussion and open problems}

    Our first main theorem (\Cref{thm:TsallisQE-containment-informal}) provides an efficiently computable lower bound for the von Neumann entropy $\S(\rho)$. This naturally raises the question: 
    \begin{enumerate}[label={\upshape(\roman*)},itemsep=0.3em,topsep=0.3em,parsep=0.3em]
        \item Is there an efficiently computable upper bound for $\S(\rho)$, perhaps based on some relaxed notion of the von Neumann entropy? 
    \end{enumerate}

    The quantum Tsallis entropy $\Sq(\rho)$ in the regime $1<q<2$ exhibits distinct behavior compared to both $\S(\rho)$ and $\S_2(\rho)=1-\tr(\rho^2)$, leading to another open problem: 
    \begin{enumerate}[label={\upshape(\roman*)},itemsep=0.3em,topsep=0.3em,parsep=0.3em]
        \setcounter{enumi}{1}
        \item Can we find further applications of estimating $\Sq(\rho)$ in the regime $1<q<2$? 
    \end{enumerate}
    Moreover, two open problems arise regarding quantitative bounds and (\textsf{NI})\QSZK{} containments: 
    \begin{enumerate}[label={\upshape(\roman*)},itemsep=0.3em,topsep=0.3em,parsep=0.3em]
        \setcounter{enumi}{2}
        \item \label{probitem:improving-bounds} Can the query and sample bounds in \Cref{table:query-and-sample-complexity} be improved, especially for $q \geq 1 + \Omega(1)$? 
        \item Can we establish that \TsallisQED{} (or \TsallisQEA{}) in the regime $1 < q < 1+\frac{1}{n-1}$, as specified in \Cref{thm:TsallisQE-hardness-informal}\ref{thmitem:TsallisQE-QSZKhard}, is also contained in \QSZK{} (or \NIQSZK{})? 
    \end{enumerate}

    Lastly, it is natural to consider generalizations of the von Neumann entropy tighter than $\Sq(\rho)$ for $q > 1$, particularly $\Sq(\rho)$ for $0 < q < 1$ and the quantum R\'enyi entropy $\S_{\alpha}^R(\rho) \coloneqq \frac{\ln\tr(\rho^{\alpha})}{1-\alpha}$: 
    \begin{enumerate}[label={\upshape(\roman*)},itemsep=0.3em,topsep=0.3em,parsep=0.3em]
        \setcounter{enumi}{4}
        \item What are the containment and hardness of estimating $\Sq(\rho)$ in the regime $0 <q<1$? 
        \item Since \textsc{R\'enyiQEA}$_{\alpha}$ for $1 < \alpha \leq \frac{1}{n-1}$ is \textit{intuitively} \QSZK{}-hard, as per \Cref{thm:TsallisQE-hardness-informal}\ref{thmitem:TsallisQE-QSZKhard}, can we obtain (rigorous) computational hardness results for estimating $\S^R_{\alpha}(\rho)$ with $\alpha > 0$? 
    \end{enumerate}
    
\subsection{Related works}

\paragraph{(Quantum) property testing for probability distributions.}
(Near-)optimal classical estimators are known for Shannon, R\'enyi, and Tsallis entropies \cite{JVHW15,JVHW17,WY16,AOST17}. 
Quantum testers for classical probability distributions were initiated in \cite{BHH11}.
Quantum algorithms for $\ell_1$ distance of probability distributions were investigated in \cite{BHH11,CFMdW10,Mon15,GL20,LWL24}.
Quantum estimators for the Shannon and R{\'e}nyi entropies were proposed in \cite{LW19,GL20,WZL24}.
Notably, matching query lower bounds for estimating the Shannon entropy and $\ell_1$ distance from the uniform distribution were shown in \cite{BKT20}.

\paragraph{Quantum property testing for quantum states.}

Quantum sample complexities for a series of problems have been studied in the literature.
For von Neumann entropy and R\'enyi entropy estimations, the dimension-dependence was studied in \cite{AISW20}, the dependence on the reciprocal of the minimum non-zero eigenvalue of the quantum state was studied in \cite{WZW23}, and the rank-dependence and time-efficiency were studied in \cite{WZ24b}.
Other problems include tomography \cite{HHJ+17,OW16}, spectrum testing \cite{OW21}, closeness testing or estimation with respect to fidelity and trace distance \cite{BOW19,GP22,WZ24,WZ24c,LWWZ25}.
Quantum inner product estimation is a basic task and is well-known to be solved by the SWAP test \cite{BCWdW01}. Recently, a distributed quantum algorithm for quantum inner product estimation was given in \cite{ALL22}, where they also provided a matching lower bound; this was later generalized to fidelity estimation between pure states with limited quantum communications \cite{AS24}.
As a special case of quantum inner product estimation, tight bounds for purity estimation with and without restricted quantum measurements were shown in \cite{CWLY23,GHYZ24,LGDC24}.

The quantum query complexities are also extensively studied.
For von Neumann entropy estimation, the dimension-dependence was studied in \cite{GL20}, the dependence on the reciprocal of the minimum non-zero eigenvalue of the quantum state was studied in \cite{CLW20}, the multiplicative error-dependence was studied in \cite{GHS21}, and the rank-dependence was studied in \cite{WGL+22}. 
In \cite{SLLJ24}, they presented a rank-dependent estimator for the $q$-Tsallis entropy with integer $q$ larger than the rank of quantum states.
Additionally, a dimension-dependent estimator for the quantities $\cbra{\tr(\rho^k)}_{k=1}^N$ was given in~\cite{WSP+24}.
For R\'enyi entropy estimation, the query complexity was first studied in \cite{SH21}, the rank-dependence was studied in \cite{WGL+22}, and was later improved in \cite{WZL24}.
Other problems include tomography \cite{vACGN23}, and the estimations of fidelity and trace distance \cite{WZC+23,WGL+22,GP22,WZ24,Wang24pureQSD,LWWZ25}.

In \cite{GH20}, the \textsc{Quantum Entropy Difference Problem} (with respect to von Neumann entropy) with shallow circuits was shown to have (conditional) hardness.
The computational complexity of the space-bounded versions of the \textsc{Quantum Entropy Difference Problem} and \textsc{Quantum State Distinguishability Problem} were studied in \cite{LGLW23}.

\paragraph{Quantum algorithms for estimating the Schatten $p$-norm.} Estimating the Schatten $p$-norm, $\tr(|A|^p)$, of an $O(\log n)$-local Hermitian matrix $A$ on $n$ qubits to within an additive error of $2^{n-p} \epsilon \|A\|^p$, where $\epsilon(n) \leq 1/\poly(n)$ and real $p(n) \leq \poly(n)$, was proven to be $\mathsf{DQC1}$-complete in \cite{CM18}. 
In \cite{LS20}, with a unitary block-encoding of a matrix $A$, a quantum algorithm was proposed for estimating the Schatten $p$-norm, $\rbra{\tr\rbra{\abs{A}^p}}^{1/p}$, to relative error $\epsilon$ for integer $p$. This algorithm requires a condition number $\kappa$ such that $A \geq I/\kappa$, particularly when $p$ is odd.

\paragraph{Recent developments.} More than six months after the release of our work, a related dichotomy -- analogous to our results in \Cref{thm:TsallisQE-containment-informal,thm:TsallisQE-hardness-informal} -- was established in~\cite{LW25Lalpha}, concerning the computational complexity of estimating the quantum $\ell_a$ distance defined via the Schatten norm. While our query complexity upper bound for estimating $\Sq(\rho)$ in the regime $q \geq 1+\Omega(1)$ (\Cref{thm:tr-power-constant-queries}) remains the state of the art, the sample complexity upper bound (\Cref{thm:tr-power-constant-samples}) was subsequently improved in~\cite{CW25}, which appeared in the same month. This improvement, based on a more direct approach using weak Schur sampling,  made progress on our open problem \ref{probitem:improving-bounds}. Specifically, \cite{CW25} obtained a bound of $\widetilde{\Theta}(1/\epsilon^2)$ for $q>2$, and bounds of $\widetilde{O}\rbra[\big]{1/\epsilon^{\frac{2}{q-1}}}$ and $\Omega\rbra{1/\epsilon^{\max\cbra[\big]{\frac{1}{q-1},2}}}$ for $1+\Omega(1) \leq q < 2$. Another follow-up work~\cite{CWYZ25}, also released that month, established a sample complexity lower bound of $\Omega(q/\epsilon^2)$ for estimating $\tr(\rho^q)$ when $q$ is an integer. 
One year after the release of our work, it was shown in~\cite{Wan25} that the quantum query complexity of estimating $\Sq(\rho)$ is $\widetilde{\Theta}\rbra{1/\sqrt{q}\varepsilon}$ for integer $q \geq 2$.
Later, it was established in~\cite{Liu26} that \ConstRankTsallisQEA{} is \BQP{}-hard under Karp reduction for all $q>0$. This result extends the \BQP{}-hard regime $1 \leq q \leq 2$ identified in our work (\Cref{thm:ConstRankTsallisQEA-BQPhard}), and further implies that both \TsallisQEA{} and \TsallisQED{} are \BQP{}-complete for all $q\geq 1+\Omega(1)$. 
More recently, one and a half years after the release of our work, rank-dependent sample complexity upper and lower bounds for estimating $\Sq(\rho)$ in the real-valued regime $0 < q < 1$, together with the corresponding \NIQSZK{}-hardness results, were established in~\cite{CLW26}.

In addition to the aforementioned progress in quantum state testing, our efficient uniform approximation of the positive constant power function (\Cref{lemma:computable-fullrange-bestPolyApprox-positivePower}) has been recently used to learn the $q$-th singular value moments of unknown quantum channels for real $q>2$ in~\cite{NRTM25}.

\section{Preliminaries}

We assume that the reader is familiar with quantum computation and the theory of quantum information. For an introduction, the textbook \cite{NC10} provides a good starting point.

We adopt the following conventions throughout the paper: (1) we denote $\sbra{n} \coloneqq \cbra{1, 2, \dots, n}$; (2) we use both notations, $\log(x)$ and $\ln(x)$, to represent the natural logarithm for any $x \in \bbR^{+}$; (3) the notation $\widetilde{O}(f)$ is defined as $O(f \operatorname{polylog}(f))$; (4) we utilize the notation $\ket{\bar{0}}$ to represent $\ket{0}^{\otimes a}$ with $a>1$; and (5) we use $\Abs{A}$ to denote the operator norm (equivalently, the Schatten $\infty$-norm) of a matrix $A$. 

\paragraph{Notions on linear maps and quantum channels.}
We recommend~\cite[Section 2.3]{AS17} as an introduction on superoperators and quantum channels. 
Let $\calH_1$ and $\calH_2$ be finite-dimensional Hilbert spaces with $\dim(\calH_i) = N_i = 2^{n_i}$ for $i\in\{1,2\}$. Let $\rmL(\calH_1, \calH_2)$ denote linear maps from $\calH_1$ to $\calH_2$, and specifically, let $\rmL(\calH)$ denote linear maps from $\calH$ to $\calH$. 
A map $\Phi \colon \rmL(\calH_1) \rightarrow \rmL(\calH_2)$ is called \textit{self-adjointness-preserving} if $\Phi(X^{\dagger}) = \rbra*{\Phi(X)}^{\dagger}$ for any $X \in  \rmL(\calH_1)$.
We further say that a self-adjointness-preserving map $\Phi \colon \rmL(\calH_1) \rightarrow \rmL(\calH_2)$ is a \textit{quantum channel} if $\Phi$ is a completely positive trace-preserving map. 
Here, a map $\Phi$ is \textit{trace-preserving} if $\tr(\Phi(X)) = \tr(X)$ for any $X \in  \rmL(\calH_1)$.
Let $\{\ket{v_i}\}_{i\in[N_1]}$ be an orthonormal basis of $\calH_1$, and a map $\Phi$ is \textit{completely positive} if $\Phi \otimes I_n$ is positive for any $n \in \bbN$, where $I_n$ is the identity map of the dimension $n$. 

Let $\D(\calH)$ be the set of all density operators, which are semi-positive and trace-one matrices on $\calH$. 
Let the trace norm of a linear map $X$ be $\Abs{X}_1 \coloneqq \tr\big(\sqrt{X^{\dagger} X}\big)$. 
For any quantum channels $\calE$ and $\calF$ that act on $\D(\calH)$, the \textit{diamond norm distance} between them is defined as
\[ \Abs{\calE-\calF}_{\diamond} \coloneqq \sup_{\rho \in \D(\calH\otimes \calH')} \Abs*{\rbra*{\calE\otimes \calI_{\calH'}}(\rho) - \rbra*{\calF \otimes \calI_{\calH'}}(\rho)}_1. \]

\subsection{Closeness measures for distributions and quantum states}
\label{subsec:closeness-measures-states-and-distributions}

Since both classical and quantum Tsallis entropy are central to this paper, we introduce the \textit{$q$-logarithm} function $\ln_q \colon \bbR^+ \rightarrow \bbR$ for any real $q\neq 1$:
\[ \forall x\in \bbR^+, \quad \ln_q(x) \coloneqq \frac{1-x^{1-q}}{q-1}. \]

The $q$-logarithm is a generalization of the natural logarithm, as it is straightforward to verify that $\lim_{q \rightarrow 1} \ln_q(x) = \ln(x)$ for any $x\in \bbR^+$. However, the $q$-logarithm exhibits different behavior when $q \neq 1$, for instance, $\ln_q(xy) = \ln_q(x) + \ln_q(y) + (1-q) \ln_q(x)\ln_q(y).$ For additional properties of the $q$-logarithm, see the references~\cite[Appendix]{Tsallis01} and~\cite{Yamano02}.

\vspace{1em}
In the rest of this subsection, we provide useful closeness measures for probability distributions in \Cref{subsubsec:closeness-measures-distributions} and for quantum states in \Cref{subsubsec:closeness-measures-states}, which are later used to prove new inequalities and bounds in \Cref{sec:properties-QJTq} and to establish the reductions and the corresponding computational hardness results in \Cref{sec:hardness-via-QJTq-reductions}.
For convenience, we use a general convention $D_1 \leq D_2$ to denote an inequality between two distances or divergences, whether classical or quantum. In particular, this notation, as seen in \textit{the titles of technical lemmas} (e.g., \Cref{lemma:TV-leq-JSq}), indicates that $D_1$ is bounded above by a function $f$ of $D_2$, i.e., $D_1 \leq f(D_2)$; or that $D_2$ is bounded below by a function $g$ of $D_1$, i.e., $g(D_1) \leq D_2$. 

\subsubsection{Closeness measures for classical probability distributions}
\label{subsubsec:closeness-measures-distributions}
We begin by defining the total variation distance:
\begin{definition}[Total variation distance]
	\label{def:statDist}
    Let $p_0$ and $p_1$ be two probability distributions over $[N]$. The total variation distance between two $p_0$ and $p_1$ is defined by
	\[\TV(p_0,p_1)  \coloneqq  \frac{1}{2}\|p_0-p_1\|_1 = \frac{1}{2}\sum_{x\in[N]} |p_0(x)-p_1(x)|.\]
\end{definition}

Then we define the Tsallis entropy and provide useful properties of the Tsallis entropy.  
\begin{definition}[$q$-Tsallis entropy and Shannon entropy]
    Let $p$ be a probability distribution over $[N]$. The $q$-Tsallis entropy of $p$ is defined by \[\Hq(p) \coloneqq \frac{1-\sum_{x\in[N]} p(x)^q}{q-1} = -\sum_{x\in[N]} p(x)^q \ln_q\rbra*{p(x)}.\]
    The Shannon entropy is the limiting case of the $q$-Tsallis entropy as $q \rightarrow 1$\emph{:}
    \[ \H_1(p) \coloneqq \lim_{q \rightarrow 1} \Hq(p) \quad \text{and} \quad \lim_{q \rightarrow 1} \Hq(p) = \H(p) \coloneqq - \sum_{x \in \sbra{N}} p\rbra{x} \ln \rbra{p\rbra{x}}. \]
    For $N=2$, we slightly abuse the notation by writing the $q$-Tsallis binary entropy and the \emph{(}Shannon\emph{)} binary entropy as $\Hq(p_0) = \Hq(1-p_0) = \Hq(p)$ and $\H(p_0) = \H(1-p_0) = \H(p)$, respectively. 
\end{definition}

It is noteworthy that the properties in \Cref{lemma:Tsallis-entropy-properties} were also provided in~\cite{Tsallis88} without proofs. In addition, by considering the eigenvalues of any quantum state, \Cref{lemma:Tsallis-entropy-properties} straightforwardly extends to quantum $q$-Tsallis entropy (see \Cref{def:quantum-Tsallis-entropy}). 
\begin{lemma}[Basic properties of Tsallis entropy, partially adapted from~\cite{Daroczy70}]
    \label{lemma:Tsallis-entropy-properties}
    Let $p$ and $p'$ be two probability distributions over $[N]$ with $N\geq 2$, and let $\nu$ be the uniform distribution over $[N]$. We have the following properties of the Tsallis entropy $\Hq(p)$ with $q > 0$: 
    \begin{itemize}[topsep=0.33em, itemsep=0.33em, parsep=0.33em] 
        \item \emph{\textbf{Concavity}:} For any $\lambda \in [0,1]$, $\Hq((1-\lambda) p + \lambda p') \geq (1-\lambda) \Hq(p) + \lambda \Hq(p')$. Equivalently, $F(q;x) \coloneq \frac{x-x^q}{q-1}$ is concave in $x\in[0,1]$ for any fixed $q>0$, and $\Hq(p) = \sum_{i\in[N]} F(q;p(i))$.
        \item \emph{\textbf{Extremes}:} $0 \leq \Hq(p) \leq \Hq(\nu) = \frac{1-n^{1-q}}{q-1}$. Specifically, $\Hq(p) = \Hq(\nu)$ occurs when $p=\nu$, and $\Hq(p) = 0$ occurs when $p(i)=\begin{cases} 1,~ &i=k\\ 0,~ &i\neq k \end{cases}$ for any $k \in [N]$. 
        \item \emph{\textbf{Monotonicity}:} For any $q$ and $q'$ satisfying $0 < q \leq q'$, $\H_q(p) \geq \H_{q'}(p)$. 
    \end{itemize}
\end{lemma}

\begin{proof}
    For the first item, by inspecting the proof of~\cite[Theorem 6]{Daroczy70}, we know that $\frac{x-x^q}{1-2^{1-q}} \cdot \frac{1-2^{1-q}}{q-1} = F(q;x)$ is concave in $x\in[0,1]$ for any fixed $q\neq 1$. It is easy to verify that $\Hq(p) = \sum_{i\in[N]}F(q;p(i))$, we have that $\Hq(p)$ is concave. 
        
    For the second item, note that $\frac{q-1}{1-2^{1-q}} \geq 0$ for $q\neq 1$ and $\lim_{q\rightarrow 1}  \frac{q-1}{1-2^{1-q}} = \frac{1}{\ln{2}}$. Hence, by~\cite[Theorem 6]{Daroczy70}, we deduce $0 \leq \Hq(p) \leq \Hq(\nu)$.
    Moreover, because $F(q;x)$ is non-negative and $F(q;x)=0$ occurs when $x=1$, we conclude that $\Hq(p)=0$ occurs when $p$ satisfies the desired condition.

    For the third item, since $\lim_{q\rightarrow 1}\Hq(x) = \H(x)$, it is enough to show that $\frac{\partial}{\partial q} \Hq(x) \leq 0$ for any $q \neq 1$ and $x \in [0,1]$. Given that $\Hq(p) = \sum_{i\in[n]}F(q;p(i))$, it remains to prove that $\frac{\partial}{\partial q} F(q;x) \leq 0$, specifically:
    \begin{equation}
        \label{eq:Tsallis-binary-entropy-partialQ}
        \frac{\partial}{\partial q} F(q;x) = -\frac{x-x^q}{(q-1)^2}-\frac{x^q \log (x)}{q-1} \leq 0 \Leftrightarrow G(q;x) \coloneqq x^q -x - (q-1) x^q \log(x) \leq 0. 
    \end{equation}
    
    A direct calculation implies that $\frac{\partial}{\partial q} G(q;x) = -(q-1)x^q\ln^2(x)$ for any $x \in [0,1]$. This inequality shows that for any fixed $x\in[0,1]$, $G(q;x)$ is monotonically increasing for $0<q\leq 1$ and monotonically decreasing for $q>1$. Hence, by noticing $\max_{q\geq 0} G(q;x) \leq G(1;x) = 0$ for any $x \in [0,1]$, we establish \Cref{eq:Tsallis-binary-entropy-partialQ} and the monotonicity. 
\end{proof}

Next, we define a variant of the Jensen--Shannon divergence based on the $q$-Tsallis entropy:
\begin{definition}[$q$-Jensen--(Shannon--)Tsallis divergence, adapted from~\cite{BR82}]
    Let $p_0$ and $p_1$ be two probability distributions over $[N]$. The $q$-Jensen--(Shannon--)Tsallis divergence between $p_0$ and $p_1$ is defined as 
    \[ \JTq(p_0, p_1) \coloneqq \begin{cases}
    \Hq \left( \frac{p_0+p_1}{2} \right) - \frac{1}{2} \left( \Hq(p_0) + \Hq(p_1) \right), &q\neq 1\\
    \H \left( \frac{p_0+p_1}{2} \right) - \frac{1}{2} \left( \H(p_0) + \H(p_1) \right), &q=1
    \end{cases}. \]
    Specifically, the Jensen--Shannon divergence $\JS(p_0,p_1) = \JT{1}(p_0,p_1)$.
\end{definition}

Lastly, we provide a useful bound on the divergence $\JTq$, which generalizes the bound $\H\big(\frac{1}{2}\big) - \H\big(\frac{1}{2} -\frac{\TV(p_0,p_1)}{2}\big) \leq \JS(p_0,p_1)$ in~\cite[Theorem 5]{Top00} for the case of $q=1$: 
\begin{lemma}[$\TV \leq \JTq$, adapted from~{\cite[Theorem 9]{BH09}}]
    \label{lemma:TV-leq-JSq}
    Let $p_0$ and $p_1$ be two probability distributions over $[N]$. For any $1 \leq q \leq 2$, we have the following inequality\emph{:}\footnote{It is evident that $\Hq\big( \frac{1-x}{2} \big) = \Hq\big( \frac{1+x}{2} \big)$ for any $x \in [0,1]$. Moreover, the proof of the lower bound in~\cite[Theorem 9]{BH09} uses the notation $V(p_0,p_1) \coloneqq \sum_{i=1}^n |p_0(i)-p_1(i)| = \|p_0-p_1\|_1 = 2\TV(p_0,p_1)$ defined in~\cite{Top00}, where $p_0$ and $p_1$ are probability distributions over $[N]$. }
    \[ \Hq\left(\frac{1}{2}\right) - \Hq\left( \frac{1}{2} - \frac{\TV(p_0,p_1)}{2} \right) \leq \JTq(p_0,p_1). \]
\end{lemma}

It is important to note, the joint convexity of $\JTq$~\cite[Corollary 1]{BR82} plays a key role in proving \Cref{lemma:TV-leq-JSq}. And additionally, for $N \geq 3$, the joint convexity of $\JTq$ holds if and only if $q \in [1,2]$, as stated in~{\cite[Corollary 2]{BR82}}.

\subsubsection{Closeness measures for quantum states}
\label{subsubsec:closeness-measures-states}

We start by defining the trace distance and providing useful properties of this distance: 
\begin{definition}[Trace distance]
    The trace distance between two quantum states $\rho_0$ and $\rho_1$ is  
     \[\td(\rho_0,\rho_1) \coloneqq \frac{1}{2}\tr\rbra{|\rho_0-\rho_1|}= \frac{1}{2} \tr\rbra*{\rbra*{(\rho_0-\rho_1)^\dagger(\rho_0-\rho_1)}^{1/2}}.\]
\end{definition}
Importantly, the trace distance is a measured version of the total variation distance~\cite[Theorem 9.1]{NC10}. In particular, for any classical $f$-divergence $\mathrm{D}_f(\cdot,\cdot)$, let $\rho_0$ and $\rho_1$ be two $N$-dimensional quantum states that are mixed in general, we can define the \textit{measured quantum $f$-divergence} by considering the probability distributions induced by the POVM $\calM$:
\begin{equation}
\label{eq:measured-f-divergences}
\mathrm{D}^{\rm meas}_f(\rho_0,\rho_1) = \sup_{\mathrm{POVM}~\calM} \mathrm{D}_f\rbra*{ p_0^{(\calM)}, p_1^{(\calM)} } \text{ where } p_b^{(\calM)} \coloneqq \left(\tr(\rho_b M_1), \cdots, \tr(\rho_b M_N)\right).
\end{equation}

Moreover, the trace distance is a distance metric (e.g.,~\cite[Lemma 9.1.8]{Wil13}). In addition, as indicated in~\cite[Equation (9.134)]{Wil13}, for pure states $\ket{\psi_0}$ and $\ket{\psi_1}$, we have 
\begin{equation}
    \label{eq:pure-traceDist-vs-fidelity}
    \td(\ketbra{\psi_0}{\psi_0}, \ketbra{\psi_1}{\psi_1}) = \sqrt{1-|\innerprod{\psi_0}{\psi_1}|^2}.
\end{equation}

Additionally, the trace distance characterizes the maximum success probability of discriminating quantum states in quantum hypothesis testing, as explained in~\cite[Section 9.1.4]{Wil13}, which is later used in \Cref{subsec:sample-lower-bound} to prove a sample complexity lower bound: 

\begin{lemma} [Helstrom-Holevo bound~\cite{Hel67,Hol73}] 
    \label{lemma:Holevo-Helstrom-bound}
    Suppose that a mixed quantum state $\rho$ is given such that either $\rho = \rho_0$ or $\rho = \rho_1$ with equal probability. 
    Then, any POVM distinguishes the two cases with success probability upper bounded by 
    \[
    p_{\textup{succ}} \leq \frac{1}{2} + \frac{1}{2} \td\rbra{\rho_0, \rho_1}.
    \]
\end{lemma}

\vspace{1em}
Next, we define the quantum $q$-Tsallis entropy, generalizing the von Neumann entropy: 
\begin{definition}[Quantum $q$-Tsallis entropy and von Neumann entropy]
    \label{def:quantum-Tsallis-entropy}
    Let $\rho$ be a \emph{(}mixed\emph{)} quantum state. The quantum $q$-Tsallis entropy of $\rho$ is defined by \[\Sq(\rho) \coloneqq \frac{1-\tr(\rho^{q})}{q-1} = - \tr\rbra*{\rho^q \ln_q\rbra*{\rho}}.\]
    Furthermore, as $q\rightarrow 1$, the quantum $q$-Tsallis entropy coincides with the von Neumann entropy: 
    \[ \S_{1}(\rho) \coloneqq \lim_{q\rightarrow 1} \Sq(\rho) \quad \text{and} \quad \lim_{q\rightarrow 1} \Sq(\rho) = \mathrm{S}\rbra{\rho} \coloneqq -\tr\rbra*{\rho \ln \rbra{\rho}}. \]
\end{definition}

% The following lemma also works for 0<q<1
\begin{lemma}[Pseudo-additivity of $\Sq$, adapted from~{\cite[Lemma 3]{Raggio95}}]
    \label{lemma:Tsallis-entropy-pseudoAdditivity}
    For any quantum states $\rho_0$ and $\rho_1$, and any $q \geq 1$, we have\emph{:}
    \[ \Sq(\rho_0 \otimes \rho_1) = \Sq(\rho_0) + \Sq(\rho_1) - (q-1) \Sq(\rho_0) \Sq(\rho_1).\]
    Specifically, the equality $\Sq(\rho_0 \otimes \rho_1) = \Sq(\rho_0) + \Sq(\rho_1)$ holds if and only if \emph{(a)} $q=1$, or \emph{(b)} for $q>1$, either of the states $\rho_0$ or $\rho_1$ is pure.      
\end{lemma}

Now we define a variant of the quantum Jensen--Shannon divergence~\cite{MLP05} based on the quantum $q$-Tsallis entropy, as stated in~\Cref{def:quantum-Jensen-Tsallis-divergence}. Notably, the study of quantum analogs of the Jensen--Shannon divergence could date back to the Holevo bound~\cite{Holevo73}.\footnote{The quantum Jensen--Shannon divergence ($\QJS$) is a special case of the Holevo $\chi$ quantity (the right-hand side of the Holevo bound~\cite{Holevo73}). Following the notations in \cite[Theorem 12.1]{NC10}, $\QJS$ coincides with the Holevo $\chi$ quantity on size-$2$ ensembles with a uniform distribution.}
\begin{definition}[Quantum $q$-Jensen--(Shannon--)Tsallis Divergence, adapted from~\cite{BH09}]
    \label{def:quantum-Jensen-Tsallis-divergence}
    Let $\rho_0$ and $\rho_1$ be two quantum states that are mixed in general. The quantum \mbox{$q$-Jensen-\emph{(}Shannon-\emph{)}Tsallis} divergence between $\rho_0$ and $\rho_1$ is defined by \[ \QJTq(\rho_0, \rho_1) \coloneqq \begin{cases}
        \Sq \left( \frac{\rho_0+\rho_1}{2} \right) - \frac{1}{2} \left( \Sq(\rho_0) + \Sq(\rho_1) \right), &q\neq 1\\
        \S \left( \frac{\rho_0+\rho_1}{2} \right) - \frac{1}{2} \left( \S(\rho_0) + \S(\rho_1) \right), &q=1
    \end{cases}.\]

    \noindent Specifically, for pure states $\ketbra{\psi_0}{\psi_0}$ and $\ketbra{\psi_1}{\psi_1}$, $\QJTq(\ketbra{\psi_0}{\psi_0},\ketbra{\psi_1}{\psi_1}) = \Sq\big( \frac{\ketbra{\psi_0}{\psi_0} + \ketbra{\psi_1}{\psi_1}}{2} \big)$.
    Furthermore, the quantum Jensen--Shannon divergence $\QJS(\rho_0,\rho_1) = \QJT{1}(\rho_0,\rho_1)$. 
\end{definition}

It is worth noting that the square root of $\QJTq$ is a distance metric when $0 \leq q \leq 2$~\cite{Sra21} (also see~\cite{Virosztek21} for the $q=1$ case). 
Moreover, whereas $\QJS$ can be expressed as a symmetrized version of the quantum relative entropy $\D(\rho_0\|\rho_1) \coloneqq \tr(\rho_0 (\ln(\rho_0) - \ln(\rho_1)))$ by
\begin{equation}
    \label{eq:QJS-as-symmetrized-relative-entropy}
    \QJS(\rho_0,\rho_1) = \frac{1}{2} \left( \D\left(\rho_0 \Big\| \frac{\rho_0+\rho_1}{2} \right) + \D\left(\rho_1 \Big\| \frac{\rho_0+\rho_1}{2} \right) \right),
\end{equation}
a similar equality does not hold for $\QJTq$ with respect to the quantum Tsallis relative entropy $\Dq(\rho_0 \| \rho_1) \coloneqq \frac{1-\tr(\rho_0^q \rho_1^{1-q})}{1-q}$ (see, e.g.,~\cite{FYK04}).\footnote{A symmetrized version of the quantum Tsallis relative entropy will lead to a different quantity, see~\cite{JMA21}.} 
In addition to $\QJS$, the work of~\cite{FvdG99} studied a measured variant of the Jensen--Shannon divergence $\measQJS(\rho_0,\rho_1)$ in terms of~\Cref{eq:measured-f-divergences}, namely the \textit{quantum Shannon distinguishability}. 

\vspace{1em}
Lastly, we provide more useful properties of $\QJTq$. 
By combining \cite[Theorem 1.5]{FYK07} and \cite[Remark V.3]{Furuichi05}, we can immediately
derive \Cref{lemma:QJSq-unitary-invariance,lemma:Tsallis-joint-entropy-theorem}. 
In particular, the equality in \Cref{lemma:Tsallis-joint-entropy-theorem} holds even in a stronger form: $\Sq\big( \sum_{i \in [k]} \mu_i \rho_i \big) = \Hq(\mu) + \sum_{i \in [k]} \mu_i^q \Sq(\rho_i)$ for orthogonal quantum states $\rho_1,\cdots,\rho_k$. 
Additionally, it is noteworthy that \Cref{lemma:Tsallis-joint-entropy-theorem} admits a simple proof in~\cite[Lemma 1]{Kim16}. 

\begin{lemma}[Unitary invariance of $\QJTq$, adapted from~{\cite[Theorem 1.5]{FYK07}}]
    \label{lemma:QJSq-unitary-invariance}
    For any quantum states $\rho_0$ and $\rho_1$, and any unitary transformation $U$ acting on $\rho_0$ or $\rho_1$, it holds that\emph{:}
        \[ \QJTq(U^{\dagger} \rho_0 U, U^{\dagger} \rho_1 U) = \QJTq(\rho_0,\rho_1). \]
\end{lemma}

\begin{lemma}[Joint $q$-Tsallis entropy theorem, adapted from~{\cite[Theorem 1.5]{FYK07}}]
    \label{lemma:Tsallis-joint-entropy-theorem}
    Let $k$ be an integer, and let $\{\rho_i\}_{i\in[k]}$ be a set of \emph{(}mixed\emph{)} quantum states. Let $k$-tuple $\mu \coloneq (\mu_1, \cdots, \mu_k)$ be a probability distribution.
    Then, for any $q \geq 0$, we have the following\emph{:}  
    \[ \Sq \left( \sum_{i\in[k]} \mu_i \ketbra{i}{i} \otimes \rho_i \right) = \Hq(\mu) + \sum_{i\in[k]} \mu_i^q \Sq(\rho_i). \]
\end{lemma}

Following the discussion in~\cite[Section 3]{Rastegin11}, Fannes' inequality for the (quantum) $q$-Tsallis entropy, where $0 \leq q \leq 2$, was established in~\cite[Theorem 2.4]{FYK07}. Notably, for $\Sq(\rho)$ with $q > 1$, a sharper Fannes-type inequality was provided in~\cite[Theorem 2]{Zhang07}: 

\begin{lemma}[Fannes' inequality for the $q$-Tsallis entropy, adapted from Theorem 2 and Corollary 2 in~\cite{Zhang07}]
    \label{lemma:Fannes-inequality-Sq}
    For any quantum states $\rho_0$ and $\rho_1$ of dimension $N$, we have: 
    \[ \forall q>1,~|\Sq(\rho_0) - \Sq(\rho_1)| \leq \td(\rho_0,\rho_1)^q \cdot \ln_q(N-1)+ \Hq(\td(\rho_0,\rho_1)). \]
    Moreover, for the case of $q=1$ \emph{(}von Neumann entropy\emph{)}, we have: 
    \[ |\S(\rho_0) - \S(\rho_1)| \leq \td(\rho_0,\rho_1) \cdot \ln(N-1) + \H(\td(\rho_0,\rho_1)). \]
\end{lemma}

\subsection{Closeness testing of quantum states via state-preparation circuits}
\label{subsec:state-closeness-testing}

We begin by defining the closeness testing of quantum states with respect to the trace distance, denoted as \QSD{}$[\alpha,\beta]$,\footnote{While \Cref{def:QSD} aligns with the classical counterpart of \QSD{} defined in~\cite[Section 2.2]{SV97}, it is slightly less general than the definition in~\cite[Section 3.3]{Wat02}. Specifically, \Cref{def:QSD} assumes that the input length $m$ and the output length $n$ are \textit{polynomially equivalent}, whereas~\cite[Section 3.3]{Wat02} allows for cases where the output length (e.g., a single qubit) is \textit{much smaller} than the input length.} along with a variant of this promise problem, as described in \Cref{def:QSD}. 
In particularly, we say that $\calP = (\calP_{\yes}, \calP_{\no})$ is a \textit{promise} problem, if it satisfies the conditions $\calP_{\yes} \cap \calP_{\no} =\emptyset$ and $\calP_{\yes} \cup \calP_{\no} \subseteq \binset^*$.

\begin{definition}[Quantum State Distinguishability, \QSD{}, adapted from~{\cite[Section 3.3]{Wat02}}]
	\label{def:QSD}
    Let $Q_0$ and $Q_1$ be quantum circuits acting on $m$ qubits \emph{(}``input length''\emph{)} and having $n$ specified output qubits \emph{(}``output length''\emph{)}, where $m(n)$ is a polynomial function of $n$. Let $\rho_i$ denote the quantum state obtained by running $Q_i$ on state $\ket{0}^{\otimes m}$ and tracing out the non-output qubits. Let $\alpha(n)$ and $\beta(n)$ be efficiently computable functions. Decide whether\emph{:} 
	\begin{itemize}[topsep=0.33em, itemsep=0.33em, parsep=0.33em]
		\item \emph{Yes:} The pair of quantum circuits $(Q_0,Q_1)$ such that $\td(\rho_0,\rho_1) \geq \alpha(n)$; 
		\item \emph{No:} The pair of quantum circuits $(Q_0,Q_1)$ such that $\td(\rho_0,\rho_1) \leq \beta(n)$.
	\end{itemize}

    \noindent Furthermore, we denote the restricted version, where $\rho_0$ and $\rho_1$ are pure states, as \PureQSD{}. 
\end{definition}

In addition to \QSD{}, we can similarly define the closeness testing of a quantum state to the maximally mixed state (with respect to the trace distance), denoted as $\QSCMM[\beta,\alpha]$:

\begin{definition}[Quantum State Closeness to Maximally Mixed State, \QSCMM{}, adapt from~{\cite[Section 3]{Kobayashi03}}]
    Let $Q$ be a quantum circuit acting on $m$ qubits and having $n$ specified output qubits, where $m(n)$ is a polynomial function of $n$. Let $\rho$ denote the quantum state obtained by running $Q$ on state $\ket{0}^{\otimes m}$ and tracing out the non-output qubits. Let $\alpha(n)$ and $\beta(n)$ be efficiently computable functions. Decide whether\emph{:} 
	\begin{itemize}[topsep=0.33em, itemsep=0.33em, parsep=0.33em]
		\item \emph{Yes:} The quantum circuit $Q$ such that $\td\rbra*{\rho,(I/2)^{\otimes n}} \leq \beta(n)$; 
		\item \emph{No:} The quantum circuit $Q$ such that $\td\rbra*{\rho,(I/2)^{\otimes n}} \geq \alpha(n)$.
	\end{itemize}
\end{definition}

In this subsection, we first introduce the input models in \Cref{subsubsec:input-models-and-reductions} and summarize useful known hardness results in \Cref{subsec:hardness-QSD-and-QSCMM}. These results are later combined with our new reductions in \Cref{subsec:pure-state-reduction,subsec:mixed-state-reductions} to derive the new hardness results presented in \Cref{subsec:hardness-results}. Next, we list several query and sample lower bounds in \Cref{subsec:known-query-and-sample-lower-bounds}, which are later used to establish our new query and sample lower bounds in \Cref{subsec:query-lower-bound,subsec:sample-lower-bound} (possibly in combination with the reductions in \Cref{subsec:pure-state-reduction,subsec:mixed-state-reductions}).

\subsubsection{Input models and the concept of reductions}
\label{subsubsec:input-models-and-reductions}

In this work, we consider the \textit{purified quantum access input model}, as defined in~\cite{Wat02}, in both white-box and black-box scenarios: 
\begin{itemize}[leftmargin=2em]
    \item \textbf{White-box input model}: The input of the problem \QSD{} consists of descriptions of polynomial-size quantum circuits $Q_0$ and $Q_1$. Specifically, for $b\in\binset$, the description of $Q_b$ includes a sequence of polynomially many $1$- and $2$-qubit gates.
    \item \textbf{Black-box input model}: In this model, instead of providing the descriptions of the quantum circuits $Q_0$ and $Q_1$, only query access to $Q_b$ is allowed, denoted as $O_b$ for $b\in\binset$. For convenience, we also allow query access to $Q_b^{\dagger}$ and controlled-$Q_b$, denoted by $O_b^{\dagger}$ and controlled-$O_b$, respectively. 
\end{itemize}

% Adapted from Definition 2.9 and 2.11 in [Goldreich08]
Next, the concept of \textit{reductions} between promise problems is used to address computational hardness in the context of the while-box input model, particularly in relation to complexity classes. Following the definitions in~\cite[Section 2.2.1]{Goldreich08}, we introduce two types of reductions from a promise problem $\calP = (\calP_{\yes},\calP_{\no})$ to another promise problem $\calP' = (\calP'_{\yes},\calP'_{\no})$:
\begin{itemize}[leftmargin=2em, topsep=0.33em, itemsep=0.33em, parsep=0.33em]
    \item \textbf{Karp reduction}. A deterministic polynomial-time computable function $f$ is called a \textit{Karp reduction} from a promise problem $\calP$ to another promise problem $\calP'$ if, for every $x$, the following holds: $x\in\calP_{\yes}$ if and only if $f(x) \in \calP'_{\yes}$, and $x\in\calP_{\no}$ if and only if $f(x) \in \calP'_{\no}$. 
    \item \textbf{Turing reduction}. A promise problem $\calP$ is \textit{Turing-reducible} to a promise problem $\calP'$ if there exists a deterministic polynomial-time oracle machine $\calA$ such that, for every function $f$ that solves $\calP'$ it holds that $\calA^f$ solves $\calP$. Here, $\calA^f(x)$ denotes the output of machine $\calA$ on input $x$ when given oracle access to $f$. 
\end{itemize}
It is noteworthy that Karp reduction is a special case of Turing reduction.

\subsubsection{Computational hardness of \QSD{} and \QSCMM{}}
\label{subsec:hardness-QSD-and-QSCMM}

Note that the polarization lemma for the total variation distance~\cite{SV97} and the trace distance~\cite{Wat02,Wat09} have the same inequalities. Consequently, using the parameters chosen in~\cite[Theorem 3.14]{BDRV19}, we obtain the \QSZK{} hardness of $\QSD$:
\begin{lemma}[\QSD{} is \QSZK{}-hard]
    \label{lemma:QSD-is-QSZKhard}
    Let $\alpha(n)$ and $\beta(n)$ be efficiently computable functions satisfying $\alpha^2(n)-\beta(n) \geq 1/O(\log{n})$. For any constant $\tau \in (0,1/2)$, $\QSD[\alpha,\beta]$ is \QSZK{}-hard under Karp reduction when $\alpha(n) \leq 1-2^{-n^{\tau}}$ and $\beta(n) \geq 2^{-n^{\tau}}$ for every $n\in\bbN$.
\end{lemma}

Following the construction in~\cite[Theorem 12]{RASW23} (see also~\cite[Lemma 4.17]{LGLW23} and \cite[Theorem IV.1]{WZ24}), we can establish that \PureQSD{} is \BQP-hard under Karp reduction:

\begin{lemma}[\PureQSD{} is \BQP{}-hard]
    \label{lemma:PureQSD-is-BQPhard}
    Let $\alpha(n)$ and $\beta(n)$ be efficiently computable functions such that $\alpha(n)-\beta(n) \geq 1/\poly(n)$. For any polynomial $l(n)$, let  $n' \coloneqq n+1$, $\PureQSD[\alpha(n'),\beta(n')]$ is \BQP{}-hard when $\alpha(n') \leq \sqrt{1-2^{-2l(n'-1)}}$ and $\beta(n') \geq 2^{\frac{-l(n'-1)+1}{2}}$ for every integer $n' \geq 2$.

    \noindent Specifically, by choosing $l(n'-1)=n'$, it holds that\emph{:}  For every integer $n' \geq 2$, 
    \[\PureQSD\sbra*{\sqrt{1-2^{-2n'}}, 2^{\frac{-n'+1}{2}}} \text{ is } \BQP{}\text{-hard under Karp reduction}. \] 
\end{lemma}

\begin{proof}
    For any promise problem $(\calP_{yes},\calP_{no})\in \BQP[a(n),b(n)] $ with $a(n)-b(n) \geq 1/\poly(n)$, we assume without loss of generality that the quantum circuit $\hat{C}_x$ has an output length of $n$.\footnote{More precisely, the circuit $\hat{C}_x$ acts on $n$ qubits that are all initialized to $\ket{0}$. No qubit is traced out at the end, and the output length refers to the resulting state just before the final measurement. Following standard convention, we say that $\hat{C}_x$ accepts if the measurement outcome of the designated output qubit is $1$.} Leveraging error reduction for \BQP{} via sequential repetition, for any polynomial $l(n)$, we can achieve that the acceptance probability satisfies $\Pr[C_x \text{ accepts}] \geq 1-2^{-l(n)}$ for \textit{yes} instances, whereas $\Pr[C_x \text{ accepts}] \leq 2^{-l(n)}$ for \textit{no} instances. 
    
    Next, we construct a new quantum circuit $C'_x$ with an additional single-qubit register $\sfF$ initialized to zero. The circuit $C'_x$ is defined as $C'_x \coloneqq C^{\dagger}_x \CNOT_{\sfO\rightarrow \sfF} C_x$, where the single-qubit register $\sfO$ corresponds to the output qubit. It is evident that the output length $n'$ of $C'_x$ satisfies $n' = n+1$. 
    We say that $C'_x$ accepts if the measurement outcomes of all qubits are all zero.  Then, we have: 
    \begin{subequations}
        \label{eq:PureQSD-BQPhard-pacc}
        \begin{align}
            \Pr[C'_x \text{ accepts}] &=  \Abs[\big]{ (\ketbra{\bar{0}}{\bar{0}} \otimes \ketbra{0}{0}_{\sfF}) C'_x (\ket{\bar{0}} \otimes \ket{0}_{\sfF}) }_2^2\\
            &= \abs[\big]{ \bra{\bar{0}} C_x^{\dagger} \ketbra{0}{0}_{\sfO} C_x \ket{\bar{0}} }^2\\
            &= \rbra*{1- \mathrm{Pr}[ C_x \text{ accepts} ]}^2. 
        \end{align}
    \end{subequations}
    Here, the second equality owes to $\CNOT_{\sfO \rightarrow \sfF} = \ketbra{0}{0}_{\sfO}\otimes I_{\sfF} + \ketbra{1}{1}_{\sfO} \otimes X_{\sfF}$. By defining two pure states $\ket{\psi_0} \coloneqq \ket{\bar{0}}\otimes \ket{0}_{\sfF}$ and $\ket{\psi_1} \coloneqq  C'_x (\ket{\bar{0}}\otimes \ket{0}_{\sfF})$ corresponding to $Q_0 = I$ and $Q_1=C'_x$, respectively, we can derive the following from \Cref{eq:pure-traceDist-vs-fidelity}:
    \begin{equation}
        \label{eq:PureQSD-BQPhard}
        \Pr[C'_x \text{ accepts}] = |\innerprod{\psi_0}{\psi_1}|^2 = 1-\td(\ketbra{\psi_0}{\psi_0},\ketbra{\psi_1}{\psi_1})^2.
    \end{equation}
    \Cref{eq:PureQSD-BQPhard-pacc,eq:PureQSD-BQPhard} together imply that $\td(\ketbra{\psi_0}{\psi_0}, \ketbra{\psi_1}{\psi_1}) = \sqrt{1 - \rbra*{1-\Pr[C_x \text{ accepts}]}^2}$. Therefore, we obtain the following bounds: 
    \begin{itemize}[itemsep=0.33em,topsep=0.33em,parsep=0.33em]
        \item For \textit{yes} instances, $\td(\ketbra{\psi_0}{\psi_0}, \ketbra{\psi_1}{\psi_1}) \geq \sqrt{1-2^{-2l(n)}} \geq \sqrt{1-2^{-2l(n'-1)}}$.
        \item For \textit{no} instances, $\td(\ketbra{\psi_0}{\psi_0}, \ketbra{\psi_1}{\psi_1}) \leq \sqrt{1-\rbra*{1-2^{-l(n)}}^2} \leq 2^{\frac{-l(n)+1}{2}} \leq 2^{\frac{-l(n'-1)+1}{2}}$. \qedhere 
    \end{itemize}
\end{proof}

Lastly, combining the proof strategy outlined in~\cite[Section 3]{Kobayashi03} and the reduction from \QEA{} to \QSCMM{} in~\cite[Section 5.3]{BASTS10}, the \NIQSZK{} hardness of \QSCMM{} was established in~\cite[Section 8.1]{CCKV08} with an appropriate parameter trade-off:
\begin{lemma}[\QSCMM{} is \NIQSZK{}-hard, adapted from~{\cite[Section 8.1]{CCKV08}}]
    \label{QSCMM-is-NIQSZKhard}
    \[\text{For any } n\geq 3,~ \QSCMM{}\sbra*{1/n, 1- 1/n} \text{ is } \NIQSZK{}\text{-hard under Karp reduction}.\]
\end{lemma}

\subsubsection{Query and sample complexity lower bounds for states and distributions}
\label{subsec:known-query-and-sample-lower-bounds}

We begin by stating a query complexity lower bound for \QSD{}.
Note that an $n$-qubit maximally mixed state $(I/2)^{\otimes n}$ is commutative with any $n$-qubit quantum states $\rho$. Consider the spectral decomposition $\rho = \sum_{i \in [2^n]} \mu_i \ket{v_i}\bra{v_i}$, where $\{\ket{v_i}\}_{i\in [2^n]}$ is an orthonormal basis, we have $\td\rbra*{\rho,(I/2)^{\otimes n}} = \TV(\mathbf{\mu},U_{2^n})$, where $U_{2^n}$ is a uniform distribution over $[2^n]$. Leveraging a similar argument for $\rho_{\ttU}$, as in \Cref{lemma:query-complexity-lower-bound-QSD,lemma:sample-complexity-lower-bound-QSD}, where the eigenvalues of $\rho_{\ttU}$ form a uniform distribution on the support of $\rho$, we can obtain:
\begin{lemma}[Query complexity lower bound for \QSD{}, adapted from~{\cite[Theorem 2]{CFMdW10}}]
    \label{lemma:query-complexity-lower-bound-QSD}
    For any $\epsilon \in (0,1/2]$, there exists an $n$-qubit quantum state $\rho$ of rank $r$ and the corresponding $n$-qubit state $\rho_{\ttU}$ such that the quantum query complexity to decide whether $\td\rbra*{\rho,\rho_{\ttU}}$ is at least $\epsilon$ or exactly $0$, in the purified quantum query access model, is $\Omega(r^{1/3})$.
\end{lemma}

It is noteworthy that the quantum query model used in~\cite{CFMdW10} differs from the purified quantum query access model. Nevertheless, this lower bound also applies to our query model, as the discussion after Definition 3 in~\cite{GL20}.

\vspace{1em}
Next, we introduce a query complexity lower bound of distinguishing probability distributions provided in \cite{Belovs19}, which will be used to prove the quantum query complexity lower bound for estimating the quantum Tsallis entropy. 

\begin{lemma}[Query complexity for distinguishing probability distributions, {\cite[Theorem 4]{Belovs19}}] \label{lemma:q-distinguish-prob-distri}
    Suppose that $U_{p_0}$ and $U_{p_1}$ are two unitary operators such that     
    \[ U_{p_0} \ket{0} = \sum_{j \in [N]} \sqrt{p_0(j)} \ket{j} \ket{\varphi_j} \text{ and } U_{p_1} \ket{0} = \sum_{j \in [N]} \sqrt{p_1(j)} \ket{j} \ket{\psi_j}. \]
    Here, $p_0$ and $p_1$ are probability distributions on $\sbra{N}$, and $\{\ket{\varphi_j}\}$ and $\{\ket{\psi_j}\}$ are orthonormal bases. 
    Then, any quantum query algorithm that distinguishes $U_{p_0}$ and $U_{p_1}$ requires query complexity 
    \[\Omega\rbra{1/d_{\mathrm{H}}\rbra{p_0, p_1}}.\]
    Here, the Hellinger distance is defined as 
    \[
    d_{\mathrm{H}}\rbra{p_0, p_1} = \sqrt{\frac{1}{2}\sum_{j\in\sbra{N}}\rbra*{\sqrt{p_0(j)} - \sqrt{p_1(j)}}^2}.  
    \]
\end{lemma}

It is noteworthy that \cref{lemma:q-distinguish-prob-distri} was ever used as a tool to prove the quantum query complexity lower bounds for the
closeness testing of probability distributions \cite{LWL24} and the estimations of trace distance and fidelity \cite{Wang24pureQSD}. 

\vspace{1em}
Furthermore, we also need a sample complexity lower bound for \QSD{}, which follows from \cite[Theorem 4.2]{OW21} and is specified in \Cref{lemma:sample-complexity-lower-bound-QSD}. Here, \textit{sample complexity} denotes the number of copies of $\rho$ required to accomplish a specific closeness testing task. 

\begin{lemma}[Sample complexity lower bound for \QSD{}, adapted from~{\cite[Corollary 4.3]{OW21}}]
    \label{lemma:sample-complexity-lower-bound-QSD}
    For any $\epsilon \in (0,1/2]$, there exists an $n$-qubit quantum state $\rho$ of rank $r$ and the corresponding $n$-qubit state $\rho_{\ttU}$ such that the quantum sample complexity to decide whether $\td(\rho,\rho_{\ttU})$ is at least $\epsilon$ or exactly $0$ is $\Omega(r/\epsilon^2)$. 
\end{lemma}

\subsection{Polynomial approximations}

We provide several useful results and tools for polynomial approximations in this subsection, which are later used in \Cref{subsec:positive-power-func-efficient-poly} to make uniform polynomial approximation \textit{efficient}.

\subsubsection{Best uniform polynomial approximations}
\label{subsubsec:best-uniform-poly-approx}

Let $f(x)$ be a continuous function defined on the interval $[-1,1]$ that we aim to approximate using a polynomial of degree at most $d$. We define $P^*_d$ as a \textit{best uniform approximation} on $[-1,1]$ to $f$ of degree $d$ if, for any degree-$d$ polynomial approximation $P_d$ of $f$, the following holds:  
\[ \max_{x\in[-1,1]} \abs*{f(x) - P^*_d(x)} \leq \max_{x\in[-1,1]} \abs*{f(x) - P_d(x)}.\]
Let $\bbR_d[x]$ be the set of all polynomials (with real coefficients) of degree at most $d$. Equivalently, the best uniform approximation $P^*_d$ to $f$ is the polynomial that solves the minimax problem
\[\min_{P_d \in \bbR_d[x]} \max_{x\in[-1,1]} \abs*{f(x) - P_d(x)}.\] 
Especially, we need the best uniform polynomial approximation of positive constant powers:
\begin{lemma}[Best uniform approximation of positive constant powers, adapted from~{\cite[Section 7.1.41]{Timan63}}]
    \label{lemma:fullrange-bestPolyApprox-positivePower}
    For any positive \emph{(}constant\emph{)} integer $r$ and order $\alpha \in (-1,1)$, let $P^*_d\in\bbR[x]$ be the best uniform polynomial approximation for $f(x) \!=\! x^{r-1} |x|^{1+\alpha}$ of degree $d \!=\! \ceil*{\rbra*{\beta_{\alpha,r}/\epsilon}^{\frac{1}{r+\alpha}}}$, where $\beta_{\alpha,r}$ is a constant depending on $r+\alpha$. Then, for sufficiently small $\epsilon$, it holds that 
    \[\max_{x\in[-1,1]} | P^*_d(x) - f(x) | \leq \epsilon.\]  
\end{lemma}

\subsubsection{Chebyshev expansion and truncations}

We introduce Chebyshev polynomial and an averaged variant of the Chebyshev truncation. 
We recommend \cite[Chapter 3]{Rivlin90} for a comprehensive review of Chebyshev expansion.

\begin{definition}[Chebyshev polynomials]
The Chebyshev polynomials \emph{(}of the first kind\emph{)} $T_k(x)$ are defined via the following recurrence relation: $T_0(x)\coloneqq1$, $T_1(x)\coloneqq x$, and $T_{k+1}(x)\coloneqq2x T_k(x)-T_{k-1}(x)$. For $x \in [-1,1]$, an equivalent definition is $T_k(\cos \theta) = \cos(k \theta)$.
\end{definition}

To use Chebyshev polynomials (of the first kind) for Chebyshev expansion, we first need to define an inner product between two functions, $f$ and $g$, as long as the following integral exists:

\begin{equation}
    \label{eq:polynomial-inner-product}
    \innerprodF{f}{g} \coloneqq \frac{2}{\pi} \int_{-1}^1 \frac{f(x)g(x)}{\sqrt{1-x^2}} \dx.
\end{equation}

The Chebyshev polynomials form an orthonormal basis in the inner product space induced by $\innerprodF{\cdot}{\cdot}$ defined in \Cref{eq:polynomial-inner-product}.
As a result, any continuous and integrable function $f: [-1,1] \rightarrow \bbR$ whose Chebyshev coefficients satisfy $\lim_{k \rightarrow \infty} c_k=0$, where $c_k$ is defined in \Cref{eq:Chebyshev-expansion}, has a Chebyshev expansion given by:
\begin{equation}
    \label{eq:Chebyshev-expansion}
    f(x)=\frac{1}{2} c_0 T_0(x) + \sum_{k=1}^{\infty} c_k T_k(x), \text{ where }  c_k\coloneqq\innerprodF{T_k}{f}.
\end{equation} 

Instead of approximating functions directly by the Chebyshev truncation $\tilde{P}_d = c_0/2 + \sum_{k=1}^d c_k T_k$, we use an average of Chebyshev truncations, known as the \textit{de La Vall\'ee Poussin partial sum}, we obtain the degree-$d$ \textit{averaged Chebyshev truncation} $\hat{P}_{d'}$, which is a polynomial of degree $d'=2d-1$:
\begin{equation}
    \label{eq:averaged-Chebyshev-truncation}
    \hat{P}_{d'}(x) \coloneqq \frac{1}{d} \sum_{l=d}^{d'} \tilde{P}_l(x) 
    = \frac{\hat{c}_0}{2} + \sum_{k=1}^{d'} \hat{c}_k T_k(x), \text{ where } \hat{c}_k = \begin{cases}
        c_k ,& 0 \leq k \leq d\\
        \frac{2d-k}{d} c_k,& k > d
    \end{cases},
\end{equation}
we can achieve the truncation error $4 \epsilon$ for any function that admits Chebyshev expansion. 
\begin{lemma}[Asymptotically best  approximation by averaged Chebyshev truncation, adapted from Exercises 3.4.6 and 3.4.7 in~\cite{Rivlin90}]
    \label{lemma:averaged-Chebyshev-truncation}
    For any function $f$ that has a Chebyshev expansion, consider the degree-$d$ averaged Chebyshev truncation $\hat{P}_{d'}$ defined in \Cref{eq:averaged-Chebyshev-truncation}. 
    Let $\varepsilon_d(f)$ be the truncation error corresponds to the degree-$d$ best uniform approximation on $[-1,1]$ to $f$. If there is a degree-$d$ polynomial $P^*_d\in\bbR[x]$ such that $\max_{x\in[-1,1]} |f(x)-P^*_d(x)| \leq \epsilon$, then
    \[ \max_{x\in[-1,1]} \big| f(x) - \hat{P}_{d'}(x) \big| \leq 4 \varepsilon_d(f) \leq 4 \max_{x\in[-1,1]} |f(x)-P^*_d(x)| \leq 4\epsilon. \]
\end{lemma}

\subsection{Quantum algorithmic toolkit}

In this subsection, we provide several quantum algorithmic tools that serve as elementary building blocks used in \Cref{subsec:TsallisQEA-in-BQP}: the quantum singular value transformation, three useful quantum algorithmic subroutines, and the quantum samplizer, which enables a quantum query-to-sample simulation. 

\subsubsection{Quantum singular value transformation}
We begin by introducing the notion of block-encoding: 

\begin{definition} [Block-encoding]
    A linear operator $A$ on an $\rbra{n+a}$-qubit Hilbert space is said to be an $\rbra{\alpha, a, \epsilon}$-block-encoding of an $n$-qubit linear operator $B$, if 
    \[
    \Abs{\alpha \rbra{\bra{0}^{\otimes a} \otimes I_n} A \rbra{\ket{0}^{\otimes a} \otimes I_n} - B} \leq \epsilon,
    \]
    where $I_n$ is the $n$-qubit identity operator and $\Abs{\cdot}$ is the operator norm.
\end{definition}

Then, we state the quantum singular value transformation: 
\begin{lemma} [Quantum singular value transformation, {\cite[Theorem 31]{GSLW19}}] \label{lemma:qsvt}
    Suppose that unitary operator $U$ is an $\rbra{\alpha, a, \epsilon}$-block-encoding of Hermitian operator $A$, and $P \in \mathbb{R}\sbra{x}$ is a polynomial of degree $d$ with $\abs{P\rbra{x}} \leq \frac 1 2$ for $x \in \sbra{-1, 1}$.
    Then, we can implement a quantum circuit $\tilde U$ that is a $\rbra{1, a+2, 4d\sqrt{\epsilon/\alpha} + \delta}$-block-encoding of $P\rbra{A/\alpha}$, by using $O\rbra{d}$ queries to $U$ and $O\rbra{\rbra{a+1}d}$ one- and two-qubit quantum gates. 
    Moreover, the classical description of $\tilde U$ can be computed in deterministic time $\poly\rbra{d, \log\rbra{1/\delta}}$. 
\end{lemma}

\subsubsection{Quantum subroutines}

The first subroutine is the quantum amplitude estimation: 
\begin{lemma} [Quantum amplitude estimation, {\cite[Theorem 12]{BHMT02}}] \label{lemma:qaa}
    Suppose that $U$ is a unitary operator such that 
    \[
    U \ket{0} \ket{0} = \sqrt{p} \ket{0} \ket{\phi_0} + \sqrt{1-p} \ket{1} \ket{\phi_1},
    \]
    where $\ket{\phi_0}$ and $\ket{\phi_1}$ are normalized pure quantum states and $p \in \sbra{0, 1}$. 
    Then, there is a quantum query algorithm using $O\rbra{M}$ queries to $U$ that outputs $\tilde p$ such that 
    \[
    \Pr\sbra*{ \abs*{\tilde p - p} \leq \frac{2\pi\sqrt{p\rbra{1-p}}}{M} + \frac{\pi^2}{M^2} } \geq \frac{8}{\pi^2}.
    \]
    Moreover, if $U$ acts on $n$ qubits, then the quantum query algorithm can be implemented by using $O\rbra{Mn}$ one- and two-qubit quantum gates.
\end{lemma}

The second subroutine prepares a purified density matrix, originally stated in \cite{LC19}: 
\begin{lemma} [Block-encoding of density operators, {\cite[Lemma 25]{GSLW19}}] \label{lemma:block-encoding-of-density-operators}
    Suppose that $U$ is an $\rbra{n+a}$-qubit unitary operator that prepares a purification of an $n$-qubit mixed quantum state $\rho$. 
    Then, we can implement a unitary operator $W$ by using $1$ query to each of $U$ and $U^\dag$ such that $W$ is a $\rbra{1, n+a, 0}$-block-encoding of $\rho$. 
\end{lemma}

The third subroutine is a specific version of one-bit precision phase estimation~\cite{Kitaev95}, often referred to as the Hadamard test~\cite{AJL09}, as stated in~\cite{GP22}:
\begin{lemma} [Hadamard test for block-encodings, adapted from {\cite[Lemma 9]{GP22}}] \label{lemma:hadamard}
    Suppose that unitary operator $U$ is a $\rbra{1, a, 0}$-block-encoding of an $n$-qubit operator $A$. 
    Then, we can implement a quantum circuit that, on input an $n$-qubit mixed quantum state $\rho$, outputs $0$ with probability $\frac{1}{2}+\frac{1}{2}\Real\sbra{\tr\rbra{A\rho}}$ (resp., $\frac{1}{2}+\frac{1}{2}\Imag\sbra{\tr\rbra{A\rho}}$), by using $1$ query to controlled-$U$ and $O\rbra{1}$ one- and two-qubit quantum gates. 

    Moreover, if an $\rbra{n+a}$-qubit unitary operator $\mathcal{O}$ prepares a purification of $\rho$, then, by combining \cref{lemma:qaa}, we can estimate $\tr\rbra{A\rho}$ to within additive error $\epsilon$ by using $O\rbra{1/\epsilon}$ queries to each of $U$ and $\mathcal{O}$ and $O\rbra{\rbra{n+a}/\epsilon}$ one- and two-qubit quantum gates. 
\end{lemma}

\subsubsection{Quantum samplizer}

We introduce the notion of samplizer in \cite{WZ24b}, which helps us establish the sample complexity upper bound from the query complexity upper bound. 

\begin{definition} [Samplizer]
    A samplizer $\mathsf{Samplize}_*\ave{*}$ is a mapping that converts quantum query algorithms (quantum circuit families with query access to quantum unitary oracles) to quantum sample algorithms (quantum channel families with sample access to quantum states) such that: 
    For any $\delta > 0$, quantum query algorithm $\mathcal{A}^{U}$, and quantum state $\rho$, there exists a unitary operator $U_\rho$ that is a $\rbra{2, a, 0}$-block-encoding of $\rho$ for some $a > 0$, satisfying
    \[
    \Abs*{\mathsf{Samplize}_\delta\ave{\mathcal{A}^U}\sbra{\rho} - \mathcal{A}^{U_{\rho}}}_\diamond \leq \delta,
    \]
    where $\Abs{\cdot}_{\diamond}$ is the diamond norm and $\mathcal{E}\sbra{\rho}\rbra{\cdot}$ is a quantum channel $\mathcal{E}$ with sample access to $\rho$. 
\end{definition}

Then, we include an efficient implementation of the samplizer in \cite{WZ24b}, which is based on quantum principal component analysis \cite{LMR14,KLL+17} and generalizes \cite[Corollary 21]{GP22} and \cite[Theorem 1.1]{WZ23b}. 

\begin{lemma} [Optimal samplizer, {\cite[Theorem 4]{WZ24b}}] \label{lemma:samplizer}
    There is a samplizer $\mathsf{Samplize}_*\ave{*}$ such that for $\delta > 0$ and quantum query algorithm $\mathcal{A}^{U}$ with query complexity $Q$, the implementation of $\mathsf{Samplize}_\delta\ave{\mathcal{A}^U}\sbra{\rho}$ uses $\widetilde O\rbra{Q^2/\delta}$ samples of $\rho$.
\end{lemma}

\section{Efficient quantum algorithms for estimating quantum \texorpdfstring{$q$}{q}-Tsallis entropy} 
\label{sec:algos}

In this section, we propose efficient quantum algorithms for estimating the quantum Tsallis entropy $\Sq(\rho)$ when $q \geq 1+\Omega(1)$, using either queries to the state-preparation circuit or samples of the state $\rho$. 
The key ingredient underlying our algorithms is an \textit{efficient} uniform approximation of positive constant power functions. 
Specifically, our polynomial approximation (\Cref{lemma:computable-fullrange-bestPolyApprox-positivePower}) is ``full-range'', meaning it maintains a uniform error bound across the entire interval $[-1,1]$. This differs from the polynomial approximations commonly used in QSVT, which typically provide separate error bounds for the intervals $[-\delta,\delta]$ and $[-1,-\delta) \cup (\delta,1]$. 

Utilizing our ``full-range'' polynomial approximation, we construct a query-efficient quantum algorithm for estimating $\tr\rbra*{\rho^q}$, as established in \Cref{thm:tr-power-constant-queries}. Consequently, our quantum query algorithm (\Cref{thm:tr-power-constant-queries}) directly leads to \BQP{} containments for the promise problems \TsallisQEA{} and \TsallisQED{}, defined in \Cref{sec:hardness-via-QJTq-reductions}. Furthermore, by employing the samplizer in~\cite{WZ24b}, we develop a sample-efficient quantum algorithm for estimating $\tr\rbra*{\rho^q}$, as presented in \Cref{thm:tr-power-constant-samples}.

\subsection{Efficient uniform approximations to positive constant power functions}
\label{subsec:positive-power-func-efficient-poly}
We provide an efficiently computable uniform approximation of positive constant powers: 

\begin{lemma}[Efficient uniform polynomial approximation of positive constant powers]
    \label{lemma:computable-fullrange-bestPolyApprox-positivePower}
    Let $r$ be a positive \emph{(}constant\emph{)} integer and let $\alpha$ be a real number in $(-1,1)$. For any $\epsilon \in (0,1/2)$, there is a degree-$d$ polynomial $P_d \in \bbR[x]$, where $d = \ceil*{\rbra*{\beta'_{\alpha,r}/\epsilon}^{\frac{1}{r+\alpha}}}$ and $\beta'_{\alpha,r}$ is a constant depending on $r+\alpha$, that can be deterministically computed in $\widetilde{O}(d)$ time. For sufficiently small $\epsilon$, it holds that\emph{:}
    \[ \max_{x \in [-1,1]} \abs*{ \frac{1}{2} x^{r-1} |x|^{1+\alpha} - P_d(x) }\leq \epsilon \quad \text{and} \quad \max_{x \in [-1,1]} \abs*{P_d(x)} \leq 1.\]
    Furthermore, $P_d$ has the same parity as the integer $r-1$. 
\end{lemma}

\begin{proof}
    Let $f(x) \coloneqq \frac{1}{2} x^{r-1} |x|^{1+\alpha}$. 
    For any $\tilde{\epsilon} \in (0,1/8)$, using \Cref{lemma:fullrange-bestPolyApprox-positivePower}, we obtain the degree-$\tilde{d}$ best polynomial approximation $P_{\tilde{d}}^*(x)$, where $\tilde{d} = \ceil*{\rbra*{\beta_{\alpha,r}/\tilde{\epsilon}}^{\frac{1}{r+\alpha}}}$ and $\beta_{\alpha,r}$ is a constant depending on $r+\alpha$, such that
    \begin{equation}
        \label{eq:fullrange-positivePower-errorBound}
        \max_{x\in [-1,1]} \abs*{\frac{1}{2}x^{r-1}|x|^{1+\alpha} - P_{\tilde{d}}^*(x)} \leq \tilde{\epsilon} \quad \text{and} \quad \max_{x \in [-1,1]} \abs*{P^*_{\tilde{d}}(x)} \leq \frac{1}{2}+\tilde{\epsilon}. 
    \end{equation}

    Next, we consider the degree-$\tilde{d}$ averaged Chebyshev truncation (\Cref{eq:averaged-Chebyshev-truncation}) of $f(x)$. In particular, let $d \coloneqq 2\tilde{d}-1 = \ceil*{\rbra*{\beta'_{\alpha,r}/\epsilon}^{\frac{1}{r+\alpha}}}$, where $\beta'_{\alpha,r}$ is another constant depending on $r+\alpha$ and $\epsilon$ will be specified later. We obtain the following degree-$d$ polynomial: 
    \begin{equation}
        \label{eq:fullrange-positivePower-poly}
        P_d(x) = \frac{\hat{c}_0}{2} + \sum_{k=1}^{d} \hat{c}_k T_k(x), \text{ where } \hat{c}_k \coloneqq \begin{cases}
        c_k,& 0\leq k \leq \tilde{d}\\
        \frac{2\tilde{d}-k}{\tilde{d}} c_k,& k > \tilde{d}
        \end{cases} \text{ and } c_k \coloneqq \innerprodF{T_k}{f}.  
    \end{equation}

    Using the asymptotically best approximation by averaged Chebyshev truncation (\Cref{lemma:averaged-Chebyshev-truncation}) and \Cref{eq:fullrange-positivePower-errorBound}, we can derive that $P_d(x)$ satisfies the following: 
    \[ \max_{x\in [-1,1]} \abs*{\frac{1}{2}x^{r-1}|x|^{1+\alpha} - P_{d}(x)} \leq 4\tilde{\epsilon} \coloneqq \epsilon \text{ and } \max_{x \in [-1,1]} \abs*{P_{d}(x)} \leq \frac{1}{2}+4\tilde{\epsilon} = \frac{1}{2}+\epsilon < 1.  \]

    It remains to show that $P_d(x)$ can be computed in deterministic time $\widetilde{O}(d)$. A direct calculation implies that the Chebyshev coefficient $\{c_k\}_{0 \leq k \leq d}$ in \Cref{eq:fullrange-positivePower-poly} satisfy the following: 
    \begin{align*}
        c_{2l+1} &= c_{2l-1} \cdot \frac{r+\alpha-2l+1}{r+\alpha+2l+1}, \quad
        c_{2l} = c_{2l-2} \cdot \frac{r+\alpha-2l+2}{r+\alpha+2l},\\
        c_{0} &= \frac{2}{\pi} \int_{-1}^{1} \frac{\frac{1}{2} x^{r-1} |x|^{1+\alpha} \cdot T_{0}(x)}{\sqrt{1-x^2}} \dx = - \frac{-1+(-1)^r}{2\sqrt{\pi}} \cdot \frac{\Gamma\rbra*{\frac{1}{2}(r+\alpha+1)}}{\Gamma\rbra*{\frac{1}{2}(r+\alpha+2)}},\\
        c_{1} &= \frac{2}{\pi} \int_{-1}^{1} \frac{\frac{1}{2} x^{r-1} |x|^{1+\alpha} \cdot T_{1}(x)}{\sqrt{1-x^2}} \dx =  \frac{1+(-1)^r}{2\sqrt{\pi}} \cdot \frac{\Gamma\rbra*{\frac{1}{2}(r+\alpha+2)}}{\Gamma\rbra*{\frac{1}{2}(r+\alpha+3)}}.
    \end{align*}
    Here, the Gamma function $\Gamma(x) \coloneqq \int_{0}^{\infty} t^{x-1} e^{-x} \dd t$ for any $x > 0$. 

    Consequently, we can recursively compute the averaged Chebyshev coefficient $\{\hat{c}_k\}_{0 \leq k \leq d}$ in deterministic time $\widetilde{O}(d)$. 
    We complete the proof by noting that the Chebyshev polynomials $\cbra*{T_k(x)}_{0\leq k \leq d}$ also can be recursively computed in deterministic time $\widetilde{O}(d)$.
\end{proof}

\subsection{Quantum \texorpdfstring{$q$}{q}-Tsallis entropy approximation for \texorpdfstring{$q$}{q} constantly larger than \texorpdfstring{$1$}{1}}
\label{subsec:TsallisQEA-in-BQP}

\subsubsection{Query-efficient quantum algorithm for estimating \texorpdfstring{$\tr(\rho^q)$}{tr(ρ\^q)}}

We now present efficient quantum query algorithms for estimating the $q$-Tsallis entropy of a mixed quantum state. 
For readability, their framework is given in \cref{algo:q-Tsallis-estimation}. 

\begin{algorithm}[!htp]
    \caption{A framework for estimating $q$-Tsallis entropy for $q\geq 1+\Omega(1)$ (query access).}
    \label{algo:q-Tsallis-estimation}
    \begin{algorithmic}[1]
        \Require A quantum circuit $Q$ that prepares a purification of an $n$-qubit mixed quantum state $\rho$, and a precision parameter $\epsilon \in \rbra{0, 1}$.
        \Ensure A single bit $b \in \cbra{0, 1}$ such that $\Pr\sbra{b = 0} \approx \frac{1}{2} + \frac{1}{8} \tr\rbra{\rho^q}$.

        \State Implement a unitary operator $U_\rho$ that is a block-encoding of $\rho$ by \cref{lemma:block-encoding-of-density-operators}, using $O\rbra{1}$ queries to $Q$.

        \State Let $P\rbra{x}$ be a polynomial that approximates $\frac 1 4 x^{q-1}$ in the range $\sbra{0, 1}$, where $P\rbra{x}$ is determined according to $\epsilon$, $n$, and $q$. More precisely, for constant $q > 1$, $P\rbra{x}$ is chosen by \cref{lemma:computable-fullrange-bestPolyApprox-positivePower}. 
    
        \State Implement a unitary operator $U_{P\rbra{\rho}}$ that is a block-encoding of $P\rbra{\rho}$ by quantum singular value transformation (\cref{lemma:qsvt}), using $O\rbra{\deg\rbra{P}}$ queries to $U_\rho$. 
    
        \State Perform the Hadamard test on $\rho$ and $U_{P\rbra{\rho}}$ by \cref{lemma:hadamard}, and return the measurement outcome. 
    \end{algorithmic}
\end{algorithm}

\begin{theorem} [Trace estimation of quantum state constant powers via queries] \label{thm:tr-power-constant-queries}
    Suppose that $Q$ is a unitary operator that prepares a purification of mixed quantum state $\rho$. 
    For every $q \geq 1+\Omega(1)$, there is a quantum query algorithm that estimates $\tr\rbra{\rho^q}$ to within additive error $\epsilon$ by using $O\rbra{1/\epsilon^{1+\frac{1}{q-1}}}$ queries to $Q$.
\end{theorem}

\begin{proof}
    Let $Q$ be an $\rbra{n+a}$-qubit unitary operator that prepares a purification of the $n$-qubit mixed quantum state $\rho$. 
    Then, by \cref{lemma:block-encoding-of-density-operators}, we can implement a unitary operator $U_{\rho}$ that is a $\rbra{1, n+a, 0}$-block-encoding of $\rho$, by using $O\rbra{1}$ queries to $Q$. 

    Let $\epsilon_p \in (0, 1)$ be a parameter to be determined later. 
    By \cref{lemma:computable-fullrange-bestPolyApprox-positivePower} with $r \coloneqq \max\cbra{\floor{q-1}, 1}$, $\alpha \coloneqq q-1-r$, and $\epsilon \coloneqq \epsilon_p$, there exists a polynomial $P \in \mathbb{R}\sbra{x}$ of degree $d = O\rbra{{1}/{\epsilon_p^{\frac{1}{q-1}}}}$ such that 
    \[
    \max_{x \in \sbra{0, 1}} \abs*{P\rbra{x} - \frac{1}{2}x^{q-1}} \leq \epsilon_p, \quad \textup{and} \max_{x \in \sbra{-1, 1}} \abs{P\rbra{x}} \leq 1.
    \]
    By \cref{lemma:qsvt} with $P \coloneqq \frac{1}{2}P$, $\alpha \coloneqq 1$, $a \coloneqq n+a$, $\epsilon \coloneqq 0$ and $d \coloneqq O\rbra{{1}/{\epsilon_p^{\frac{1}{q-1}}}}$, we can implement a quantum circuit $U_{P\rbra{\rho}}$ that is a $\rbra{1, n+a+2, \delta}$-block-encoding of $\frac{1}{2}P\rbra{\rho}$, by using $O\rbra{{1}/{\epsilon_p^{\frac{1}{q-1}}}}$ queries to $U_{\rho}$. 
    Moreover, the classical description of $U_{P\rbra{\rho}}$ can be computed in deterministic time $\poly\rbra{1/\epsilon_p, \log\rbra{1/\delta}}$.

    Suppose that $U_{P\rbra{\rho}}$ is a $\rbra{1, n+a+2, 0}$-block-encoding of $A$, i.e., $\Abs{A - \frac{1}{2}P\rbra{\rho}} \leq \delta$. 
    Then, by \cref{lemma:hadamard}, we can obtain an estimate $\tilde x$ of $\tr\rbra{A\rho}$ to within additive error $\epsilon_H$ by using $O\rbra{1/\epsilon_H}$ queries to each of $U_{P\rbra{\rho}}$ and $Q$ such that 
    \begin{equation} \label{eq:tr-power-had-constant}
    \Pr\sbra[\big]{ \abs*{ \tilde x - \tr\rbra{A\rho} } \leq \epsilon_H } \geq \frac{2}{3}.
    \end{equation}
    It can be seen that, in the overall quantum circuit to obtain $\tilde x$, the number of queries to $Q$ is
    \[
    O\rbra*{\frac{1}{\epsilon_H}} \cdot O\rbra*{\frac{1}{\epsilon_p^{\frac{1}{q-1}}}} = O\rbra*{\frac{1}{\epsilon_H\epsilon_p^{\frac{1}{q-1}}}},
    \]
    and the number of one- and two-qubit quantum gates is
    \[
    O\rbra*{\frac{n+a}{\epsilon_H\epsilon_p^{\frac{1}{q-1}}}}.
    \]
    Moreover, the classical description of the overall quantum circuit can be computed in deterministic time $\poly\rbra{1/\epsilon_p, 1/\epsilon_H, \log\rbra{1/\delta}}$.
    
    On the other hand, we have
    \begin{equation} \label{eq:tr-power-err-delta-constant}
    \abs*{ \tr\rbra{A\rho} - \tr\rbra*{\frac{1}{2}P\rbra{\rho}\rho} } \leq \Abs*{A - \frac{1}{2}P\rbra{\rho}} \leq \delta,
    \end{equation}
    where we use the inequality $\abs{\tr\rbra{AB}} \leq \Abs{A} \tr\rbra{\abs{B}}$ (which is a special case of the matrix H\"{o}lder inequality, e.g., \cite[Theorem 2]{Bau11}).
    We also have
    \begin{equation} \label{eq:tr-power-err-constant}
    \abs*{ \tr\rbra*{\frac{1}{2}P\rbra{\rho}\rho} - \tr\rbra*{\frac{1}{4}\rho^q} } \leq \frac{1}{2}\epsilon_p.
    \end{equation}
    To see \cref{eq:tr-power-err-constant}, suppose that $\rho = \sum_{j} \lambda_j \ketbra{\psi_j}{\psi_j}$ is the spectrum decomposition of $\rho$ with $\lambda_j \geq 0$ for all $j$ and $\sum_j \lambda_j = 1$.
    Then,
    \begin{align*}
        \abs*{ \tr\rbra*{\frac{1}{2}P\rbra{\rho}\rho} - \tr\rbra*{\frac{1}{4}\rho^q} }
        & = \abs*{ \sum_j \rbra*{  \frac{1}{2} P\rbra{\lambda_j} \lambda_j - \frac{1}{4} \lambda_j^q} } \\
        & \leq \sum_{j} \frac{1}{2} \lambda_j \abs*{P\rbra{\lambda_j} - \frac{1}{2}\lambda_j^{q-1}} \\
        & \leq \frac{1}{2} \sum_{j} \lambda_j \epsilon_p = \frac{1}{2} \epsilon_p.
    \end{align*}
    Finally, by combining \cref{eq:tr-power-had-constant,eq:tr-power-err-delta-constant,eq:tr-power-err-constant}, we obtain
    \[
    \Pr\sbra[\big]{ \abs*{ 4\tilde x - \tr\rbra*{\rho^q} } \leq 2 \epsilon_p + 4\epsilon_H + 4\delta } \geq \frac{2}{3}.
    \]
    To make $4\tilde x$ an $\epsilon$-estimate of $\tr\rbra{\rho^q}$ with high probability, it is sufficient to take  $\epsilon_p = \epsilon_H = \delta = \epsilon/10$, thereby using 
    \[
    O\rbra*{\frac{1}{\epsilon^{1+\frac{1}{q-1}}}}
    \]
    queries to $Q$. 
\end{proof}

\subsubsection{Sample-efficient quantum algorithm for estimating \texorpdfstring{$\tr(\rho^q)$}{tr(ρ\^q)}}

We also study the sample complexity for the trace estimation of quantum state powers, which is obtained by extending the quantum query algorithm in \cref{thm:tr-power-constant-queries} via the samplizer in \cref{lemma:samplizer}.
An illustrative framework is given in \cref{algo:q-Tsallis-estimation-sample}.

\begin{algorithm}[!htp]
    \caption{A framework for estimating $q$-Tsallis entropy for $q \geq 1+\Omega(1)$ (sample access).}
    \label{algo:q-Tsallis-estimation-sample}
    \begin{algorithmic}[1]
        \Require Independent and identical samples of an $n$-qubit mixed quantum state $\rho$, and parameters $q > 1$ and $\delta, \epsilon_p, \delta_p \in \rbra{0, 1}$.
        \Ensure A single bit $b \in \cbra{0, 1}$ such that $\Pr\sbra{b = 0} \approx \frac{1}{2} + \frac{1}{2^{q+3}} \tr\rbra{\rho^q}$.

        \begin{tcolorbox}[colback=gray!12]
        \Function{ApproxPower}{$q, \epsilon_p, \delta_p$}${}^U$
        \renewcommand{\algorithmicrequire}{\qquad \textbf{Input:}}
        \renewcommand{\algorithmicensure}{\qquad \textbf{Output:}} 
        \Require A unitary $\rbra{1, a, 0}$-block-encoding $U$ of $A$, and parameters $q > 1, \epsilon_p, \delta_p \in \rbra{0, 1}$. 
        \Ensure A unitary operator $\widetilde U$.
        \State Let $P\rbra{x}$ be a polynomial of degree $d = O\rbra{{1}/{\epsilon_p^{\frac{1}{q-1}}}}$ such that $\max_{x \in \sbra{0, 1}} \abs{P\rbra{x} - \frac{1}{2}x^{q-1}} \leq \epsilon_p$ and $\max_{x \in \sbra{-1, 1}} \abs{P\rbra{x}} \leq 1$ (by \cref{lemma:computable-fullrange-bestPolyApprox-positivePower}).
        \State Construct a unitary $\rbra{1, a+2, \delta_p}$-block-encoding $\widetilde U$ of $\frac{1}{2}P\rbra{A}$ (by \cref{lemma:qsvt}). 
        \State \Return $\widetilde U$.
        \EndFunction
        \end{tcolorbox}

        \State Let $b'$ be the outcome of the Hadamard test (by \cref{lemma:hadamard}) performing on the quantum state $\rho$ and $\mathsf{Samplize}_{\delta}\ave{\begin{tcolorbox}[nobeforeafter,
  after={\xspace},
  hbox,
  tcbox raise base,
  fontupper=\ttfamily,
  colback=gray!12,
  size=fbox] $\texttt{ApproxPower}\rbra{q, \epsilon_p, \delta_p}^{U}$\end{tcolorbox}}\sbra{\rho}$ (as if it were unitary).

        \State \Return $b'$.
    \end{algorithmic}
\end{algorithm}

\begin{theorem}[Trace estimation of quantum state constant powers via samples] \label{thm:tr-power-constant-samples}
    For every $q \geq 1+\Omega(1)$, there is a quantum sample algorithm that estimates $\tr\rbra{\rho^q}$ to within additive error $\epsilon$ by using $\widetilde{O}\rbra{1/\epsilon^{3+\frac{2}{q-1}}}$ samples of $\rho$.
\end{theorem}

\begin{proof}
    Let unitary operator $U$ be a $\rbra{1, a, 0}$-block-encoding of $A$ for some $a > 0$ and let $\epsilon_p, \delta_p \in \rbra{0, 1}$ be parameters to be determined. 
    By \cref{lemma:computable-fullrange-bestPolyApprox-positivePower} with $r \coloneqq \max\cbra{\floor{q-1}, 1}$, $\alpha \coloneqq q-1-r$, and $\epsilon \coloneqq \epsilon_p$, there is a polynomial $P \in \mathbb{R}\sbra{x}$ of degree $d = O\rbra{{1}/{\epsilon_p^{\frac{1}{q-1}}}}$ such that 
    \[
    \max_{x \in \sbra{0, 1}} \abs*{P\rbra{x} - \frac{1}{2}x^{q-1}} \leq \epsilon_p, \quad \textup{and} \max_{x \in \sbra{-1, 1}} \abs{P\rbra{x}} \leq 1.
    \]
    By \cref{lemma:qsvt} with $P \coloneqq \frac{1}{2}P$, $\alpha \coloneqq 1$, $a \coloneqq n+a$, $\epsilon \coloneqq 0$, $\delta \coloneqq \delta_p$ and $d \coloneqq O\rbra{{1}/{\epsilon_p^{\frac{1}{q-1}}}}$, we can implement a quantum circuit $U_{P\rbra{A}}$ that is a $\rbra{1, n+a+2, \delta_p}$-block-encoding of $\frac{1}{2}P\rbra{A}$, by using $O\rbra{{1}/{\epsilon_p^{\frac{1}{q-1}}}}$ queries to $U$. 
    Moreover, the classical description of $U_{P\rbra{A}}$ can be computed in deterministic time $\poly\rbra{1/\epsilon_p, \log\rbra{1/\delta_p}}$.
    Let $\texttt{ApproxPower}\rbra{q, \epsilon_p, \delta_p}^{U}$ denote the procedure of implementing $U_{P\rbra{A}}$ by using queries to $U$.

    For our purpose, we take $A \coloneqq \rho/2$.\footnote{Compared to the proof of \cref{thm:tr-power-constant-queries}, the choice of $A = \rho/2$ incurs an additional factor of $1/2$, which leads to an exponential coefficient of $2^q$ in the complexity (despite the fact that $q$ is considered to be a constant). This factor, required by \cref{lemma:samplizer}, is due to the technical reason mentioned in~\cite[Footnote 3]{WZ23b} and~\cite[Footnote 8]{WZ24b}.} 
    Suppose that $U_{P\rbra{\frac{\rho}{2}}}$ is a $\rbra{1, n+a+2, 0}$-block-encoding of $B$, then $\Abs{B - \frac{1}{2}P\rbra{\frac{\rho}{2}}} \leq \delta_p$. 
    Let $b \in \cbra{0, 1}$ be the outcome of the Hadamard test (by \cref{lemma:hadamard}) on $\rho$ and $U_{P\rbra{\frac{\rho}{2}}}$, then 
    \begin{equation} \label{eq:Pr-b=0}
    \Pr\sbra{b = 0} = \frac{1}{2} + \frac{1}{2} \Real\sbra{\tr\rbra{B\rho}}.
    \end{equation}
    Let $\delta \in \rbra{0, 1}$ be a parameter to be determined, and let $b' \in \cbra{0, 1}$ be the outcome of the Hadamard test (by \cref{lemma:hadamard}) on $\rho$ and $\mathsf{Samplize}_{\delta}\ave{\texttt{ApproxPower}\rbra{q, \epsilon_p, \delta_p}^{U}}\sbra{\rho}$ (as if it were $U_{P\rbra{\frac{\rho}{2}}}$).
    Then, 
    \begin{equation} \label{eq:pr-diff}
    \abs[\big]{ \Pr\sbra{b = 0} - \Pr\sbra{b' = 0} } \leq \delta.
    \end{equation}
    Now we repeat the Hadamard test $k$ times, obtaining outcomes $b_1', b_2', \dots, b_k' \in \cbra{0, 1}$, where $k$ is an integer to be determined. 
    Let $X = \frac{1}{k}\sum_{j=1}^k b_j'$. 
    Then, by Hoeffding's inequality ({\cite[Theorem 2]{Hoe63}}), we have 
    \begin{equation} \label{eq:hoeffding}
    \Pr\sbra*{ \abs*{X - \mathbb{E}\sbra{b'}} \leq \epsilon_H } \geq 1 - 2\exp\rbra*{-2k\epsilon_H^2}.
    \end{equation}
    On the other hand, similar to the proof of \cref{thm:tr-power-constant-queries}, we have
    \begin{equation} \label{eq:tr-power-err-delta-constant-sample}
    \abs*{ \Real\sbra*{\tr\rbra{B\rho}} - \tr\rbra*{\frac{1}{2}P\rbra*{\frac{\rho}{2}}\rho} } \leq \abs*{ \tr\rbra{B\rho} - \tr\rbra*{\frac{1}{2}P\rbra*{\frac{\rho}{2}}\rho} } \leq \Abs*{B - \frac{1}{2}P\rbra*{\frac{\rho}{2}}} \leq \delta_p.
    \end{equation}
    We also have
    \begin{equation} \label{eq:tr-power-err-constant-sample}
    \abs*{ \tr\rbra*{\frac{1}{2}P\rbra*{\frac{\rho}{2}}\rho} - \tr\rbra*{\frac{1}{2^{q+2}}\rho^q} } \leq \frac{1}{2}\epsilon_p.
    \end{equation}
    To see \Cref{eq:tr-power-err-constant-sample}, suppose that $\rho = \sum_{j} \lambda_j \ketbra{\psi_j}{\psi_j}$ is the spectrum decomposition of $\rho$ with $\lambda_j \geq 0$ for all $j$ and $\sum_j \lambda_j = 1$.
    Then,
    \begin{align*}
        \abs*{ \tr\rbra*{\frac{1}{2}P\rbra*{\frac{\rho}{2}}\rho} - \tr\rbra*{\frac{1}{2^{q+2}}\rho^q} }
        & = \abs*{ \sum_j \rbra*{  \frac{1}{2} P\rbra*{\frac{\lambda_j}{2}} \lambda_j - \frac{1}{2^{q+2}} \lambda_j^q} } \\
        & \leq \sum_{j} \frac{1}{2} \lambda_j \abs*{P\rbra*{\frac{\lambda_j}{2}} - \frac{1}{2}\rbra*{\frac{\lambda_j}{2}}^{q-1}} \\
        & \leq \frac{1}{2} \sum_{j} \lambda_j \epsilon_p = \frac{1}{2} \epsilon_p.
    \end{align*}
    Finally, by combining \cref{eq:Pr-b=0,eq:pr-diff,eq:hoeffding,eq:tr-power-err-delta-constant-sample,eq:tr-power-err-constant-sample}, we obtain
    \[
    \Pr\sbra*{ \abs*{ 2^{q+2}\rbra*{1-2X} - \tr\rbra*{\rho^q} } \leq 2^{q+1} \rbra*{4\delta + 4\epsilon_H + 2\delta_p + \epsilon_p} } \geq 1 - 2\exp\rbra{-2k\epsilon_H^2}.
    \]
    By taking $\delta = \epsilon_H = \delta_p = \epsilon_p \coloneqq 2^{-q-5}\epsilon$ and $k \coloneqq \ceil*{\frac{\ln\rbra{6}}{2\epsilon_H^2}}$, we have
    \[
    \Pr\sbra*{ \abs*{ 2^{q+2}\rbra*{1-2X} - \tr\rbra*{\rho^q} } \leq \epsilon } \geq \frac{2}{3},
    \]
    which means that $2^{q+2}\rbra{1-2X}$ is an $\epsilon$-estimate of $\tr\rbra{\rho^q}$ with high probability. 

    To complete the proof, we analyze the sample complexity of our algorithm. 
    The algorithm consists of $k$ repetitions of the Hadamard test, and each Hadamard test uses one sample of $\rho$ and one call to $\mathsf{Samplize}_{\delta}\ave{\texttt{ApproxPower}\rbra{q, \epsilon_p, \delta_p}^{U}}\sbra{\rho}$.
    Here, $\texttt{ApproxPower}\rbra{q, \epsilon_p, \delta_p}^{U}$ uses $O\rbra{{1}/{\epsilon_p^{\frac{1}{q-1}}}}$ queries to $U$, and thus by \cref{lemma:samplizer} we can implement  $\mathsf{Samplize}_{\delta}\ave{\texttt{ApproxPower}\rbra{q, \epsilon_p, \delta_p}^{U}}\sbra{\rho}$ by using $\widetilde{O}\rbra{1/\rbra{\delta \epsilon_p^{\frac{2}{q-1}}}}$ samples of $\rho$.
    Therefore, the total number of samples of $\rho$ is 
    \[
    k \cdot \widetilde{O}\rbra*{\frac{1}{\delta \epsilon_p^{\frac{2}{q-1}}}} = \widetilde{O}\rbra*{\frac{1}{\epsilon^{3+\frac{2}{q-1}}}}. \qedhere
    \]
\end{proof}

\section{Properties of quantum Jensen-Tsallis divergence and Tsallis entropy}
\label{sec:properties-QJTq}

In this section, we present inequalities between the quantum $q$-Jensen-Tsallis divergence ($1 \leq q \leq 2$) and the trace distance. Our results (\Cref{thm:QJTq-vs-td}) extend the previous results for the quantum Jensen--Shannon divergence ($q=1$), as stated in~\cite[Theorem 14]{BH09}.
\begin{theorem}[$\QJTq$ vs.~$\td$]
    \label{thm:QJTq-vs-td}
    For any quantum states $\rho_0$ and $\rho_1$, and $1 \leq q \leq 2$, we have\emph{:}
    \[\Hq\rbra*{\frac{1}{2}} - \Hq\rbra*{\frac{1-\td(\rho_0,\rho_1)}{2}} \leq \QJTq(\rho_0,\rho_1) 
    \leq \Hq\rbra*{\frac{1}{2}} \cdot \td(\rho_0,\rho_1)^q.\]
\end{theorem}

To prove \Cref{thm:QJTq-vs-td}, we first need to prove the data-processing inequality for $\QJTq$ (\Cref{lemma:QJSq-data-processing-inequaltiy}), which crucially relies on the relatively recent results on the \textit{joint convexity} of $\QJTq$~\cite{CT14,Virosztek19}. 
Consequently, we can establish \Cref{thm:QJTq-vs-td} by proving the inequalities in \Cref{subsec:QJTq-vs-td-inequalities}. In particular, the lower bound on $\QJTq$ in terms of $\td$ (\Cref{lemma:td-leq-QJSq}) holds for $q\in[1,2]$, and the upper bound on $\QJTq$ in terms of $\td$ (\Cref{lemma:QJSq-leq-traceDistLike}) for the same range of $q$.

\vspace{1em}
Next, to utilize \Cref{lemma:td-leq-QJSq}, we provide bounds of the Tsallis binary entropy in \Cref{subsec:Tsallis-binary-entropy-bounds}:
\begin{theorem}[Tsallis binary entropy bounds]
    \label{thm:Tsallis-binary-entropy-bounds}
    For any $p=(x,1-x)$, let $\Hq(x)$ denote the Tsallis binary entropy with $1\leq q \leq 2$, we have\emph{:}
    \[ \Hq(1/2) \cdot 4x(1-x) \leq \Hq(x) \leq \Hq(1/2) \cdot \smash{(4x(1-x))^{1/2}}.\]
\end{theorem}

It is noteworthy that the best known bounds for the Shannon binary entropy ($q=1$) are $\H(1/2) \cdot 4x(1-x) \leq \H(x) \leq \H(1/2) \cdot (4x(1-x))^{\frac{1}{2\H(1/2)}}$, as shown in~\cite[Theorem 1.2]{Topsoe01}. Our lower bound on the Tsallis binary entropy (\Cref{lemma:Tsallis-binary-entropy-lower-bound}) matches the case of $q=1$, whereas our upper bound (\Cref{lemma:Tsallis-binary-entropy-upper-bound}) only aligns with a weaker bound $\H(x) \leq \H(1/2) \cdot (4x(1-x))^{1/2}$ in~\cite[Theorem 8]{Lin91} and the proof of \Cref{lemma:Tsallis-binary-entropy-upper-bound} is more complicated than in the case of $q=1$. 

\vspace{1em}
Lastly, we provide the inequalities between the Tsallis entropy of a distribution $p$ and the total variation distance between $p$ and the uniform distribution $\nu$ of the same dimension, as stated in~\Cref{lemma:inequality-uniformTV-TsallisEA}. By adding an additional assumption regarding $q$ and $\TV(p,\nu)$, this lemma partially generalizes the previous result for the case of $q=1$ (cf.~\cite[Fact 8.4]{CCKV08} and~\cite[Lemma 16]{KLGN19}) to the case of $q>1$. 

\subsection{Data-processing inequality for \QJTq{} from the joint convexity}

With the correspondence between $\QJS$ and the quantum relative entropy (\Cref{eq:QJS-as-symmetrized-relative-entropy}), the joint convexity of $\QJS$ directly follows from the joint convexity of the quantum relative entropy~\cite{Lieb73,Uhlmann77} (see also~\cite{Ruskai22} for a simple proof). However, since $\QJTq$ does not correspond to a Tsallis variant of quantum relative entropy (e.g., quasi-entropy~\cite[Equation (3.23)]{Petz07}) in this sense, the joint convexity of $\QJTq$ can only be established by the recent results of~\cite{CT14,Virosztek19}:
\begin{lemma}[Joint convexity of $\QJTq$, adapted from~\cite{CT14,Virosztek19}]
    \label{lemma:QJSq-joint-convexity}
    Let $k$ be an integer. For any $i \in [k]$, let $\rho^{(i)}_0$ and $\rho^{(i)}_1$ be two quantum states. Let $k$-tuple $\mu \coloneq (\mu_1, \cdots, \mu_k)$ be a probability distribution. Then, for any $q \in [1,2]$ and $t \in (0,1)$, the joint convexity of $\QJTq$ holds\emph{:}
    \[ \QJTq \!\left(\sum_{i \in [k]} \mu_i \rho_0^{(i)}, \sum_{i \in [k]} \mu_i \rho_1^{(i)} \right) \leq \sum_{i \in [k]} \mu_i \QJTq\!\left(\rho_0^{(i)}, \rho_1^{(i)}\right). \]
\end{lemma}
\begin{proof}
    Following~\cite[Theorem 2.3(2)]{CT14}, we know that the quantum $q$-Tsallis entropy $\Sq(\rho)$ for $1\leq q \leq 2$ is in the Matrix Entropy Class~\cite[Definition 2.2]{CT14} (or~\cite[Definition 2]{Virosztek19}). Therefore, as a corollary of~\cite[Theorem 1]{Virosztek19}, we can obtain: for any $1 \leq q \leq 2$ and $0 < \lambda < 1$, 
    \begin{equation}
        \label{eq:QJSq-joint-convexity}
        \QJTq\rbra*{(1-\lambda) \rho_0 + \lambda \rho'_0, (1-\lambda) \rho_1 + \lambda \rho'_1} \leq (1-\lambda)\QJTq\rbra*{\rho_0,\rho_1) + \lambda \QJTq(\rho'_0,\rho'_1}.
    \end{equation}
    
    Hence, we can complete the proof by applying \Cref{eq:QJSq-joint-convexity} inductively. 
\end{proof}

\begin{remark}[Joint convexity of $\QJT{q,t}$]
It is noteworthy that \Cref{lemma:QJSq-joint-convexity} applies to a generalized version of $\QJTq$, denoted as $\QJT{q,t}$, such that $\QJT{q,1/2} = \QJTq$:
\[\forall t \in (0,1),~\QJT{q,t} \coloneqq \Sq \left( (1-t)\rho_0 + t \rho_1 \right) - (1-t) \Sq(\rho_0) - t\Sq(\rho_1).\] 
\Cref{lemma:QJSq-unitary-invariance} also directly extends to $\QJT{q,t}$, and consequently, \Cref{lemma:QJSq-data-processing-inequaltiy} holds for $\QJT{q,t}$ with $1 \leq q \leq 2$. However, the inequalities between $\QJTq$ and the trace distance provided in this work, particularly \Cref{lemma:td-leq-QJSq,lemma:QJSq-leq-traceDistLike}, do not extend to $\QJT{q,t}$ for $1 \leq q \leq 2$. 
\end{remark}

\begin{lemma}[Data-processing inequality for $\QJTq$]
    \label{lemma:QJSq-data-processing-inequaltiy}
    For any quantum state $\rho_0$ and $\rho_1$, any quantum channel $\Phi$, and $1 \leq q \leq 2$, we have \[ \QJTq(\Phi(\rho_0),\Phi(\rho_1)) \leq \QJTq(\rho_0,\rho_1). \]
\end{lemma}

Interestingly, the inequality in \Cref{lemma:QJSq-data-processing-inequaltiy} \textit{cannot} hold for $0 \leq q < 1$. We can see this by considering pure states $\ketbra{\psi}{\psi}$ and $\ketbra{\phi}{\phi}$, and their average $\hat{\rho}_{\psi,\phi} \coloneqq \frac{1}{2} \big(\ketbra{\psi}{\psi}+\ketbra{\phi}{\phi} \big)$, then $\QJTq(\ketbra{\psi}{\psi},\ketbra{\phi}{\phi}) = \Sq\big( \hat{\rho}_{\psi,\phi} \big)$. Following \Cref{lemma:QJSq-data-processing-inequaltiy}, we have $\Sq\big(\Phi\big( \hat{\rho}_{\psi,\phi} \big)\big) \leq \Sq\big(\hat{\rho}_{\psi,\phi}\big)$ for $q \in [1,2]$. However, using~\cite[Corollary 2.6]{FYK04}, we have $\Sq\big(\Phi\big(\hat{\rho}_{\psi,\phi})\big) \geq \Sq\big(\hat{\rho}_{\psi,\phi}\big)$ for $q \in [0,1)$.

\begin{proof}[Proof of \Cref{lemma:QJSq-data-processing-inequaltiy}]
    The case of $q=1$ coincides with the quantum Jensen--Shannon divergence: Using \Cref{eq:QJS-as-symmetrized-relative-entropy}, $\QJS(\Phi(\rho_0),\Phi(\rho_1)) \leq \QJS(\rho_0,\rho_1)$ follows from the data-processing inequality of the quantum relative entropy~\cite{Lindblad75,Uhlmann77} (see also~\cite[Theorem 3.9]{Petz07}). 

    It remains to prove the case for $1 < q \leq 2$. We use the standard proof strategy to derive the data-processing inequality from joint convexity, as in~\cite[Theorem 2.5]{FYK04}. %(see also~\cite[Section 7]{Ruskai22})
    
    First, we consider the case of the partial trace $\tr_{\sfB}$ on the quantum registers $\sfA$ and $\sfB$, where $\rho_0,\rho_1 \in \mathrm{L}\big(\calH_{\sfA\sfB}\big)$ and $\dim(\calH_{\sfB}) = N_{\sfB}$. 
    Since $\QJTq(\rho_0 \otimes \tilde{I}_{\sfB}, \rho_1 \otimes \tilde{I}_{\sfB}) = \QJTq(\rho_0,\rho_1)$ where $\tilde{I}_{\sfB} \coloneqq I_{\sfB}/N_{\sfB}$ is the maximally mixed state in $\sfB$, it suffices to consider a quantum channel on registers $\sfA$ and $\sfB$ that is completely depolarizing on $\sfB$ and identity on $\sfA$, denoted as $\Phi_{\tr_{\sfB}}$. Noting that $\Phi_{\tr_{\sfB}}$ can be expressed as a convex combination of unitary channels (e.g.,~\cite[Exercise 4.4.9]{Wil13} or~\cite[Equation (9)]{Ruskai22}), for any quantum state $\rho_{\sfA\sfB}$ on registers $\sfA$ and $\sfB$, we can obtain: 
    \[ \Phi_{\tr_{\sfB}}(\rho_{\sfA\sfB}) \coloneqq \tr_{\sfB}(\rho_{\sfA\sfB}) \otimes \tr(\rho_{\sfA\sfB}) \tilde{I}_{\sfB} = \sum_{l \in [N_{\sfB}^2]} \frac{1}{N_{\sfB}^2} (I_{\sfA} \otimes U_l) \rho_{\sfA\sfB} (I_{\sfA} \otimes U_l)^{\dagger},\]
    where $U_l$ is a unitary operator on $\sfB$ for each $l \in [N_{\sfB}^2]$.

    Leveraging the joint convexity (\Cref{lemma:QJSq-joint-convexity}) and the unitary invariance (\Cref{lemma:QJSq-unitary-invariance}) of $\QJS_q$, we derive the following data-processing inequality concerning the quantum channel $\Phi_{\tr_{\sfB}}$:
    \begin{subequations}
    \label{eq:QJSq-data-processing-partTr}
    \begin{align}
        &\QJTq\left( \tr_{\sfB}(\rho_0), \tr_{\sfB}(\rho_1) \right)\\ 
        =~& \QJTq( \Phi_{\tr_{\sfB}}(\rho_0), \Phi_{\tr_{\sfB}}(\rho_1) )\\
        \leq~& \sum_{l \in [N_{\sfB}^2]} \frac{1}{N^2_{\sfB}} \QJTq\left( (I_{\sfA} \otimes U_l)\rho_0 (I_{\sfA} \otimes U_l)^{\dagger}, (I_{\sfA} \otimes U_l) \rho_1 (I_{\sfA} \otimes U_l)^{\dagger}  \right)\\
        =~& \sum_{l \in [N_{\sfB}^2]} \frac{1}{N^2_{\sfB}} \QJTq\left( \rho_0, \rho_1 \right)\\
        =~& \QJTq(\rho_0, \rho_1).
    \end{align}
    \end{subequations}

    Next, we move to the general case. Using the Stinespring dilation theorem (e.g.,~\cite[Theorem 2.25]{AS17}), for any quantum channel $\Phi$ on the registers $(\sfA,\sfB)$, we have the following representation with some unitary $U_{\Phi}$ on the registers $(\sfA, \sfB, \sfE)$ where $\dim(\calH_{\sfE}) \leq \dim(\calH_{\sfA\sfB})^2$: 
    \[ \Phi(\rho_{\sfA\sfB}) = \tr_{\sfE} \left( U_{\Phi} (\rho_{\sfA\sfB} \otimes \ketbra{\bar{0}}{\bar{0}}_{\sfE}) U_{\Phi}^{\dagger} \right). \]
    
    Consequently, we can obtain the following for any quantum channel $\Phi$: 
    \begin{align*}
        \QJTq(\Phi(\rho_0), \Phi(\rho_1)) &\leq \QJTq \left( U_{\Phi} (\rho_0 \otimes \ketbra{\bar{0}}{\bar{0}}_{\sfE}) U_{\Phi}^{\dagger}, U_{\Phi} (\rho_1 \otimes \ketbra{\bar{0}}{\bar{0}}_{\sfE}) U_{\Phi}^{\dagger}  \right)\\
        &= \QJTq \left( \rho_0 \otimes \ketbra{\bar{0}}{\bar{0}}_{\sfE}, \rho_1 \otimes \ketbra{\bar{0}}{\bar{0}}_{\sfE} \right)\\
        &= \QJTq (\rho_0,\rho_1).
    \end{align*}
    Here, the first line owes to \Cref{eq:QJSq-data-processing-partTr}, the second line is due to the unitary invariance of $\QJTq$ (\Cref{lemma:QJSq-unitary-invariance}), and the last line is because $\tr\big( (\rho_b \otimes \ket{\phi}\bra{\phi}_{\sfE})^q\big) = \tr(\rho_b^q)$ for any $b\in\binset$ and $q\in[1,2]$. We now complete the proof. 
\end{proof}

\subsection{Inequalities between the trace distance and \QJTq{}}
\label{subsec:QJTq-vs-td-inequalities}

We begin by establishing the lower bound on $\QJTq$ in terms of the trace distance, as stated in~\Cref{lemma:td-leq-QJSq}. The measured variant of the $q$-Jensen-Tsallis divergence ($\JTq$), denoted by $\measQJTq$, is derived from the definition provided in \Cref{eq:measured-f-divergences}.
\begin{lemma}[$\td \leq \QJTq$]
    \label{lemma:td-leq-QJSq}
    For any quantum states $\rho_0$ and $\rho_1$, we have the following inequality\emph{:}
    \[ \forall q\in[1,2],~\Hq\left(\frac{1}{2}\right) - \Hq\left( \frac{1}{2} - \frac{\td(\rho_0,\rho_1)}{2} \right) \leq \measQJTq(\rho_0,\rho_1) \leq \QJTq(\rho_0,\rho_1). \]
\end{lemma}

\begin{proof}
    The case of $q=1$ follows from~\cite[Theorem 14]{BH09}, and can also be derived by combining~\cite[Theorem 1]{FvdG99} with the Holevo bound (see~\cite[Lemma 2.15]{Liu23} for details). 
    
    Our focus will be on the cases where $1 < q \leq 2$. 
    We first prove the second inequality. 
    Let $\calM^*$ be an optimal POVM corresponding to $\measQJTq(\rho_0,\rho_1)$, then this POVM $\calM^*$ corresponds to a quantum-to-classical channel $\Phi_{\calM^*}(\rho) = \sum_{i=1}^N \ket{i}\bra{i} \tr(\rho M^*_i)$ (e.g.,~\cite[Equatin (2.41)]{AS17}). Using the data-processing inequality for $\QJTq$ (\Cref{lemma:QJSq-data-processing-inequaltiy}), for $1 < q \leq 2$, we obtain: 
    \[ \measQJTq(\rho_0,\rho_1) = \QJTq\left( \Phi_{\calM^*}(\rho_0), \Phi_{\calM^*}(\rho_1) \right) \leq \QJTq(\rho_0,\rho_1). \]

    Next, let us move to the first inequality. Let $p_b^{\calM}$ be the induced distribution with respect to the POVM $\calM$ of $\rho_b$ for any $b\in\binset$. Utilizing \Cref{lemma:TV-leq-JSq}, for $1 < q \leq 2$, we can derive that: 
    \begin{equation}
        \label{eq:measQJSq-geq-traceDist}
        \QJS^{\rm meas}_{q,\calM^*}(\rho_0,\rho_1) \geq \QJS^{\rm meas}_{q,\calM}(\rho_0,\rho_1) = \JTq\!\rbra*{ p_0^{\calM}, p_1^{\calM} } \geq 
            \Hq\!\rbra*{\frac{1}{2}} - \Hq\!\rbra*{\frac{1}{2}\!-\!\frac{\TV\left(p_0^{\calM}, p_1^{\calM} \right)}{2}}.
    \end{equation}
    We then consider the function $g(q;x)$ and its first derivative $\frac{\partial}{\partial x} g(q;x)$:
    \begin{align*}
        g(q;x) &\coloneqq 
        \Hq\left(\frac{1}{2}\right) - \Hq\left(\frac{1-x}{2}\right) = \frac{2^{-q}}{q-1}\left( (1+x)^q+(1-x)^q-2 \right),\\
        \frac{\partial}{\partial x} g(q;x) &= \frac{2^{-q}q}{q-1} \left( (1+x)^{q-1} - (1-x)^{q-1} \right).
    \end{align*}
    
    Since it is easy to see that $\frac{\partial}{\partial x} g(q;x) \geq 0$ for $0 \leq x \leq 1$ when $1 < q \leq 2$, we know that $g(q;x)$
    is monotonically increasing for $0 \leq x \leq 1$. 
    Noting that \Cref{eq:measQJSq-geq-traceDist} holds for arbitrary POVM $\calM$, and the trace distance is the measured version of the total variation distance (e.g.,~\cite[Theorem 9.1]{NC10}), we thus complete the proof by choosing the POVM that maximizes $\td(\rho_0,\rho_1)$.
\end{proof}

Next, we demonstrate the upper bound on $\QJTq$ in terms of the trace distance: 
\begin{lemma}[$\QJTq \leq \td$]
    \label{lemma:QJSq-leq-traceDistLike}
    For any quantum states $\rho_0$ and $\rho_1$, and any $1 \leq q \leq 2$, we have\emph{:} 
    \[ \QJTq(\rho_0,\rho_1) \leq \Hq\rbra*{\frac{1}{2}} \cdot \frac{1}{2} \tr\rbra*{ |\rho_0-\rho_1|^q } \leq \Hq\rbra*{\frac{1}{2}} \cdot \td(\rho_0,\rho_1)^q. \]
\end{lemma}

\begin{proof}
    We begin with the construction for establishing $\QJTq \leq \ln{2} \cdot \td$ for $q=1$ as in~\cite[Theorem 14]{BH09}. 
    Our analysis differs since we need to address the cases of $1 \leq q \leq 2$. Consider a single qutrit register $B$ with basis vectors $\ket{0},\ket{1},\ket{2}$. 
    Define $\hat{\rho}_0$ and $\hat{\rho}_1$ on $\calH\otimes \calB$ as below, where $\calB=\bbC^3$ is the Hilbert space corresponding to the register $\sfB$:
    \begin{align*}
        \hat{\rho}_0  \coloneqq & \frac{\rho_0+\rho_1-|\rho_0-\rho_1|}{2} \otimes \ketbra{2}{2} + \frac{\rho_0-\rho_1+|\rho_0-\rho_1|}{2} \otimes \ketbra{0}{0}  \coloneqq & \sigma_2\otimes\ketbra{2}{2} + \sigma_0\otimes\ketbra{0}{0},\\
        \hat{\rho}_1  \coloneqq & \frac{\rho_0+\rho_1-|\rho_0-\rho_1|}{2} \otimes \ketbra{2}{2} + \frac{\rho_1-\rho_0+|\rho_0-\rho_1|}{2} \otimes \ketbra{1}{1}  \coloneqq & \sigma_2\otimes\ketbra{2}{2} + \sigma_1\otimes\ketbra{1}{1}.
    \end{align*}

    Intuitively, $\sigma_b$ represents the case where $\rho_b$ is ``larger than'' $\rho_{1-b}$ for $b \in \binset$ (i.e., $\rho_0$ and $\rho_1$ are ``distinguishable''), while $\sigma_2$ represents the case where $\rho_0$ is ``indistinguishable'' from $\rho_1$. 
    This construction generalizes the proof of the classical analogs (e.g.,~\cite[Claim 4.4.2]{Vad99}). 

Noting that $\QJTq$ is contractive when applying a partial trace (\Cref{lemma:QJSq-data-processing-inequaltiy}), we obtain: 
\begin{subequations}
    \label{eq:QJSq-traceDistLike-RHS}
    \begin{align}
    \QJTq(\rho_0,\rho_1) &= \QJTq(\tr_{\sfB}(\hat{\rho}_0),\tr_{\sfB}(\hat{\rho}_1))\\
    &\leq \QJTq(\hat{\rho}_0,\hat{\rho}_1)\\
    &= \frac{1}{q-1} \left( \tr \left(\left( \frac{\hat{\rho}_0+\hat{\rho}_1}{2} \right)^q\right) - \frac{1}{2} \tr(\hat{\rho}_0^q) - \frac{1}{2} \tr(\hat{\rho}_1^q)  \right).
    \end{align}
\end{subequations}

Noting that $\sigma_0 \otimes \ketbra{0}{0}$, $\sigma_1 \otimes \ketbra{1}{1}$, and $\sigma_2 \otimes \ketbra{2}{2}$ are orthogonal to each other, we have: 
\begin{subequations}
    \label{eq:QJSq-traceDistLike-RHS-terms-i}
    \begin{align}
        \tr \left(\left( \frac{\hat{\rho}_0+\hat{\rho}_1}{2} \right)^q\right) &= \tr\left(\left( \sigma_2 \otimes \ketbra{2}{2} + \sum_{b\in\binset}\frac{\sigma_b}{2}\otimes \ketbra{b}{b} \right)^q\right) = \tr\left( \sigma_2^q + \frac{\sigma^q_0}{2^q}+ \frac{\sigma^q_1}{2^q} \right), 
    \end{align}
\end{subequations}
\begin{subequations}
    \label{eq:QJSq-traceDistLike-RHS-terms-ii}
    \begin{align}
        \forall b \in \binset,~\tr(\hat{\rho}_b^q) &= \tr\left(\left( \sigma_2\otimes \ketbra{2}{2} + \sigma_b \otimes \ketbra{b}{b} \right)^q\right) = \tr\left( \sigma_2^q + \sigma_b^q \right).
    \end{align}
\end{subequations}

Plugging \Cref{eq:QJSq-traceDistLike-RHS-terms-i,eq:QJSq-traceDistLike-RHS-terms-ii} and the equality $\Hq\rbra*{\frac{1}{2}} = \frac{1-2^{1-q}}{q-1}$ into \Cref{eq:QJSq-traceDistLike-RHS}, we obtain: 
   \begin{subequations}
    \label{eq:QJTq-upperBound-Lq}
    \begin{align}
       \QJTq(\rho_0,\rho_1) 
       &\leq \Hq\rbra*{\frac{1}{2}} \cdot \frac{1}{2} \tr\rbra*{\sigma^q_0+\sigma^q_1}\\
       &\leq \Hq\rbra*{\frac{1}{2}} \cdot \frac{1}{2} \rbra*{ \tr(\sigma_0)^q + \tr(\sigma_1)^q }\\
       &= \Hq\rbra*{\frac{1}{2}} \cdot \td(\rho_0,\rho_1)^q.
    \end{align}
   \end{subequations}
   Here, the second line is due to the monotonicity of the Schatten $p$-norm (e.g.,~\cite[Equation (1.31)]{AS17}), equivalently, $\tr\rbra*{M^q} \leq \tr(M)^q$ for any positive semi-definite matrix $M$ and $q \geq 1$. The last line owes to the fact that: 
   \[\forall b\in\binset,\quad\tr(\sigma_b) = (-1)^b \tr\Big(\frac{\rho_0-\rho_1}{2}\Big) + \frac{1}{2}\tr\rbra*{|\rho_0-\rho_1|} = \frac{1}{2}\tr\rbra*{|\rho_0-\rho_1|}.\]
   
   Lastly, since $\sigma_0$ and $\sigma_1$ are orthogonal to each other, we complete the proof by plugging the equality $\tr\rbra*{\sigma^q_0+\sigma^q_1} = \tr\rbra*{\rbra*{\sigma_0+\sigma_1}^q} = \tr\rbra*{ |\rho_0-\rho_1|^q}$ into the first line in \Cref{eq:QJTq-upperBound-Lq}.  
\end{proof}

\subsection{Bounds for the Tsallis binary entropy}
\label{subsec:Tsallis-binary-entropy-bounds}

In this subsection, we establish lower and upper bounds (\Cref{lemma:Tsallis-binary-entropy-lower-bound,lemma:Tsallis-binary-entropy-upper-bound}, respectively) for the Tsallis binary entropy, which are useful when applying the lower bound on \QJTq{} in terms of the trace distance (\Cref{lemma:td-leq-QJSq}). 

We begin by proving an lower bound, which extends the bound $\H(1/2) \cdot 4x(1-x) \leq \H(x)$ from the specific case of $q=1$, as stated in~\cite[Theorem 1.2]{Topsoe01}, to a broader range of $q$:
\begin{lemma}[Tsallis binary entropy lower bound]
    \label{lemma:Tsallis-binary-entropy-lower-bound}
    For any $p=(x,1-x)$, let $\Hq(x)$ denote the Tsallis binary entropy with $q \in [0,2] \cup [3,+\infty)$, we have\emph{:}
    \[ \Hq(1/2) \cdot 4x(1-x) \leq \Hq(x).\]
\end{lemma}

\begin{proof}
    We need only consider the cases where $q\in\calI \coloneqq [0,1) \cup (1,2] \cup [3,+\infty)$, as the case $q=1$ directly follows from~\cite[Theorem 1.2]{Topsoe01}.  
    Our proof strategy is inspired by the approach used in that theorem. 
    We start by defining functions $F(q;x)$ and $G(q;x)$ on $0 \leq x \leq 1$ and $q\in\calI$: 
    \[ F(q;x) \coloneqq \frac{\Hq(x)}{x(1-x)} = \frac{1-x^q-(1-x)^q}{(q-1)x(1-x)} \text{ and } G(q;x) \coloneqq \frac{x^{q-1}-1}{(q-1)(x-1)}. \]

    It is evident that $F(q;x) = F(q;1-x)$ for all $x\in[0,1]$ and that $F(q;1/2) = 4\Hq(1/2)$. For the endpoints $x=0$ and $x=1$, one can verify the following evaluation holds:\footnote{For $q\in \calI\setminus \cbra{0}$, this evaluation follows from an application of the L'H\^opital's rule, which implies that $\lim_{x\to 0^+} F(q;x) = \lim_{x\to 0^+} \frac{-q x^{q-1} + q (1-x)^{q-1}}{(q-1)(1-2x)} = \frac{q}{q-1} \cdot \lim_{x\to 0^+} F_\star(q;x)$, where $F_\star(q;x)\coloneqq \frac{(1-x)^{q-1}-x^{q-1}}{1-2x}$. The desired limit then follows from the facts that when $q\in(0,1)$, one has $q/(q-1)<0$ and $\lim_{x\to 0^+} F_\star(q;x) = -\infty$, and that when $q>1$, it holds that $\lim_{x\to 0^+} F_\star(q;x)=1$. The case of $q=0$ holds straightforwardly.} 
     \[F(q,1)=F(q,0) =\lim_{x\to 0^+}F(q,x) = \begin{cases} \infty, &\text{if } q\in[0,1)\\ q/(q-1), &\text{if } q\in \calI\setminus [0,1) \end{cases}.\]
    
    We then assume that $G(q;x)$ is convex on $x\in[0,1]$ for any fixed $q\in\calI$:
    \begin{equation}
        \label{eq:Tsallis-binary-entropy-convexity}
        \text{For any } x\in[0,1] \text{ and } q\in(1,2],~ \frac{\partial^2 G(q;x)}{\partial x^2}  = \frac{(q\!-\!2)x^{q-3}}{x\!-\!1} - \frac{2x^{q-2}}{(x\!-\!1)^2} + \frac{2(x^{q-1}\!-\!1)}{(x\!-\!1)^3(q\!-\!1)} \geq 0.
    \end{equation}
    
    Since $F(q;x) = G(q;x) + G(q;1-x)$, \Cref{eq:Tsallis-binary-entropy-convexity} implies that $F(q,x)$ is convex on $x\in[0,1]$ for any fixed $q\in\calI$. We can obtain that: for any $q\in\calI$, $F(q;x)$ is monotonically decreasing on $x\in(0,1/2)$ and monotonically increasing on $x\in(1/2,1)$. Consequently, we establish the lower bound by noticing that: 
    \[\text{For any } x\in[0,1] \text{ and } q\in\calI, F(q;x) \geq F(q;1/2) = 4\Hq(1/2).\]

    It remains to prove \Cref{eq:Tsallis-binary-entropy-convexity}. 
    Noting that $(x-1)^3 \leq 0$ for any $0 \leq x \leq 1$, \Cref{eq:Tsallis-binary-entropy-convexity} holds if and only if the following holds: 
    \[ f(q;x) \coloneqq (q-2)(x-1)^2x^{q-3} - 2(x-1)x^{q-2} + \frac{2(x^{q-1}-1)}{q-1} \leq 0. \]
    
    A direct calculation implies that $\frac{\partial}{\partial x} f(q;x) = (q-3)(q-2)(x-1)^2 x^{q-4} \geq 0$ for any $q\in\calI$ and $x\in[0,1]$ since $\calI \cap (2,3) = \emptyset$. Hence, for any fixed $q\in\calI$, $f(q;x)$ is monotonically increasing for any $x \in (0,1)$. Therefore, we complete the proof by concluding that 
    \[\max_{x\in[0,1]} f(q;x) \leq f(q,1) = 0.  \qedhere\]
\end{proof}

Next, we will demonstrate an upper bound for the range of $1 < q \leq 2$ that is weaker than the best known upper bound for the case of $q=1$ as shown in~\cite[Theorem 1.2]{Topsoe01}:\footnote{Numerical evidence suggests that \Cref{lemma:Tsallis-binary-entropy-upper-bound} can be improved to $\Hq(x) \leq \Hq\big(\frac{1}{2}\big) \cdot (4x(1-x))^{\frac{1}{2\Hq(1/2)}}$ for any $0 \leq x \leq 1$ and $1 \leq q \leq 2$, which coincides with the bound $\H(x) \leq \H(1/2) (4x(1-x))^{\frac{1}{2\H(1/2)}}$ in~\cite{Topsoe01}.}

\begin{lemma}[Tsallis binary entropy upper bound]
    \label{lemma:Tsallis-binary-entropy-upper-bound}
    For any $p=(x,1-x)$, let $\Hq(x)$ denote the Tsallis binary entropy with $1\leq q \leq 2$, we have\emph{:} 
    \[\Hq(x) \leq \Hq(1/2) \cdot \smash{(4x(1-x))^{1/2}}.\]
\end{lemma}

\begin{proof}
    The case of $q=1$ follows directly from~\cite[Theorem 8]{Lin91}, it remains to address the range $1 < q \leq 2$. 
    We will establish the bound separately for $x\in \Iin$ and $x \in \Iout$, where $\Iin \cup \Iout = [0,1]$. Specifically, these intervals are defined as 
    $\Iin \coloneqq [1/10,9/10]$ and $\Iout \coloneqq [0,1/8] \cup [7/8,1]$. 

    \paragraph{The outer interval case.} We start with the case of $x\in\Iout$. Since $\Hq(x)=\Hq(1-x)$ for any $0 \leq x \leq 1$, it is sufficient to consider the case of $0 \leq x \leq 1/8$. Noting that $q-1 \geq 0$, it suffices to show that: For any $0 \leq x \leq 1/8$ and $1 < q \leq 2$,
    \begin{equation}
        \label{eq:Tsallis-binary-entropy-upper-bound-outer}
        (q-1) \left( \Hq\Big(\frac{1}{2}\Big) \sqrt{4x(1-x)} - \Hq(x) \right) =
        \left( 2-2^{2-q} \right) \sqrt{x(1-x)} - \left( 1-x^q-(1-x)^q \right) \geq 0.
    \end{equation}

    Leveraging the Taylor expansion of $1-(1-x)^q$ at $x=0$, we obtain that:
    \begin{equation}
        \label{eq:Tsallis-binary-entropy-upper-bound-outer-Taylor}
        1-(1-x)^q  = \sum_{k=1}^{\infty} \frac{{\rbra{-1}}^{k+1}}{k!} \prod_{r=0}^{k-1} (q-r) x^k \coloneqq \sum_{k=1}^{\infty} \alpha_k x^k \leq qx.
    \end{equation}
    Here, notice that $1< q \leq 2$, the last inequality owes to the fact that $\alpha_1 = q>0$ and $\alpha_k \leq 0$ for all integer $k\geq 2$.  
    Plugging \Cref{eq:Tsallis-binary-entropy-upper-bound-outer-Taylor} into \Cref{eq:Tsallis-binary-entropy-upper-bound-outer}, it remains to prove that: 
    \begin{equation*}
        F_1(q;x) \coloneqq \frac{-x^q+qx}{\sqrt{x(1-x)}} \leq 2-2^{2-q}.
    \end{equation*}
    
    A direct calculation implies that $F_1(q;1/8) = (q-2^{3-3q})/\sqrt{7}$ satisfies $2-2^{2-q} - F_1(q;1/8) > 0$ for $1< q\leq 2$.\footnote{It is noteworthy that $2-2^{2-q}-F_1(q;1/2) = 2-2^{1-q}-q < 0$ for $1< q \leq 2$, and consequently, the outer interval case is not enough to establish our Tsallis binary entropy upper bound for any $0 \leq x \leq 1$.} 
    As a consequence, it is enough to show that $F_1(q,x)$ is monotonically non-decreasing on $x\in[0,1/8]$ for any fixed $q\in(1,2]$, specifically:
    \begin{equation}
        \label{eq:Tsallis-binary-entropy-upper-bound-parital1st}
        \frac{\partial}{\partial x} F_1(q;x) = \frac{1}{2} (x(1-x))^{-3/2} \left( qx + (1+2q(x-1)-2x) x^q\right) \geq 0.
    \end{equation}
    
    Noting that $\frac{1}{2} (x(1-x))^{-3/2} \geq 0$, \Cref{eq:Tsallis-binary-entropy-upper-bound-parital1st} holds if and only if the following holds:
    \begin{equation*}
        F_2(q;x) \coloneqq \left( 2(1-q)x + 2q - 1 \right) x^{q-1} \leq q.
    \end{equation*}
    
    A direct calculation implies that $F_2(q;1/8) = 2^{1-3q}(7q-3)$ satisfies that $q-F_2(q;1/8) > 0$ for $1< q \leq 2$. 
    Consequently, it suffices to show that $F_2(q;x)$ is monotonically non-decreasing on $x\in[0,1/8]$ for any fixed $q\in(1,2]$, particularly:
    \begin{equation}
        \label{eq:Tsallis-binary-entropy-upper-bound-parital2nd}
        \frac{\partial}{\partial x} F_2(q;x) = x^{q-2} (q-1) (2q(1-x)-1) \geq 0.
    \end{equation}
    
    Since $x^{q-2} (q-1) > 0$ for any $q\in(1,2]$ and $x\in[0,1/8]$, \Cref{eq:Tsallis-binary-entropy-upper-bound-parital2nd} holds if and only if $F_3(q;x) \coloneqq 2q(1-x)-1 \geq 0$. It is evident that $F_3(q;x) \geq 0$ is equivalent to $x \leq 1-1/2q < 1/2$ for $1< q\leq 2$, and thus we complete the proof of the outer interval case.  

    \paragraph{The inner interval case.} Next, we move to the case of $x\in\Iin$. Let $x = (1+t)/2$, then it suffices to consider the case of $0 \leq t \leq 1$ since $\Hq(x)=\Hq(1-x)$ for any $0\leq x \leq 1$. Noting that $2^q/(q-1) >0$ for $1<q\leq 2$, it suffices to show that: For any $0 \leq t \leq 4/5$ and $1 < q \leq 2$,
    \begin{equation}
        \label{eq:Tsallis-binary-entropy-upper-bound-inner}
        2^q (q-1) \left( \Hq\Big(\frac{1}{2}\Big) \sqrt{4x(1-x)} - \Hq(x) \right) = (1-t)^q + (1+t)^q - \left( 2^q + (2-2^q) \sqrt{1-t^2} \right) \geq 0. 
    \end{equation}

    Utilizing the Taylor expansion of $(1-t)^q + (1+t)^q$ at $t=0$, we obtain that: 
    \begin{equation}
        \label{eq:Tsallis-binary-entropy-upper-bound-inner-Taylor}
        (1-t)^q+(1+t)^q = \sum_{k=0}^{\infty} \frac{2}{(2k)!} \prod_{r=0}^{2k-1} (q-r) t^{2k} \coloneqq \sum_{k=0}^{\infty} \beta_k t^{2k} \geq \beta_0+\beta_1^2 = 2+q(q-1)t^2.
    \end{equation}
    Here, the last inequality is because $\beta_k \geq 0$ for all integer $k\geq 0$. Substituting \Cref{eq:Tsallis-binary-entropy-upper-bound-inner-Taylor} into \Cref{eq:Tsallis-binary-entropy-upper-bound-inner}, it remains to show that: 
    \begin{equation}
        \label{eq:Tsallis-binary-entropy-upper-bound-inner-condition}
        2+q(q-1)t^2 \geq 2^q + (2-2^q) \sqrt{1-t^2}.
    \end{equation}

    Let $u\coloneqq \sqrt{1-t^2} \in [3/5,1]$. Then \Cref{eq:Tsallis-binary-entropy-upper-bound-inner-condition} can be rewritten as
    \[ G(q;u) \leq 0 \quad\text{where}\quad G(q;u) \coloneqq \frac{q(q-1)}{2^q-2} u^2 - u + 1 - \frac{q(q-1)}{2^q-2}. \]
    Observe that $G(q;u)$ is a quadratic function in $u$ with positive leading coefficient for $q\in(1,2]$. Consequently, it suffices to verify the inequality at the endpoints $u=3/5$ and $u=1$ when $1<q\leq 2$. A direct calculation shows that $G(q;1)=0$, and that $G(q;3/5) \leq 0$ is equivalent to 
    \[G_\star(q) \coloneqq 8q(q-1)-5 (2^q-2) \geq 0.\] Differentiating $G_\star$ yields 
    $\frac{\dd}{\dd q} G_\star(q) = 16q - 8 -5 \cdot 2^q \ln{2}$ and $\frac{\dd^2}{\dd q^2} G_\star(q) = 16 - 5 \cdot 2^q (\ln{2})^2$.
    Since $\frac{\dd^2}{\dd q^2} G_\star(q)$ is monotonically decreasing on $(1,2]$ and $\ln{2} \leq 5/6$, we obtain $\frac{\dd^2}{\dd q^2} G_\star(q) \geq \frac{\dd^2}{\dd q^2} G_\star(q)|_{q=2} = 16-20 (\ln{2})^2 \geq 19/9 > 0$. Therefore, $\frac{\dd}{\dd q} G_\star(q)$ is monotonically increasing on $(1,2]$, and 
    $\frac{\dd}{\dd q} G_\star(q) \geq \frac{\dd}{\dd q} G_\star(q)|_{q=1} = 8 - 10 \ln{2} \geq 1$. It follows that that $G_\star(q)$ is monotonically increasing on $q\in(1,2]$. Since $G_\star(1) = 0$, we conclude that $G_\star(q) > 0$ for all $q\in(1,2]$, thereby establishing \Cref{eq:Tsallis-binary-entropy-upper-bound-inner-condition}.
\end{proof}

\subsection{Useful bounds on Tsallis entropy}

In this subsection, we present a useful bound on Tsallis entropy. \Cref{lemma:inequality-uniformTV-TsallisEA} establishes inequalities between the Tsallis entropy of a distribution $p$ and the total variation distance between $p$ and the uniform distribution of the same dimension. 

% 0<q<1 seems also hold
\begin{lemma}[Tsallis entropy bounds by closeness to uniform distribution]
    \label{lemma:inequality-uniformTV-TsallisEA}
    Let $p$ be a probability distribution over $[N]$ with $N\geq 2$, and let $\nu$ be the uniform distribution over $[N]$. 
    Then, for any $q > 1$ and $0 \leq \TV(p,\nu) \leq 1-1/N$, it holds that\emph{:}
    \[\left(1-\TV(p,\nu) - 1/N \right) \ln_q(N) \leq \Hq(p).\]
    Moreover, for any $q > 1$ and $N$ satisfying $1/q \leq \TV(p,\nu) \leq 1-1/N$, it holds that\emph{:}
    \[\Hq(p) \leq \ln_q\!\big( N (1-\TV(p,\nu)) \big).\]
\end{lemma}

\begin{proof}
    Let $\gamma \coloneqq \TV(p,\nu)$, and let $\Delta_{N}$ be the set of probability distributions of dimension $N$. It is evident that $0 \leq \TV(p,\nu) \leq 1-1/N$. 
    To establish the lower bound, it suffices to minimize the Tsallis entropy $\Hq(p)$ subject to the constraint $\TV(p,\nu) = \gamma$, which is equivalent to solve the convex optimization problem in \Cref{eq:Tsallis-entropy-min}.\footnote{A similar formulation also appeared in the proof of~\cite[Lemma 16]{KLGN19}.}
    
    {\centering
    \begin{minipage}{0.5\textwidth}
        \centering
        \begin{mini}{}{\Hq(p')}{}{}{\label{eq:Tsallis-entropy-min}}{}
            \addConstraint{p'}{\in\Delta_N}{}
            \addConstraint{\TV(p',\nu)}{\leq \gamma}{}
        \end{mini}
    \end{minipage}%
    \begin{minipage}{0.5\textwidth}
    \begin{subequations}
    \label{eq:Tsallis-entropy-min-solution}
    \begin{align}
        &p_{\min}(i) = \begin{cases}
            \frac{1}{N}, &\text{if } i\in[k_{\min}]\\
            \frac{1}{N}+\gamma, &\text{if } i=k_{\min}+1\\
            \frac{\varepsilon}{N}, &\text{if } i=k_{\min}+2\\
            0, &\text{otherwise}
        \end{cases},\\
        &\begin{aligned}
            &\text{where } &k_{\min} &\coloneqq \floor*{N(1-\gamma)}-1,\\
            &&\varepsilon &\coloneqq N(1\!-\!\gamma) - \floor*{N(1\!-\!\gamma)}.
        \end{aligned}
    \end{align}
    \end{subequations}
    \end{minipage}  
    \vspace{1em}
    }
    
    Note that $\Hq(p)$ is concave (\Cref{lemma:Tsallis-entropy-properties}) for any fixed $q>1$, and the constraints in \Cref{eq:Tsallis-entropy-min} form a closed convex set. Since the minimum of a concave function is attained at some extreme point (e.g.,~\cite[Corollary 32.3.1]{Rockafellar70}) and the Tsallis entropy is permutation-invariant, we deduce an optimal solution $p_{\min}$ to \Cref{eq:Tsallis-entropy-min}, as stated in \Cref{eq:Tsallis-entropy-min-solution}. 

    Next, we can deduce the lower bound of the Tsallis entropy by evaluating $\Hq(p_{\min})$: 
    \begin{align*}
        \Hq(p_{\min}) &= \frac{1}{q-1} \left( 1 - k_{\min} \Big( \frac{1}{N} \Big)^q - \Big(\frac{1}{N} + \gamma\Big)^q - \Big( \frac{\varepsilon}{N} \Big)^q \right)\\
        &\geq \frac{1}{q-1} \left( \frac{(\lfloor N(1-\gamma)\rfloor-1)}{N} + \frac{\varepsilon}{N}  - (\floor*{N(1-\gamma)}-1 + \varepsilon^q) \Big( \frac{1}{N} \Big)^q\right)\\ 
        &\geq \frac{1}{q-1} \left( 1-\gamma-\frac{1}{N}  - \Big(1-\gamma -\frac{1}{N}\Big) \Big( \frac{1}{N} \Big)^{q-1} \right)\\ 
        &= \left(1-\gamma - \frac{1}{N} \right) \ln_q(N).
    \end{align*}
    Here, the second line excludes terms corresponding to $p_{\min}(k_{\min}+1)$,\footnote{We exclude the terms $\rbra[\big]{\frac{1}{N}+\gamma}$ and $-\rbra[\big]{\frac{1}{N}+\gamma}^q$, whose combined contribution is non-negative for $q\geq 1$. Moreover, the term $\frac{1}{N}+\gamma$ is implicitly accounted for in the term $1$, since $\sum_i p_{\min}(i) = \frac{k_{\min}}{N} + \rbra[\big]{\frac{1}{N}+\gamma} + \frac{\epsilon}{N} = 1$ according to \Cref{eq:Tsallis-entropy-min-solution}.} and the third line follows from the fact that $\varepsilon^q \leq \varepsilon$ for $q\geq 1$ and $0\leq \varepsilon \leq 1$.

    \vspace{1em}
    To demonstrate the upper bound, it remains to maximize the Tsallis entropy $\Hq(p)$ subject to the constraint $\TV(p,\nu) = \gamma$, which is equivalent to solve a \textit{non-convex} optimization problem analogous to \Cref{eq:Tsallis-entropy-min}. This task is challenging in general, but we consider only the regime $\TV(p,\nu) \geq 1/q$.\footnote{For the regime $0 \leq \TV(p,\nu) \leq 1/q$, the optimal solution to \Cref{eq:Tsallis-entropy-max} depends on the choice of $q$.}. Particularly, we focus on the following optimization problem:
    
    \begin{maxi}{}{\Hq(p')}{}{}{\label{eq:Tsallis-entropy-max}}{}
        \addConstraint{p'}{\in\Delta_N}{}
        \addConstraint{\TV(p',\nu)}{\geq \gamma \geq 1/q}{}
    \end{maxi}

    It is not too hard to obtain an optimal solution $p_{\max}$ to \Cref{eq:Tsallis-entropy-max}, where $\varepsilon$ is defined as in \Cref{eq:Tsallis-entropy-min-solution}, as stated in \Cref{prop:Tsallis-entropy-max-solution}. The proof is deferred in \Cref{subsec:omitted-proofs-Tsallis-properties}. 
    \begin{restatable}{proposition}{TsallisEntropyUpperBoundSolution}
        \label{prop:Tsallis-entropy-max-solution}
        For the optimization problem presented in \Cref{eq:Tsallis-entropy-max}, an optimal solution is the distribution provided in \Cref{eq:Tsallis-entropy-max-solution}, where $\varepsilon = N(1-\gamma) - \floor*{N(1-\gamma)}$: 
        \begin{equation}
        \label{eq:Tsallis-entropy-max-solution}
            p_{\max}(i) = \begin{cases}
            \frac{1}{N}+\frac{\gamma}{k_{\max}}, &\text{if } i\in[k_{\max}]\\
            \frac{\varepsilon}{N (N-k_{\max})}, &\text{otherwise}
            \end{cases},
            \text{where } k_{\max} \coloneqq \floor*{N(1-\gamma)}.             
        \end{equation}
    \end{restatable}
    
    Consequently, we can derive the upper bound of the Tsallis entropy by evaluating $\Hq(p_{\max})$:
    \begin{align*}
        \Hq(p_{\max}) &= \frac{1}{q-1} \left( 1- k_{\max} \Big( \frac{1}{N} + \frac{\gamma}{k_{\max}} \Big)^q - (N-k_{\max}) \Big( \frac{\varepsilon}{N(N-k_{\max})} \Big)^q \right)\\
        &= \frac{1}{q-1} \left( 1- \Big(1 -\frac{\varepsilon}{N}\Big)^q \Big(\frac{1}{N(1-\gamma) - \varepsilon}\Big)^{q-1} - \Big( \frac{\varepsilon}{N} \Big)^q \Big( \frac{1}{N\gamma + \varepsilon} \Big)^{q-1} \right)\\
        & \leq \frac{1}{q-1} \left( 1- \Big(\frac{1}{(N(1-\gamma)}\Big)^{q-1} \right)\\
        & = \ln_q\!\Big(N(1-\gamma)\Big).
    \end{align*}
    Let $F(q;N,\varepsilon,\gamma) \coloneq \rbra*{1-\frac{\varepsilon}{N}}^q (N(1-\gamma) - \varepsilon)^{1-q} + \rbra*{ \frac{\varepsilon}{N}}^q (N\gamma + \varepsilon)^{1-q}$, then the third line holds by assuming that $F(q;N,\varepsilon,\gamma)$ is monotonically non-decreasing on $0 \leq \varepsilon \leq 1$ for any fixed $\gamma$, $q$, and $N$ satisfying $q \gamma \geq 1$ and $N \geq q/(q-1)$. 
    
    It remains to prove $\frac{\partial}{\partial \varepsilon} F(q;N,\varepsilon,\gamma) \geq 0$ the aforementioned range of $x$, $\gamma$, $q$, and $N$. By a direct calculation, we complete the proof by noticing all terms in the following are non-negative:
    \[ \frac{\partial}{\partial \varepsilon} F(q;N,\varepsilon,\gamma) = \rbra*{1-\frac{\epsilon }{N}}^q  \frac{(N (\gamma  q-1)+\epsilon )}{ (N-\epsilon) \rbra*{N(1-\gamma)-\epsilon }^{q}}  + \rbra*{\frac{\epsilon }{N}}^q \frac{(\gamma  N q+\epsilon )   }{\epsilon (\gamma  N+\epsilon )^{q} } \geq 0. \qedhere\]
\end{proof}

\section{Hardness and lower bounds via \QJTq{}-based reductions}
\label{sec:hardness-via-QJTq-reductions}

In this section, we will establish reductions from the closeness testing of quantum states via the trace distance to testing via the quantum $q$-Tsallis entropy difference. Our proof crucially depends on the properties of the quantum Jensen-Tsallis divergence (\QJTq{}) demonstrated in \Cref{sec:properties-QJTq}. 
Using these reductions, we will prove computational hardness results and query complexity lower bounds for several problems related to the quantum $q$-Tsallis entropy difference under various circumstances. 

\vspace{1em}
We begin by defining the \textsc{Quantum $q$-Tsallis Entropy Difference} and the \textsc{Quantum $q$-Tsallis Entropy Approximation}, denoted by $\TsallisQED[g(n)]$ and $\TsallisQEA[t(n),g(n)]$, respectively. These definitions generalize the counterpart definitions in~\cite{BASTS10} from the von Neumann entropy (i.e., $\QJTq$ with $q=1$) to the quantum $q$-Tsallis entropy for $1 \leq q \leq 2$. 

\begin{definition}[Quantum $q$-Tsallis Entropy Difference, \TsallisQED{}]
	\label{def:TsallisQED}
    Let $Q_0$ and $Q_1$ be quantum circuits acting on $m$ qubits and having $n$ specified output qubits, where $m(n)$ is a polynomial in $n$. Let $\rho_i$ be the quantum state obtained by running $Q_i$ on $\ket{0}^{\otimes m}$ and tracing out the non-output qubits. Let $g(n)$ be a positive efficiently computable function. Decide whether\emph{:}
    \begin{itemize}[topsep=0.33em, itemsep=0.33em, parsep=0.33em]
	   \item \emph{Yes:} The pair of quantum circuits $(Q_0,Q_1)$ such that $\Sq(\rho_0)-\Sq(\rho_1) \geq g(n)$;
	   \item \emph{No:} The pair of quantum circuits $(Q_0,Q_1)$ such that $\Sq(\rho_1)-\Sq(\rho_0) \geq g(n)$.
    \end{itemize}
\end{definition}

\begin{definition}[Quantum $q$-Tsallis Entropy Approximation, \TsallisQEA{}]
	\label{def:TsallisQEA}
    Let $Q$ be a quantum circuit acting on $m$ qubits and having $n$ specified output qubits, where $m(n)$ is a polynomial in $n$. Let $\rho$ be the quantum state obtained by running $Q$ on $\ket{0}^{\otimes m}$ and tracing out the non-output qubits. Let $g(n)$ and $t(n)$ be positive efficiently computable functions. Decide whether\emph{:}
    \begin{itemize}[topsep=0.33em, itemsep=0.33em, parsep=0.33em]
	   \item \emph{Yes:} The quantum circuit $Q$ such that $\Sq(\rho) \geq t(n) + g(n)$;
	   \item \emph{No:} The quantum circuit $Q$ such that $\Sq(\rho) \leq t(n) - g(n)$.
    \end{itemize}
\end{definition}

Notably, the quantum $q$-Tsallis entropy of any pure state is zero. Hence, similar to \Cref{subsec:state-closeness-testing}, it is reasonable to define \textit{constant-rank} variants of \TsallisQED{} and \TsallisQEA{}:
\begin{enumerate}[label={\upshape(\arabic*)}, topsep=0.33em, itemsep=0.33em, parsep=0.33em]
	\item \ConstRankTsallisQED{}: the ranks of $\rho_0$ and $\rho_1$ are at most $O(1)$. 
	\item \ConstRankTsallisQEA{}: the rank of $\rho$ is at most $O(1)$. 
\end{enumerate}

\vspace{1em}
Next, we present the main theorem in this section: 

\begin{theorem}[Computational hardness for \TsallisQED{} and \TsallisQEA{}]
    \label{thm:TsallisQE-hardness}
    The promise problems \TsallisQED{} and \TsallisQEA{} capture the computational power of their respective complexity classes in the corresponding regimes of $q$\emph{:}
    \begin{enumerate}[label={\upshape(\arabic*)}, topsep=0.33em, itemsep=0.33em, parsep=0.33em]
        \item For any $q \in [1,2]$ and $n \geq 3$, it holds that\emph{:} For $1/\poly(n) \leq g_q(n) \leq 2^q\Hq(1/2) \rbra[\big]{1-2^{ -\frac{n}{2}+\frac{7}{5} }}$,  $\ConstRankTsallisQED[g_q(n)]$ is \BQP{}-hard under Karp reduction. Consequently, \ConstRankTsallisQEA{} with $g(n) = \Theta(1)$ is \BQP{}-hard under Turing reduction. \label{thmitem:ConstRankTsallisQED-BQPhard}
        \item For any $q\in \big(1, 1+\frac{1}{n-1}\big]$ and $n \geq 90$, it holds that\emph{:} For $1/\poly(n) \leq g(n) \leq 1/400$, $\TsallisQED[g(n)]$ is \QSZK{}-hard under Karp reduction. Consequently, \TsallisQEA{} with $g(n) = \Theta(1)$ is \QSZK{}-hard under Turing reduction. \label{thmitem:TsallisQED-QSZKhard}   
        \item For any $n\geq 5$, it holds that\emph{:} For $1/\poly(n) \leq g(n) \leq 1/150$, $\TsallisQEAnoq_{1+\frac{1}{n-1}}$ with $g(n)$ is \NIQSZK{}-hard. \label{thmitem:TsallisQEA-NIQSZKhard}   
    \end{enumerate}
\end{theorem}

In particular, \Cref{thm:TsallisQE-hardness}\ref{thmitem:ConstRankTsallisQED-BQPhard} is derived from the pure-state reduction (\Cref{lemma:reduction-PureQSD-ConstRankTsallisQED}), and the detailed statements are \Cref{thm:ConstRankTsallisQED-BQPhard,thm:ConstRankTsallisQEA-BQPhard}. 
Moreover, \Cref{thm:TsallisQE-hardness}\ref{thmitem:TsallisQED-QSZKhard} is obtained through a mixed-state reduction (\Cref{lemma:reduction-QSD-TsallisQED}), and the detailed statements are \Cref{thm:TsallisQED-QSZK-hardness,thm:TsallisQEA-QSZKhard}. 
Furthermore, \Cref{thm:TsallisQE-hardness}\ref{thmitem:TsallisQEA-NIQSZKhard} follows from a tailor-made mixed state reduction for \QSCMM{} (\Cref{lemma:reduction-MaxMixedQSD-TsallisQEA}), and the detailed statement is \Cref{thm:TsallisQEA-NIQSZKhard}.  

Lastly, using the reductions in \Cref{lemma:reduction-QSD-TsallisQED}, we derive lower bounds on the quantum query and sample complexity for estimating $\Sq(\rho)$ where $1 < q \leq 1\!+\!\frac{1}{n-1}$, as presented in \Cref{thm:Tsallis-query-complexity-lower-bound-smallQ,thm:Tsallis-sample-complexity-lower-bound-smallQ}. These theorems build on prior works in quantum query complexity~\cite{CFMdW10} and sample complexity~\cite{OW21} lower bounds for the trace distance. In addition, we provide quantum query and sample complexity lower bounds for estimating $\Sq(\rho)$ when $q \geq 1\!+\!\Omega(1)$, leveraging the hard instances from~\cite{Belovs19}, as detailed in \Cref{thm:Tsallis-query-complexity-lower-bound-largeQ,thm:Tsallis-sample-complexity-lower-bound-largeQ}. 

\subsection{Pure-state reduction: \PureQSD{} \texorpdfstring{$\leq$}{≤} \ConstRankTsallisQED{} for \texorpdfstring{$1\leq q \leq 2$}{1≤q≤2}}
\label{subsec:pure-state-reduction}

The reduction in \Cref{lemma:reduction-PureQSD-ConstRankTsallisQED} is from the trace distance between two $n$-qubit pure states (\PureQSD{}) to the quantum $q$-Tsallis entropy difference between two new constant-rank $(n+1)$-qubit states (\ConstRankTsallisQED{}), for $1 \leq q \leq 2$.

\begin{lemma}[$\PureQSD\leq\ConstRankTsallisQED$]
    \label{lemma:reduction-PureQSD-ConstRankTsallisQED}
    Let $Q_0$ and $Q_1$ be quantum circuits acting on $n$ qubits and having the same number of output qubits. Let $\ket{\psi_i}$ be the quantum state obtained by running $Q_i$ on $\ket{0}^{\otimes n}$. For any $b \in \binset$, there is a new quantum circuit $Q'_b$ acting on $n+3$ qubits, using $O\rbra{1}$ queries to controlled-$Q_0$ and controlled-$Q_1$, as well as $O(1)$ one- and two-qubit gates. The circuit $Q'_b$ prepares a new quantum state $\rho'_b$, which has constant rank and acts on $n' \coloneqq n+1$ qubits, such that for any efficiently computable functions $\alpha(n)$ and $\beta(n)$, where $\beta(n) + \sqrt{1-\alpha(n)^2} < 1$, and any $q\in[1,2]$, the following holds\emph{:}
    \begin{align*}
        \td(\ketbra{\psi_0}{\psi_0}, \ketbra{\psi_1}{\psi_1}) \geq \alpha(n)  &~\Rightarrow~  \Sq(\rho'_0) - \Sq(\rho'_1) \geq g_q(n') = g_q(n+1),\\
        \td(\ketbra{\psi_0}{\psi_0}, \ketbra{\psi_1}{\psi_1}) \leq \beta(n) &~\Rightarrow~ \Sq(\rho'_1) - \Sq(\rho'_0) \geq g_q(n') = g_q(n+1),
    \end{align*}
    where $g_q(n+1) \coloneqq 2^{-q} \cdot \Hq(1/2) \cdot \left( 1 - \beta(n)^q - \sqrt{1-\alpha(n)^2} \right).$
\end{lemma}

\begin{proof}
    Our proof strategy is inspired by the proof of~\cite[Corollary 4.3 and Lemma 4.4]{Liu23}. We begin by considering the following constant-rank quantum states $\rho'_0$ and $\rho'_1$, which can be prepared by the quantum circuits $Q'_0$ and $Q'_1$, respectively:
    \begin{align*}
        \rho'_0 &\coloneqq (p_0 \ketbra{0}{0} + p_1 \ketbra{1}{1}) \otimes \frac{1}{2}  \left( \ketbra{\psi_0}{\psi_0} + \ketbra{\psi_1}{\psi_1} \right)\\
        \rho'_{1} &\coloneqq \frac{1}{2} \ketbra{0}{0} \otimes \ketbra{\psi_0}{\psi_0} +  \frac{1}{2} \ketbra{1}{1} \otimes \ketbra{\psi_1}{\psi_1}.
    \end{align*}
    Here, $(p_0,p_1)$ is some two-element probability distribution that will be specified later. Moreover, for any $b\in\binset$, the quantum circuit $Q'_b$ uses $O\rbra{1}$ queries to controlled-$Q_0$ and controlled-$Q_1$ as well as $O(1)$ one- and two-qubit gates, as presented in~\cite[Figure 1 and Figure 2]{Liu23}. 

    Using the pseudo-additivity of $\Sq$ (\Cref{lemma:Tsallis-entropy-pseudoAdditivity}),  we can obtain that: 
    \begin{subequations}
        \label{eq:LowRankTsallisQED-hardness-state0}
        \begin{align}
            \Sq(\rho'_0) &= \Hq(p_0) + \Sq\rbra*{\frac{\ketbra{\psi_0}{\psi_0} + \ketbra{\psi_1}{\psi_1}}{2}} - (q-1) \cdot \Hq(p_0) \cdot \Sq\rbra*{\frac{\ketbra{\psi_0}{\psi_0} + \ketbra{\psi_1}{\psi_1}}{2}}\\
            &= \Hq(p_0) + \rbra*{1-(q-1)\Hq(p_0)} \cdot \Sq\rbra*{\frac{\ketbra{\psi_0}{\psi_0} + \ketbra{\psi_1}{\psi_1}}{2}}.
        \end{align}
    \end{subequations}
    
    By the joint $q$-Tsallis entropy theorem (\Cref{lemma:Tsallis-joint-entropy-theorem}), we have:
    \begin{equation}
        \label{eq:LowRankTsallisQED-hardness-state1}
        \Sq(\rho'_1) = \Hq(1/2) + 2^{-q} (\Sq(\ketbra{\psi_0}{\psi_0}) + \Sq(\ketbra{\psi_1}{\psi_1})) = \Hq(1/2).
    \end{equation}

    Combining \Cref{eq:LowRankTsallisQED-hardness-state0,eq:LowRankTsallisQED-hardness-state1}, we conclude that:
    \begin{subequations}
        \label{eq:LowRankTsallisQED-hardness-formula}
        \begin{align}
            \Sq(\rho'_0) - \Sq(\rho'_1) &= \rbra*{1-(q-1)\Hq(p_0)} \cdot \Sq\!\left(\frac{\ketbra{\psi_0}{\psi_0} + \ketbra{\psi_1}{\psi_1}}{2}\right) + \Hq(p_0) - \Hq\Big(\frac{1}{2}\Big) \\
            &= \rbra*{1-(q-1)\Hq(p_0)} \cdot \QJTq(\ketbra{\psi_0}{\psi_0},\ketbra{\psi_1}{\psi_1}) + \Hq(p_0) - \Hq\Big(\frac{1}{2}\Big).
        \end{align}
    \end{subequations}

    Next, we choose $p_0 \in (0,1/2)$ satisfying the following equality:
    \begin{equation}
        \label{eq:LowRankTsallisQED-hardness-chosenParameter}
        \Hq\Big(\frac{1}{2}\Big) - \Hq(p_0) = \frac{1-(q-1)\Hq(p_0)}{2} \left( \Hq\Big(\frac{1}{2}\Big) - \Hq\Big( \frac{1-\alpha}{2} \Big) + \Hq\Big(\frac{1}{2}\Big) \cdot \beta^q \right).
    \end{equation}
    
    As a consequence, we can derive that: 
    \begin{itemize}
        \item For the case where $\td(\ketbra{\psi_0}{\psi_0},\ketbra{\psi_1}{\psi_1}) \geq \alpha$, plugging the lower bound on $\QJTq$ in terms of the trace distance (\Cref{lemma:td-leq-QJSq}) into \Cref{eq:LowRankTsallisQED-hardness-formula,eq:LowRankTsallisQED-hardness-chosenParameter}, we obtain
        \begin{align*}
            \Sq(\rho'_0) - \Sq(\rho'_1) &\geq \rbra*{1-(q-1)\Hq(p_0)} \cdot \left( \Hq\Big(\frac{1}{2}\Big) - \Hq\Big( \frac{1-\alpha}{2} \Big) \right) + \Hq(p_0) - \Hq\Big(\frac{1}{2}\Big)\\
            &= \frac{1-(q-1)\Hq(p_0)}{2} \left( \Hq\Big(\frac{1}{2}\Big) \cdot \rbra*{1-\beta^q} - \Hq\Big( \frac{1-\alpha}{2} \Big)\right) \coloneqq \tilde{g}_q.
        \end{align*}
        
        \item For the case where $\td(\ketbra{\psi_0}{\psi_0},\ketbra{\psi_1}{\psi_1}) \leq \beta$, plugging the upper bound on $\QJTq$ in terms of the trace distance (\Cref{lemma:QJSq-leq-traceDistLike}) into \Cref{eq:LowRankTsallisQED-hardness-formula,eq:LowRankTsallisQED-hardness-chosenParameter}, we obtain
        \begin{align*}
            \Sq(\rho'_0) - \Sq(\rho'_1) &\leq \rbra*{1-(q-1)\Hq(p_0)} \cdot \beta^q \cdot \Hq\Big(\frac{1}{2}\Big) + \Hq(p_0) - \Hq\Big(\frac{1}{2}\Big)\\
            &= - \frac{1-(q-1)\Hq(p_0)}{2} \left( \Hq\Big(\frac{1}{2}\Big) \cdot \rbra*{1-\beta^q} - \Hq\Big( \frac{1-\alpha}{2} \Big) \right) = - \tilde{g}_q.
        \end{align*}
    \end{itemize}

    It is left to show a lower bound on $\tilde{g}(n)$. Using $\Hq(x) \leq \Hq(1/2)$ in \Cref{lemma:Tsallis-entropy-properties}, we have
    \begin{equation}
        \label{eq:LowRankTsallisQED-hardness-coefficient}
        \frac{1 - (q-1) \cdot \Hq(p_0)}{2} \geq \frac{1}{2} - \frac{q-1}{2} \cdot \Hq\Big(\frac{1}{2}\Big) = 2^{-q}. 
    \end{equation}

    Plugging the Tsallis binary entropy upper bound (\Cref{lemma:Tsallis-binary-entropy-upper-bound}) and \Cref{eq:LowRankTsallisQED-hardness-coefficient} into $\tilde{g}(n)$, we complete the proof by concluding the following: 
    \[ \tilde{g}_q(n) \geq 2^{-q} \cdot \Hq(1/2) \cdot \left(1- \beta(n)^q - \sqrt{1-\alpha^2(n)}\right) \coloneqq g_q(n+1) = g_q(n'). \qedhere\]
\end{proof}

\subsection{Mixed-state reductions}
\label{subsec:mixed-state-reductions}

In this subsection, we present two reductions for mixed states. The first reduction is from the trace distance between two $n$-qubit states (\QSD{}), to the quantum $q$-Tsallis entropy difference between two new $(n+1)$-qubit states (\TsallisQED{}), for $1 \leq q \leq 2$, under appropriate assumptions about $\Sq(\rho_0)$ and $\Sq(\rho_1)$, as stated in \Cref{lemma:reduction-QSD-TsallisQED}. The second reduction is from the trace distance between an $n$-qubit quantum state (\QSCMM{}) and the $n$-qubit maximally mixed state to the quantum $q$-Tsallis entropy of the state (\TsallisQEA{}) for $q=1+\frac{1}{n-1}$, as stated in \Cref{lemma:reduction-MaxMixedQSD-TsallisQEA}. 

\subsubsection{\QSD{} \texorpdfstring{$\leq$}{≤} \TsallisQED{} for \texorpdfstring{$1\leq q \leq 2$}{1≤q≤2}}

\begin{lemma}[$\QSD\leq\TsallisQED$]
    \label{lemma:reduction-QSD-TsallisQED}
    Let $Q_0$ and $Q_1$ be quantum circuits acting on $m$ qubit, defined in \Cref{def:TsallisQED}, that prepares the purification of $n$-qubit mixed states $\rho_0$ and $\rho_1$, respectively. For any $b \in \binset$, there is a new quantum circuits $Q'_b$ acting on $m+3$ qubits, requiring $O(1)$ queries to controlled-$Q_0$ and controlled-$Q_1$, as well as $O(1)$ one- and two- qubit gates, that prepares a new $n'$-qubit mixed state $\rho'_b$, where $n' \coloneqq n+1$, such that\emph{:} For any $\rho_0$ and $\rho_1$ satisfying $\max\{\Sq(\rho_0),\Sq(\rho_1)\} \leq \gamma(n)$ with $\Sq(I/2) \leq \gamma(n) \leq \Sq\rbra*{(I/2)^{\otimes n}}$, any $\varepsilon(n) \in (0,1/2)$, and any $q\in[1,2]$, there is a $g(n)>0$ with appropriate ranges of $\gamma$, $\varepsilon$, and $n$ such that
    \begin{align*}
        \td(\rho_0,\rho_1) &\geq 1 - \varepsilon(n) &~\Rightarrow~  \Sq(\rho'_0) - \Sq(\rho'_1) \geq g_q(n') = g_q(n+1),\\
        \td(\rho_0,\rho_1) &\leq \varepsilon(n) &~\Rightarrow~ \Sq(\rho'_1) - \Sq(\rho'_0) \geq g_q(n') = g_q(n+1),
    \end{align*}
    where $g_q(n) \coloneqq \frac{1}{2} \Hq\rbra*{\frac{1}{2}} - \gamma(n) \rbra*{\frac{1}{2} \!-\! \frac{1}{2^q}} - \rbra*{\frac{1}{2}\!+\!\frac{1}{2^q}} \rbra*{ \frac{\varepsilon(n)^q}{2^q}  \ln_q\!\big(2^{n}\big) \!+\! \Hq\rbra*{\frac{1}{2}} \sqrt{\varepsilon(n) (2\!-\!\varepsilon(n))} }.$
\end{lemma}

\begin{proof}
Our proof strategy is somewhat inspired by~\cite[Section 5.4]{BASTS10}. We start by considering the following mixed states $\rho'_0$ and $\rho'_1$:
% The current construction is shifted too much. 
\begin{align*}
    \rho'_0 &\coloneqq \rbra*{\vartheta\ketbra{0}{0} + (1-\vartheta)\ketbra{1}{1}} \otimes \rho_+, \text{ where } 2\Hq(\vartheta)=\Hq\rbra*{\frac{1}{2}} \text{ and }\rho_+ \coloneqq \frac{\rho_0+\rho_1}{2},\\
    \rho'_1 &\coloneqq \frac{1}{2} \ketbra{0}{0} \otimes \rho_0 + \frac{1}{2} \ketbra{1}{1} \otimes \rho_1.
\end{align*}

These states $\rho'_0$ and $\rho'_1$ can be prepared by the quantum circuits $Q'_0$ and $Q'_1$, respectively. For instance, adapting the constructions in~\cite[Figure 1 and Figure 2]{Liu23}, for any $b\in\binset$, the quantum circuit $Q'_b$ uses $O(1)$ queries to controlled-$Q_0$ and controlled-$Q_1$, as well as $O(1)$ one- and two-qubit gates. 

Utilizing the pseudo-additivity of $\Sq$ (\Cref{lemma:Tsallis-entropy-pseudoAdditivity}), we have: 
\begin{equation}
    \label{eq:QSD-to-TsallisQED-pseudoAdditivity}
    \Sq(\rho'_0) 
    = \Hq\rbra*{\vartheta} + \rbra*{1-(q-1) \Hq\rbra*{\vartheta}} \Sq(\rho_+)
    = \frac{1}{2}\Hq\rbra*{\frac{1}{2}} + \rbra*{1-\frac{q-1}{2} \cdot \Hq\rbra*{\frac{1}{2}}} \Sq(\rho_+).
\end{equation}

Using the joint $q$-Tsallis entropy theorem (\Cref{lemma:Tsallis-joint-entropy-theorem}), we obtain: 
\begin{equation}
    \label{eq:QSD-to-TsallisQED-RhoP0}
    \Sq(\rho'_1) = \Hq\rbra*{\frac{1}{2}} + \frac{1}{2^{q}} \rbra*{\Sq(\rho_0) +  \Sq(\rho_1)}. 
\end{equation}

Combining \Cref{eq:QSD-to-TsallisQED-RhoP0,eq:QSD-to-TsallisQED-pseudoAdditivity}, we obtain: 
\begin{equation}
    \label{eq:eq:QSD-to-TsallisQED-diffExpression}
    \Sq(\rho'_0) - \Sq(\rho'_1) = \rbra*{1- \frac{q-1}{2} \cdot \Hq\rbra*{\frac{1}{2}}} \Sq(\rho_+) -\frac{1}{2} \Hq\rbra*{\frac{1}{2}} - \frac{1}{2^{q}} \rbra*{\Sq(\rho_0) +  \Sq(\rho_1)}
\end{equation}

Next, we can consider the following two cases: 
\begin{itemize}[leftmargin=1em]
    \item For the case where $\td(\rho_0,\rho_1) \geq 1-\varepsilon$,  using the lower bound on \QJTq (\Cref{lemma:td-leq-QJSq}), we have: 
    \begin{equation}
        \label{eq:QSD-to-TsallisQED-QJTqLowerBound}
        \Sq(\rho_+) - \frac{1}{2} \rbra*{\Sq(\rho_0)+\Sq(\rho_1)} = \QJTq(\rho_0,\rho_1) \geq \Hq\rbra*{\frac{1}{2}} - \Hq\rbra*{\frac{1-\td(\rho_0,\rho_1)}{2}}.
    \end{equation}
    Substituting \Cref{eq:QSD-to-TsallisQED-QJTqLowerBound} into \Cref{eq:eq:QSD-to-TsallisQED-diffExpression}, we obtain:
    \begin{align*}
        &\Sq(\rho'_0) - \Sq(\rho'_1) \\
        \geq~& \rbra*{1- \tfrac{q-1}{2} \cdot \Hq\rbra*{\tfrac{1}{2}}} \rbra*{\tfrac{1}{2} \rbra*{\Sq(\rho_0)+\Sq(\rho_1)} + \Hq\rbra*{\tfrac{1}{2}} - \Hq\rbra*{\tfrac{\varepsilon}{2}}} - \tfrac{1}{2} \Hq\rbra*{\tfrac{1}{2}} - \tfrac{1}{2^{q}}\rbra*{\Sq(\rho_0) +  \Sq(\rho_1)} \\
        \geq~& \rbra*{\tfrac{1}{2} - \tfrac{1}{2^q} - \tfrac{q-1}{4} \Hq\rbra*{\tfrac{1}{2}}} \rbra*{\Sq(\rho_0)+\Sq(\rho_1)} + \rbra*{1 - \tfrac{q-1}{2} \Hq\rbra*{\tfrac{1}{2}}} \Hq\rbra*{\tfrac{1}{2}} \rbra*{1-\sqrt{\varepsilon(2-\varepsilon)}} - \tfrac{1}{2} \Hq\rbra*{\tfrac{1}{2}}\\
        \geq~& \rbra*{\tfrac{1}{2}+\tfrac{1}{2^q}} \Hq\rbra*{\tfrac{1}{2}} \rbra*{1 - \sqrt{\varepsilon(2-\varepsilon)}} - \tfrac{1}{2} \Hq\rbra*{\tfrac{1}{2}} \\
        =~& \tfrac{1}{2^q} \Hq\rbra*{\tfrac{1}{2}} - \rbra*{\tfrac{1}{2}+\tfrac{1}{2^q}} \Hq\rbra*{\tfrac{1}{2}} \sqrt{\varepsilon(2-\varepsilon)} \coloneqq \tilde{g}_q^\texttt{Y}(\varepsilon).
    \end{align*}
    Here, the third line uses the Tsallis binary entropy upper bound (\Cref{lemma:Tsallis-binary-entropy-upper-bound}) and the fact that $1 - \tfrac{q-1}{2} \Hq\rbra*{\tfrac{1}{2}} > 0$ for $q \in [1,2]$. The last line relies on the following facts: (a) $\Sq(\rho) \geq 0$ for any quantum state $\rho$; (b) $2\rbra*{\tfrac{1}{2} - \tfrac{1}{2^q} - \tfrac{q-1}{4} \Hq\rbra*{\tfrac{1}{2}}} = \frac{1}{2}-\frac{1}{2^q} \geq 0$ for $q\in[1,2]$; and (c) $1 - \tfrac{q-1}{2} \Hq\rbra*{\tfrac{1}{2}} = \tfrac{1}{2}+\tfrac{1}{2^q}$. 

    \item For the case where $\td(\rho_0,\rho_1) \leq \varepsilon$, by Fannes' inequality for the $q$-Tsallis entropy (\Cref{lemma:Fannes-inequality-Sq}), we have: 
    \begin{subequations}
    \label{eq:QSD-to-TsallisQED-Fannes}
    \begin{align}
        \Sq(\rho_+) & \leq \frac{\abs*{\Sq(\rho_+) - \Sq(\rho_0)}}{2} + \frac{\abs*{\Sq(\rho_+) - \Sq(\rho_1)}}{2}  + \frac{\Sq(\rho_0) + \Sq(\rho_1)}{2} \\
        &\leq \td(\rho_+,\rho_b)^q \cdot \ln_q\rbra*{2^n-1} + \Hq(\td(\rho_+,\rho_b)) + \frac{\Sq(\rho_0) + \Sq(\rho_1)}{2} \\
        &\leq \rbra*{\frac{\td(\rho_0,\rho_1)}{2}}^q \ln_q\!\big(2^n\big) + \Hq\rbra*{\frac{\td(\rho_0,\rho_1)}{2}} + \frac{\Sq(\rho_0) + \Sq(\rho_1)}{2} 
    \end{align}
    \end{subequations}
    Here, the first line is due to the triangle inequality, and the last line is because $\ln_q(x)$ is monotonically increasing on $x> 0$ for any fixed $q>1$. 
    
    Plugging \Cref{eq:QSD-to-TsallisQED-Fannes} into \Cref{eq:eq:QSD-to-TsallisQED-diffExpression}, we can derive that: 
    \begin{align*}
        &\Sq(\rho'_0) - \Sq(\rho'_1)\\
        \leq~& \rbra*{1- \tfrac{q-1}{2} \Hq\rbra*{\tfrac{1}{2}}} \rbra*{\rbra*{\tfrac{\varepsilon}{2}}^q \ln_q\!\big(2^n\big) + \Hq\rbra*{\tfrac{\varepsilon}{2}} + \tfrac{1}{2}\rbra*{\Sq(\rho_0) + \Sq(\rho_1)} }  - \tfrac{1}{2} \Hq\rbra*{\tfrac{1}{2}} - \tfrac{1}{2^{q}}\rbra*{\Sq(\rho_0) +  \Sq(\rho_1)}\\
        \leq~& \rbra*{\tfrac{1}{2} \!-\! \tfrac{1}{2^{q}} \!-\! \tfrac{q-1}{4} \Hq\rbra*{\tfrac{1}{2}}}  \rbra*{\Sq(\rho_0) \!+\!  \Sq(\rho_1)} \!+\!  \rbra*{1-\tfrac{q-1}{2}\Hq\rbra*{\tfrac{1}{2}}} \rbra*{ \rbra*{\tfrac{\varepsilon}{2}}^q \ln_q\!\big(2^n\big) \!+\! \Hq\rbra*{\tfrac{1}{2}} \sqrt{\varepsilon (2\!-\!\varepsilon)} }  \!-\! \tfrac{1}{2} \Hq\rbra*{\tfrac{1}{2}} \\
        \leq~& \rbra*{\tfrac{1}{2} - \tfrac{1}{2^q}} \cdot \gamma + \rbra*{\tfrac{1}{2}+\tfrac{1}{2^q}} \rbra*{ \rbra*{\tfrac{\varepsilon}{2}}^q \cdot \ln_q\!\big(2^n\big) + \Hq\rbra*{\tfrac{1}{2}} \sqrt{\varepsilon (2-\varepsilon)} } - \tfrac{1}{2} \Hq\rbra*{\tfrac{1}{2}} \coloneqq -\tilde{g}_q^\texttt{N}(\varepsilon,n,\gamma).
    \end{align*}
    Here, the third line uses the Tsallis binary entropy upper bound (\Cref{lemma:Tsallis-binary-entropy-upper-bound}) and the fact that $1 - \tfrac{q-1}{2} \Hq\rbra*{\tfrac{1}{2}} > 0$ for $q \in [1,2]$. The last line relies on the following facts: (a) $1 - \tfrac{q-1}{2} \Hq\rbra*{\tfrac{1}{2}} = \tfrac{1}{2}+\tfrac{1}{2^q}$; (b) $2\rbra*{\tfrac{1}{2} - \tfrac{1}{2^q} - \tfrac{q-1}{4} \Hq\rbra*{\tfrac{1}{2}}} = \frac{1}{2}-\frac{1}{2^q} \geq 0$ for $q\in[1,2]$; and (c) $\Sq(\rho) \leq \gamma \leq \Sq\rbra*{(I/2)^{\otimes n}}$ for any $n$-qubit state $\rho$.  
\end{itemize}

It is evident that $\tilde{g}^{\texttt{N}}_q(\varepsilon,n,\gamma)$ is monotonically decreasing on $\gamma \geq 0$ for any fixed $q$, $\varepsilon$, and $n$. Consequently, it remains to show that $\tilde{g}^{\texttt{Y}}_q(\varepsilon) \geq \tilde{g}^{\texttt{N}}_q(\varepsilon,n,\Hq(1/2)) \geq \tilde{g}^{\texttt{N}}_q(\varepsilon,n,\gamma)$ for $\Hq(1/2) = \Sq(I/2) \leq \gamma \leq \Sq\rbra*{(I/2)^{\otimes n}}$. In particular, by noting that $(\varepsilon/2)^q \cdot \ln_q\big(2^n\big) \geq 0$ for $q\geq 1$ and $\varepsilon\geq0$, we obtain:
\begin{align*}
    \tilde{g}^{\texttt{Y}}_q(\varepsilon) - \tilde{g}^{\texttt{N}}_q\rbra*{\varepsilon,n,\Hq\!\rbra*{\frac{1}{2}}} &= \frac{1}{2^q} \Hq\!\rbra*{\frac{1}{2}} + \rbra*{\frac{1}{2}+\frac{1}{2^q}} \rbra*{\frac{\varepsilon}{2}}^q \ln_q\!\big(2^n\big) - \frac{1}{2} \Hq\!\rbra*{\frac{1}{2}} + \rbra*{\frac{1}{2} - \frac{1}{2^q}} \Hq\!\rbra*{\frac{1}{2}}\\
    &= \rbra*{\frac{1}{2}+\frac{1}{2^q}} \rbra*{\frac{\varepsilon}{2}}^q \ln_q\!\big(2^n\big)\\
    & \geq 0.
\end{align*}

Therefore, we complete the proof by choosing $g_q(n) = \tilde{g}_q^{\texttt{N}}\rbra*{\varepsilon(n), n, \gamma(n)}$, specifically:
\[ g_q(n) \coloneqq \frac{1}{2} \Hq\rbra*{\frac{1}{2}} \!-\! \gamma(n) \rbra*{\frac{1}{2} \!-\! \frac{1}{2^q}} \!-\! \rbra*{\frac{1}{2}\!+\!\frac{1}{2^q}} \rbra*{ \frac{\varepsilon(n)^q}{2^q} \ln_q\!\big(2^{n}\big) \!+\! \Hq\rbra*{\frac{1}{2}} \sqrt{\varepsilon(n) (2\!-\!\varepsilon(n))} }. \qedhere\]
\end{proof}

\subsubsection{\QSCMM{} \texorpdfstring{$\leq$}{≤} \TsallisQEA{} for \texorpdfstring{$q(n)=1+\frac{1}{n-1}$}{q(n)=1+1/(n-1)}}

\begin{lemma}[$\QSCMM\leq\TsallisQEA$]
    \label{lemma:reduction-MaxMixedQSD-TsallisQEA}
    Let $Q$ be a quantum circuit acting on $m$ qubit, defined in \Cref{def:TsallisQEA}, that prepares the purification of $n$-qubit mixed states $\rho$, respectively. For any $\rho$, any $n \geq 5$, and any $q(n) = 1+1/(n-1)$, let $t(n) \coloneqq \frac{1}{4}\big(3n-n^{1+\frac{1}{n}}-1\big)$, we have\emph{:}
    \begin{align*}
        \td\rbra*{\rho,(I/2)^{\otimes n}} &\leq 1/n &~\Rightarrow~ \Sq(\rho) &> t(n) + 1/150,\\
        \td\rbra*{\rho,(I/2)^{\otimes n}} &\geq 1-1/n &~\Rightarrow~  \Sq(\rho) &< t(n) - 1/150.
    \end{align*}
\end{lemma}

\begin{proof}
    Let $\rho = \sum_{i\in[2^{n}]} \lambda_i \ketbra{v_i}{v_i}$ be the spectral decomposition of $\rho$, where $\{v_i\}_{i\in[2^n]}$ is an orthonormal basis and $p \coloneq (\lambda_1, \cdots, \lambda_{2^n})$ is a probability distribution of dimension $2^n$. And let $\nu$ be the uniform distribution of dimension $2^n$. 
    Noting that $\rho$ and $(I/2)^{\otimes n}$ commute, we have $\td\rbra*{\rho, (I/2)^{\otimes n}} = \TV(p,\nu)$ and $\Sq(\rho) = \Hq(p)$. 
    
    Let $t(n) \coloneqq \frac{1}{4} \rbra*{3n -n^{1+\frac{1}{n}}-1 }$.
    Next, we can consider the following two cases: 
    \begin{itemize}[leftmargin=1em]
    \item For the case where $\td\rbra*{\rho,(I/2)^{\otimes n}} \leq 1/n$, by the lower bound on $\Hq(p)$ in \Cref{lemma:inequality-uniformTV-TsallisEA}, it follows that
        \begin{align*}
            \Sq(\rho) &\geq \ln_{1+\frac{1}{n-1}}\!\big(2^n\big) \cdot \rbra*{1 - \td\rbra*{\rho,(I/2)^{\otimes n}} - 2^{-n}}\\
            &\geq (n-1) \rbra*{1-\frac{1}{2} \cdot \Big(\frac{1}{2}\Big)^{\frac{1}{n-1}}}\rbra*{1 - \frac{1}{n} - 2^{-n}} \coloneqq \tau_{\texttt{Y}}(n).
        \end{align*}

    By a direct calculation, we obtain: 
    \begin{subequations}
    \label{eq:reduction-QSCMM-TsallisQEA-smallTD}
    \begin{align}
        \Sq(\rho) - t(n) &\geq \tau_{\texttt{Y}}(n) - t(n) = g_1(n) + g_2(n) + g_3(n) - \frac{7}{4},\\
        \text{where } g_1(n) &\coloneqq 2^{-n} + \frac{1-2^{\frac{n}{1-n}}}{n} + 2^{\frac{n^2}{1-n}} (n-1) + \frac{n}{4} \rbra*{1- 2^{\frac{1}{1-n}}},\\
        g_2(n) &\coloneqq 2^{\frac{1}{1-n}}-2^{-n} n, \quad g_3(n) \coloneqq \frac{n}{4} \rbra*{ n^{\frac{1}{n}}-2^{\frac{1}{1-n}}}.
    \end{align}
    \end{subequations}

    Through a fairly tedious calculation, we know that $g_1(n)$, $g_2(n)$, and $g_3(n)$  defined in \Cref{eq:reduction-QSCMM-TsallisQEA-smallTD} satisfy the properties in \Cref{fact:reduction-QSCMM-TsallisQEA-smallTD}, and the proof is deferred in \Cref{subsec:omitted-proofs-QJTq-reductions}. 

    \begin{restatable}{fact}{reductionQSCMMtoTsallisQEAsmallTD}
        \label{fact:reduction-QSCMM-TsallisQEA-smallTD}
        Let $g_1(n)$, $g_2(n)$, and $g_3(n)$ be functions defined in \Cref{eq:reduction-QSCMM-TsallisQEA-smallTD}. It holds that\emph{:}        
        \begin{enumerate}[label={\upshape(\arabic*)}, topsep=0.33em, itemsep=0.33em, parsep=0.33em] 
            \item For $n\geq 3$, $g_1(n) \geq 0$. 
            \item For $n\geq 3$, $g_2(n)$ and $g_3(n)$ are monotonically increasing. 
        \end{enumerate}        
    \end{restatable}

    Combining \Cref{eq:reduction-QSCMM-TsallisQEA-smallTD,fact:reduction-QSCMM-TsallisQEA-smallTD}, we obtain that: 
    \begin{equation}
        \label{eq:reduction-QSCMM-TsallisQEA-smallTD-nBound}
        \forall n \geq 5,~ \Sq(\rho) - t(n) \geq \tau_{\texttt{Y}}(n) - t(n) \geq g_2(n) + g_3(n) - \frac{7}{4} > \frac{1}{150}.
    \end{equation}
    %1/148
    
    \item For the case where $\td\rbra*{\rho,(I/2)^{\otimes n}} \geq 1-1/n$, by noting $\td\rbra*{\rho,(I/2)^{\otimes n}} q \geq \rbra*{1-\frac{1}{n}}  \rbra*{1+\frac{1}{n-1}} = 1$ and using the upper bound on $\Hq(p)$ in \Cref{lemma:inequality-uniformTV-TsallisEA}, it holds that
    \begin{align*}
        \Sq(\rho) &\leq \ln_{1+\frac{1}{n-1}}\rbra*{2^{n} \rbra*{1-\td\rbra*{\rho,(I/2)^{\otimes n}}}}\\
        &\leq \ln_{1+\frac{1}{n}}\rbra*{2^{n} \rbra*{1-\td\rbra*{\rho,(I/2)^{\otimes n}}}}\\
        &\leq n\rbra*{1 - \frac{1}{2} \cdot n^{1/n}} \coloneqq \tau_{\texttt{N}}(n).
    \end{align*}
    Here, the second line is because $\ln_{q}(x) < \ln_{q'}(x)$ for $q > q' > 0$ and $\frac{1}{n-1} > \frac{1}{n}$. 

    Similarly, a direct calculation implies that: 
    \begin{equation}
        \label{eq:reduction-QSCMM-TsallisQEA-largeTD}
        t(n) - \Sq(\rho) \geq t(n) - \tau_{\texttt{N}}(n) = \frac{g_4(n)-1}{4}, \text{ where } g_4(n) \coloneqq n \rbra*{ n^{\frac{1}{n}} - 1}.
    \end{equation}

    Next, we will prove that $g_4(n)$ is monotonically non-decreasing for $n \geq 2$. We proceed by expressing the first and second derivative of $g_4(n)$:
    \[\frac{\dd}{\dn} g_4(n) = \frac{n^{\frac{1}{n}}}{n} (n-\log (n)+1)-1, \text{ and }
        \frac{\dd^2}{\dn^2} g_4(n) = \frac{n^{\frac{1}{n}}}{n^3} \rbra*{ \rbra*{\log(n)-1}^2 - n}.
    \]
    Since $\sqrt{n} > \log{n}$, we know that $\frac{\dd^2}{\dn^2} g_4(n)$ has one zero at $n=1$. As $\frac{\dd^2}{\dn^2} g_4(n)\big|_{n=e}=-e <0$, we have that $\frac{\dd}{\dn} g_4(n)$ is monotonically decreasing for $n \geq 2$, and thus, $\frac{\dd}{\dn} g_4(n) \geq \lim_{n\rightarrow\infty} \frac{\dd}{\dn} g_4(n) = 0$ for $n \geq 2$. 
    Hence, we conclude that $g_4(n)$ is monotonically non-decreasing for $n \geq 2$. Consequently, combining with \Cref{eq:reduction-QSCMM-TsallisQEA-largeTD}, we obtain: 
    \begin{equation}
        \label{eq:reduction-QSCMM-TsallisQEA-largeTD-nBound}
        \forall n\geq 3,~ t(n) - \Sq(\rho) \geq t(n) - \tau_{\texttt{N}}(n) = \frac{g_4(n)-1}{4} > \frac{1}{13}.
    \end{equation}
    \end{itemize}

    Lastly, we finish the proof by comparing \Cref{eq:reduction-QSCMM-TsallisQEA-smallTD-nBound} with \Cref{eq:reduction-QSCMM-TsallisQEA-largeTD-nBound}. 
\end{proof}

\subsection{Computational hardness results}
\label{subsec:hardness-results}

In this subsection, we present the computational hardness results for various settings of \TsallisQED{} and \TsallisQEA{} by using our reductions established in \Cref{subsec:pure-state-reduction,subsec:mixed-state-reductions}. 

\subsubsection{\BQP{} hardness results}

\begin{theorem}[\ConstRankTsallisQED{} is \BQP{}-hard for $1\leq q \leq 2$]
    \label{thm:ConstRankTsallisQED-BQPhard}
    For any $q \in [1,2]$ and any $n\geq 3$, the following holds\emph{:} 
    \[\forall g_q(n) \in \sbra*{\frac{1}{\poly(n)}, 2^{-q} \Hq\rbra*{\frac{1}{2}} \rbra*{1-2^{-\frac{n}{2}+2} }},~ \ConstRankTsallisQED[g_q(n)] \text{ is } \BQP\text{-hard}. \]
\end{theorem}

\begin{proof}
    Using \Cref{lemma:PureQSD-is-BQPhard}, we have that $\PureQSD\sbra*{\sqrt{1-2^{-2\hat{n}}}, 2^{(-\hat{n}+1)/2}}$ is \BQP{}-hard for $\hat{n}\geq 2$. Let $Q_0$ and $Q_1$ be the corresponding \BQP{}-hard instance such that these circuits are polynomial-size and prepare the pure states $\ketbra{\psi_0}{\psi_0}$ and $\ketbra{\psi_1}{\psi_1}$, respectively. 
    Leveraging the reduction from \PureQSD{} to \ConstRankTsallisQED{} (\Cref{lemma:reduction-PureQSD-ConstRankTsallisQED}), there are two polynomial-size quantum circuits $Q'_0$ and $Q'_1$, which prepares the purifications of constant-rank states $\rho'_0$ and $\rho'_1$, such that: For any $1 \leq q \leq 2$ and any $n = \hat{n}+1 \geq 3$, 
    \begin{align*}
        \td\rbra*{\ketbra{\psi_0}{\psi_0}, \ketbra{\psi_1}{\psi_1}} &\geq 1- 2^{-\hat{n}} &~\Rightarrow~  \Sq(\rho'_0) - \Sq(\rho'_1) \geq g_q(n) = g_q(\hat{n}+1),\\
        \td\rbra*{\ketbra{\psi_0}{\psi_0}, \ketbra{\psi_1}{\psi_1}} &\leq 2^{-\hat{n}} &~\Rightarrow~ \Sq(\rho'_1) - \Sq(\rho'_0) \geq g_q(n) = g_q(\hat{n}+1).
    \end{align*} 

    Hence, we complete the proof by a direct calculation: 
    \begin{align*}
        \forall n\geq 3, \quad g_q(n) &\coloneqq 2^{-q} \cdot \Hq(1/2) \cdot \rbra*{1 - 2^{-q(n-2)/2} - \sqrt{1 - \rbra*{1-2^{-2(n-1)}}}} \\ 
        &\geq 2^{-q} \cdot \Hq(1/2) \cdot \rbra*{1-2^{-(n-2)/2}-2^{-(n-1)}}\\
        &= 2^{-q} \cdot \Hq(1/2) \cdot \rbra*{1-2^{-\frac{n}{2}+2}} > 0. \qedhere
    \end{align*}
\end{proof}

\begin{theorem}[\ConstRankTsallisQEA{} is \BQP{}-hard under Turing reduction for $1\leq q \leq 2$]
    \label{thm:ConstRankTsallisQEA-BQPhard}
    For any $q \in [1,2]$ and any $n\geq 3$, the following holds\emph{:} 
    \[\ConstRankTsallisQEA \text{ with } g(n)=\Theta(1) \text{ is } \BQP\text{-hard under Turing reduction}. \]
\end{theorem}

\begin{proof}
    For any $1 \leq q \leq 2$ and $n \geq 3$, since $\ConstRankTsallisQED\sbra*{\hat{g}_q(n)}$ is \BQP-hard{} under Karp reduction (\Cref{thm:ConstRankTsallisQED-BQPhard}), where $\hat{g}_q(n) \coloneqq 2^{-q} \Hq(1/2) \rbra[\big]{1-2^{ -\frac{n}{2}+\frac{7}{5} }}$, it suffices to provide an algorithm for $\ConstRankTsallisQED\sbra*{\hat{g}_q(n)}$ by using $\ConstRankTsallisQEA[t(n),g(n)]$ as subroutines, with appropriately adaptive choices of $t(n)$ and $g(n)$.

    Let $Q_0$ and $Q_1$ be the corresponding \BQP{}-hard instance such that these circuits are polynomial-size and prepare the constant-rank states $\rho_0$ and $\rho_1$, respectively. Let $\TsallisQEAalgo\rbra*{Q, t(n), g(n)}$ be the subroutine for decide whether $\Sq(\rho) \geq t(n)+g(n)$ or $\Sq(\rho) \leq t(n)-g(n)$. Next, we estimate $\Sq(\rho_b)$ to within additive error $\hat{g}_q(n)/2$ for $b\in\binset$. This procedure, inspired by~\cite[Appendix A.2 Part 1]{Ambainis14}, is denoted by $\BiSearch$, as presented in \Cref{algo:Tsallis-entropy-search-to-decision}.

    \begin{algorithm}[!htp]
        \caption{Tsallis entropy estimation $\BiSearch(Q, \tau, g)$ via queries to $\TsallisQEAalgo$.}
        \label{algo:Tsallis-entropy-search-to-decision}
        \begin{algorithmic}[1]
            \Require A quantum circuit $Q$ that prepares the purification of $\rho$, an upper bound $\tau$ on the $q$-Tsallis entropy $\Sq\rbra{\rho}$, and a precision parameter $g$.

            \Ensure Return $t$ such that $\abs{t - \Sq\rbra{\rho}} \leq g/2$.

            \State Let $\delta \gets g/2$, and set the interval $[a,b] \gets [0,\tau]$.

            \While {$b - a > g/2$}
                \State Query $\TsallisQEAalgo\rbra*{Q, \frac{a+b}{2}, \frac{\delta}{4}}$ to decide whether $\Sq(\rho) \geq \frac{a+b}{2} + \frac{\delta}{4}$ or $\Sq(\rho) \leq \frac{a+b}{2} - \frac{\delta}{4}$.
                \If {$\Sq(\rho) \geq \frac{a+b}{2} + \frac{\delta}{4}$}
                \State $[a,b] \gets \sbra*{\frac{a+b}{2} - \frac{g}{4}, b}$.
                \Else
                \State $[a,b] \gets \sbra*{a, \frac{a+b}{2} + \frac{g}{4}}$.
                \EndIf
            \EndWhile
            \State \Return $\frac{a+b}{2}$.
        \end{algorithmic}
    \end{algorithm}

    To solve $\ConstRankTsallisQED\sbra*{\hat{g}_q(n)}$, noting that $\max\{\rank(\rho_0),\rank(\rho_1)\} \leq r \leq O(1)$, we choose $\tau(n) = \Sq\rbra*{(I/2)^{\otimes r}}$. Subsequently, let $t_0(n) = \BiSearch(Q_0, \tau(n), \hat{g}_q(n))$ and $t_1(n) = \BiSearch(Q_1, \tau(n), \hat{g}_q(n))$, we obtain: 
    \begin{subequations}
    \label{eq:error-bounds-Turing-reduction}
    \begin{align}
        \Sq(\rho_0) - \Sq(\rho_1) &\geq \hat{g}_q(n) & \Rightarrow~ & t_0(n) - t_1(n) \geq \Sq(\rho_0) - \frac{\hat{g}_q(n)}{2} - \rbra*{\Sq(\rho_1) + \frac{\hat{g}_q(n)}{2}}  \geq 0,\\ 
        \Sq(\rho_0) - \Sq(\rho_1) &\leq -\hat{g}_q(n) & \Rightarrow~ & t_0(n) - t_1(n) \leq \Sq(\rho_0) + \frac{\hat{g}_q(n)}{2} - \rbra*{\Sq(\rho_1) - \frac{\hat{g}_q(n)}{2}} \leq 0.
    \end{align}
    \end{subequations}

    Note that $\hat{g}_q(n) = 2^{-q} \Hq(1/2) \rbra[\big]{1-2^{ -\frac{n}{2}+\frac{7}{5} }} \geq 2^{-q-1}\rbra[\big]{2-2^{\frac{9}{10}}}\cdot \Hq(1/2)$ for $n \geq 3$ and $\tau(n) \leq \S\rbra*{(I/2)^{r}} \leq O(1)$. 
    Since each query to $\TsallisQEAalgo$ in $\BiSearch$ decreases the size of the interval $[a,b]$ by almost a half, we can conclude that the number of adaptive queries to $\TsallisQEAalgo$ in $\BiSearch(Q_0, \tau(n), \hat{g}_q(n))$ and $\BiSearch(Q_1, \tau(n), \hat{g}_q(n))$ is $O\rbra*{\log\rbra*{1/\hat{g}_q(n)}} = O(1)$. 
\end{proof}

\subsubsection{\QSZK{} hardness results}
\begin{theorem}[\TsallisQED{} is \QSZK{}-hard for $1 < q \leq 1\!+\!\frac{1}{n-1}$]
    \label{thm:TsallisQED-QSZK-hardness}
    For any $q\in \big(1,1+\frac{1}{n-1}\big]$ and any $n \geq 90$, it holds that 
    \[ \forall g(n) \in \sbra*{1/\poly(n), 1/400},~ \TsallisQED[g(n)] \text{ is } \QSZK{}\text{-hard}.\]
\end{theorem}

\begin{proof}
    Following \Cref{lemma:QSD-is-QSZKhard}, we have that $\QSD\sbra*{1-2^{-\hat{n}^{0.49}}, 2^{-\hat{n}^{0.49}}}$ is \QSZK{}-hard for $\hat{n} \geq 1$.
    Let $Q_0$ and $Q_1$ be the corresponding \QSZK{}-hard instance such that these circuits are polynomial-size and prepare the purification of $\rho_0$ and $\rho_1$, respectively. 
    Leveraging the reduction from \QSD{} to \TsallisQED{} (\Cref{lemma:reduction-QSD-TsallisQED}), there are two polynomial-size quantum circuits $Q'_0$ and $Q'_1$, which prepare the purifications of $n$-qubit $\rho'_0$ and $\rho'_1$ where $n \coloneqq \hat{n}+1$, respectively, such that:
    \begin{align*}
        \td(\rho_0,\rho_1) &\geq 1 - 2^{-\hat{n}^{0.49}} &~\Rightarrow~  \Sq(\rho'_0) - \Sq(\rho'_1) \geq g_q(n) = g_q(\hat{n}+1),\\
        \td(\rho_0,\rho_1) &\leq 2^{-\hat{n}^{0.49}} &~\Rightarrow~ \Sq(\rho'_1) - \Sq(\rho'_0) \geq g_q(n) = g_q(\hat{n}+1).
    \end{align*} 
    
    Since $\sqrt{2^{-\hat{n}^{0.49}} \rbra*{2-2^{-\hat{n}^{0.49}}}} \leq 2^{\frac{1-\hat{n}^{0.49}}{2}}$ and $\gamma(n) \leq \Sq\rbra*{(I/2)^{\otimes \hat{n}}} = \frac{1-2^{\hat{n}(1-q)}}{q-1}$, we have 
    \[ g_q(\hat{n}) \geq \underbrace{\frac{1}{2} \Hq\rbra*{\frac{1}{2}} - \frac{1-2^{\hat{n}(1-q)}}{q-1} \rbra*{\frac{1}{2} \!-\! \frac{1}{2^q}}}_{G_1(q;\hat{n})} - \underbrace{\rbra*{\frac{1}{2}\!+\!\frac{1}{2^q}} \frac{2^{-\hat{n}^{0.49}q}}{2^q}  \ln_q\!\big(2^{\hat{n}}\big)}_{G_2(q;\hat{n})} -  \underbrace{\rbra*{\frac{1}{2}\!+\!\frac{1}{2^q}}  \Hq\rbra*{\frac{1}{2}} 2^{\frac{1-\hat{n}^{0.49}}{2}}}_{G_3(q;\hat{n})} . \]

    It remains to show that $g_q(\hat{n}) \geq G_1(q;\hat{n}) - G_2(q;\hat{n}) - G_3(q;\hat{n}) > 0$ for $1 \leq q \leq 1\!+\!\frac{1}{\hat{n}}$ and large enough $n$. By the Taylor expansion of $G_1(q;\hat{n})$, $G_2(q;\hat{n})$, and $G_3(q;\hat{n})$ at $q=1$, we obtain: 
    \begin{align*}
        g_q(\hat{n}) &\geq G_1(q;\hat{n}) - G_2(q;\hat{n}) - G_3(q;\hat{n})\\
        &\geq \rbra*{\frac{\log (2)}{2}-\frac{1}{4}(2 \hat{n}+1) \log ^2(2)(q-1)} 
            - \frac{\log (2)}{2} \cdot \hat{n} 2^{-\hat{n}^{0.49}}
            - \log(2) \cdot 2^{\frac{1-\hat{n}^{0.49}}{2}} \coloneqq G(q;\hat{n})
    \end{align*}

    Noting that $\frac{\partial}{\partial q} G(q;\hat{n}) = -\frac{1}{4}(2 \hat{n}+1) \log ^2(2) < 0$ for $\hat{n} \geq 1$, we know that $G(q;\hat{n})$ is monotonically decreasing on $q > 1$ for any fixed $\hat{n} \geq 1$. As a consequence, as $1 \leq q \leq 1+\frac{1}{\hat{n}}$, it is left to show that $G\rbra*{1+\frac{1}{\hat{n}};\hat{n}} > 0$ for large enough $\hat{n}$, specifically:
    \begin{equation*}
        G\rbra*{1\!+\!\frac{1}{\hat{n}};\hat{n}} = \frac{\log(2)}{4} \left(2-2\log(2)-2^{1-\hat{n}^{0.49}} \hat{n}- 4 \cdot 2^{\frac{1-\hat{n}^{0.49}}{2}}-\frac{\log (2)}{\hat{n}}\right) > 0.
    \end{equation*}
    
    A direct calculation implies that
    \begin{equation*}
        \frac{\dd}{\dd\hat{n}} G\rbra*{1\!+\!\frac{1}{\hat{n}};\hat{n}} = \frac{\log(2)}{200} \left(49 \sqrt{2^{1-\hat{n}^{0.49}}} \hat{n}^{1.49} \log (2)+2^{-\hat{n}^{0.49}} \hat{n}^2 \left(49 \hat{n}^{0.49} \log (2)-100\right)+50 \log (2)\right).
    \end{equation*}

    Since it is evident that $49 \hat{n}^{0.49} \log (2)-100 > 0$, we can deduce that $\frac{\dd}{\dd\hat{n}} G\rbra*{1\!+\!\frac{1}{\hat{n}};\hat{n}} > 0$. As $49 \hat{n}^{0.49} \log (2)-100 > 0$ holds when $\hat{n} \geq 10$, we obtain that $G\rbra*{1\!+\!\frac{1}{\hat{n}};\hat{n}}$ is monotonically increasing for $\hat{n} \geq 10$. Therefore, we complete the proof by noticing $\hat{n} = n-1$ and the following:
    \[\text{For any } q\in\Big(1,1\!+\!\frac{1}{\hat{n}}\Big] \text{ and } \hat{n}\geq 89,~ g_q(\hat{n}) \geq G(q;\hat{n}) \geq G\rbra*{1\!+\!\frac{1}{\hat{n}};\hat{n}} \geq G\rbra*{1\!+\!\frac{1}{89};89} > \frac{1}{400}. \qedhere\]
    %1/392
\end{proof}

\begin{theorem}[\TsallisQEA{} is \QSZK{}-hard under Turing reduction for $1< q \leq 1+\frac{1}{n-1}$]
    \label{thm:TsallisQEA-QSZKhard}
    For any $q \in \big(1,1+\frac{1}{n-1}\big]$ and any $n\geq 90$, the following holds\emph{:} 
    \[\TsallisQEA \text{ with } g(n)=\Theta(1) \text{ is } \QSZK\text{-hard under Turing reduction}. \]
\end{theorem}

\begin{proof}
    This proof is very similar to the proof of \Cref{thm:ConstRankTsallisQEA-BQPhard}. 
    For any $1 < q \leq 1+\frac{1}{n-1}$ and $n \geq 90$, since $\TsallisQED\sbra*{\hat{g}_q(n)}$ is \QSZK-hard{} under Karp reduction (\Cref{thm:TsallisQED-QSZK-hardness}), where $\hat{g}_q(n) = 1/400$, it suffices to provide an algorithm for $\TsallisQED\sbra*{\hat{g}_q(n)}$ by using $\TsallisQEA[t(n),g(n)]$ as subroutines, with appropriately adaptive choices of $t(n)$ and $g(n)$.

    Let $Q_0$ and $Q_1$ be the corresponding \QSZK{}-hard instance such that these circuits are polynomial-size and prepare the states $\rho_0$ and $\rho_1$, respectively. Let $\TsallisQEAalgo\rbra*{Q, t(n), g(n)}$ be the subroutine for decide whether $\Sq(\rho) \geq t(n)+g(n)$ or $\Sq(\rho) \leq t(n)-g(n)$. Next, we estimate $\Sq(\rho_b)$ to within additive error $\hat{g}_q(n)/2$ for $b\in\binset$ via the procedure $\BiSearch$, as specified in \Cref{algo:Tsallis-entropy-search-to-decision}. 
    To solve $\TsallisQED\sbra*{\hat{g}_q(n)}$, noting that $\max\{\rank(\rho_0),\rank(\rho_1)\} \leq 2^n$, we choose $\tau(n) = \Sq\rbra*{(I/2)^{\otimes n}}$. Subsequently, let $t_0(n) = \BiSearch(Q_0, \tau(n), \hat{g}_q(n))$ and $t_1(n) = \BiSearch(Q_1, \tau(n), \hat{g}_q(n))$, we obtain the same inequalities in \Cref{eq:error-bounds-Turing-reduction}. 

    Note that $\hat{g}_q(n) = 1/400$ for $n \geq 90$ and $\tau(n) \leq \Sq\rbra*{(I/2)^{n}} \leq \S\rbra*{(I/2)^{n}} = n\ln{2}$. 
    Since each query to $\TsallisQEAalgo$ in $\BiSearch$ decreases the size of the interval $[a,b]$ by almost a half, we complete the proof by concluding that the number of adaptive queries to $\TsallisQEAalgo$ in $\BiSearch(Q_0, \tau(n), \hat{g}_q(n))$ and $\BiSearch(Q_1, \tau(n), \hat{g}_q(n))$ is $O\rbra*{\log\rbra*{\tau(n)/\hat{g}_q(n)}} = O(\log{n})$.
\end{proof}

\subsubsection{\NIQSZK{} hardness result}

\begin{theorem}[\TsallisQEA{} is \NIQSZK{}-hard for $q = 1\!+\!\frac{1}{n-1}$]
    \label{thm:TsallisQEA-NIQSZKhard}
    For any $n \geq 5$, it holds that\emph{:}
    \[ \forall g(n) \in[1/\poly(n), 1/150],~ \TsallisQEAnoq_{1+\frac{1}{n-1}} \text{ with } g(n) \text{ is } \NIQSZK{}\text{-hard}. \]
\end{theorem}

\begin{proof}
    Utilizing \Cref{QSCMM-is-NIQSZKhard}, we know that $\QSCMM[1/n, 1-1/n]$ is \NIQSZK{}-hard for $n \geq 3$. Following the reduction from \QSCMM{} to $\TsallisQEAnoq_{1+\frac{1}{n-1}}$ for $n \geq 5$ (\Cref{lemma:reduction-MaxMixedQSD-TsallisQEA}), and the specific choice of $t(n)$ in the reduction, we can conclude that $g(n) \leq 1/150$.
\end{proof}

\subsection{Quantum query complexity lower bounds}
\label{subsec:query-lower-bound}

In this subsection, we present two quantum query complexity lower bounds for estimating the quantum Tsallis entropy $\Sq(\rho)$: When $q$ is constantly larger than $1$, the lower bound is \textit{independent} of the rank of $\rho$ (\Cref{thm:Tsallis-query-complexity-lower-bound-largeQ}). However, when $q>1$ is inverse-polynomially close to $1$ or even closer, the lower bound \textit{depends polynomially} on the rank of $\rho$ (\Cref{thm:Tsallis-query-complexity-lower-bound-smallQ}).

\begin{theorem}[Query complexity lower bound for estimating quantum Tsallis entropy with $q$ constantly above $1$]
    \label{thm:Tsallis-query-complexity-lower-bound-largeQ}
    For any $q \geq 1+\Omega(1)$ and sufficiently small $\epsilon > 0$, the quantum query complexity for estimating the $q$-Tsallis entropy of a quantum state to within additive error $\epsilon$, in the purified quantum query access model, is $\Omega\rbra{1/\sqrt{\epsilon}}$. 
\end{theorem}
\begin{proof}
    Consider the task of distinguishing two quantum unitary operators $U_\epsilon$ and $U_0$ corresponding to two probability distributions $p_\epsilon$ and $p_0$, where $p_x \coloneqq \rbra{1 - x, x}$, $U_x$ is a unitary operator satisfying
    \[
    U_x \ket{0} = \sqrt{1 - x} \ket{0} \ket{\varphi_0} + \sqrt{x} \ket{1} \ket{\varphi_1},
    \]
    with $\ket{\varphi_0}$ and $\ket{\varphi_1}$ being any orthogonal unit vectors. 
    By the quantum query complexity of distinguishing probability distributions given in \cref{lemma:q-distinguish-prob-distri}, we know that distinguishing $U_\epsilon$ and $U_0$ requires quantum query complexity $\Omega\rbra{1/d_{\rm H}\rbra{p_\epsilon, p_0}}$, where $d_{\rm H}\rbra{\cdot, \cdot}$ is the Hellinger distance between two probability distributions. 
    Direct calculation shows that if $\epsilon \in \rbra{0, 1}$, 
    \[
    d_{\rm H}\rbra{p_\epsilon, p_0} = \frac{1}{\sqrt{2}} \sqrt{ \rbra*{\sqrt{1-\epsilon} - 1}^2 + \rbra*{\sqrt{\epsilon} - 0}^2 } \leq \sqrt{\epsilon}.
    \]
    Thus the query complexity of distinguishing $U_0$ and $U_\epsilon$ is $\Omega\rbra{1/\sqrt{\epsilon}}$.

    On the other hand, $U_x$ prepares a purification of $\rho_x \coloneqq \rbra{1-x}\ketbra{0}{0} + x\ketbra{1}{1}$.
    Then, for sufficiently small $\epsilon > 0$, we have
    \[
    \abs*{ \Sq\rbra{\rho_\epsilon} - \Sq\rbra{\rho_0} } = \frac{1 - \rbra{1-\epsilon}^q - \epsilon^q}{q - 1} = \Omega\rbra{\epsilon}.
    \]
    Therefore, any quantum query algorithm that can compute the $q$-Tsallis entropy of a quantum state to within additive error $\Theta\rbra{\epsilon}$ can be used to distinguish $U_\epsilon$ and $U_0$, thus requiring query complexity $\Omega\rbra{1/\sqrt{\epsilon}}$.
\end{proof}

\begin{theorem}[Query complexity lower bound for estimating quantum Tsallis entropy with $q > 1$ near $1$]
    \label{thm:Tsallis-query-complexity-lower-bound-smallQ}
    For any $q \in \big(1,1\!+\!\frac{1}{n-1}\big]$, there exists a mixed quantum state $\rho$ of sufficiently large rank $r$ such that the quantum query complexity for estimating $\Sq(\rho)$, in the purified quantum query access model, is $\Omega\rbra{r^{0.17-c}}$ for any constant $c>0$. 
\end{theorem}

\begin{remark}[$\tau$-dependence in the lower bounds]
    \label{remark:tau-dependence-lower-bound}
    The lower bounds on query and sample complexities in \Cref{thm:Tsallis-query-complexity-lower-bound-smallQ,thm:Tsallis-sample-complexity-lower-bound-smallQ} are $\Omega\rbra[\big]{r^{\frac{1-\tau}{3}-c}}$ and $\Omega\rbra{r^{1-\tau-c'}}$, respectively, where $\tau = 0.49$ is chosen to establish the \QSZK{} hardness (\Cref{thm:TsallisQED-QSZK-hardness}) and $c'=3c$. 
    Notably, these bounds can be further improved by selecting a smaller $\tau$ that still satisfies all requirements in the reduction (\Cref{lemma:reduction-QSD-TsallisQED}), which is left for future work. 
\end{remark}

\begin{proof}[Proof of \Cref{thm:Tsallis-query-complexity-lower-bound-smallQ}]
    By \Cref{lemma:query-complexity-lower-bound-QSD} with $\epsilon=1/2$, there exists an $\hat{n}$-qubit state $\hat{\rho}$ of rank $\hat{r}\geq 2$ and the corresponding ``uniform'' state $\hat{\rho}_{\ttU}$ of rank $r$ on the same support as $\hat{\rho}$ such that the quantum query complexity to decide whether $\td\big(\hat{\rho},\hat{\rho}_{\ttU}\big)$ is at least $1/2$ or exactly $0$ is $\Omega(\hat{r}^{1/3})$. We apply the polarization lemma for the trace distance to the states $\hat{\rho}$ and $\hat{\rho}_{\ttU}$, particularly using only the direct product lemma~\cite[Lemma 8]{Wat02}. Let $\rho \coloneqq \hat{\rho}^{\otimes \hat{r}^{k}}$ and $\rho_{\ttU} \coloneqq \hat{\rho}_{\ttU}^{\otimes \hat{r}^{k}}$ be the resulting states, where $k$ is a parameter to be determined later. 
    Then, for any constant $k > \frac{\tau}{1-\tau}$ with $\tau = 0.49$ and for sufficiently large $\hat{r}$, the following holds:
    \begin{align*}
        \td\rbra*{\hat{\rho}, \hat{\rho}_{\ttU}} &\geq 1/2  &\Rightarrow \quad \td\rbra*{\rho, \rho_{\ttU}} &\geq 1-\exp\rbra{-\hat{r}^{k}/8} \geq 1-2^{-r^{\tau}},\\
        \td\rbra*{\hat{\rho}, \hat{\rho}_{\ttU}} &= 0  &\Rightarrow \quad \td\rbra*{\rho, \rho_{\ttU}} &\leq \hat{r}^{k} \cdot 0 = 0 \leq 2^{-r^{\tau}}.
    \end{align*}
    
    Hence, the query complexity of deciding whether $\td(\rho,\rho_{\ttU})$ is at least $1-2^{-r^{\tau}}$ or at most $2^{-r^{\tau}}$ is $\Omega\rbra[\big]{r^{\frac{1}{3(1+k)}}}$, where $r \coloneqq \hat{r}\cdot \hat{r}^{k} = \hat{r}^{1+k}$. For any $q \in \big(1, 1+\frac{1}{n-1}\big] \subseteq \big(1, 1+\frac{1}{r-1}\big]$, using the reduction from \QSD{} to \TsallisQED{} (\Cref{lemma:reduction-QSD-TsallisQED}) with parameters from \Cref{thm:TsallisQED-QSZK-hardness}, there are two corresponding states $\rho'_0$ and $\rho'_1$ of rank at most $2r$ such that the quantum query complexity for deciding whether $\Sq(\rho'_0)-\Sq(\rho'_1)$ is at least $1/400$ or at most $-1/400$ is $\Omega\rbra[\big]{r^{\frac{1}{3(1+k)}}} = \Omega\rbra[\big]{r^{\frac{1-\tau}{3}-c}} = \Omega\rbra{r^{0.17-c}}$ for any constant $c > 0$. 
    Thus, estimating $\Sq(\rho'_b)$ for $b \in \binset$ to within additive error $1/800$ requires at least the same number of quantum queries. 
\end{proof}

\subsection{Quantum sample complexity lower bounds}
\label{subsec:sample-lower-bound}

In this subsection, we present two quantum sample complexity lower bounds for estimating the quantum Tsallis entropy $\Sq(\rho)$: When $q$ is constantly larger than $1$, the lower bound is \textit{independent} of the rank of $\rho$ (\Cref{thm:Tsallis-sample-complexity-lower-bound-largeQ}). However, when $q>1$ is inverse-polynomially close to $1$, the lower bound \textit{depends polynomially} on the rank of $\rho$ (\Cref{thm:Tsallis-sample-complexity-lower-bound-smallQ}).

\begin{theorem}[Sample complexity lower bound for estimating quantum Tsallis entropy with $q$ constantly above $1$]
    \label{thm:Tsallis-sample-complexity-lower-bound-largeQ}
    For any $q \geq 1+\Omega(1)$ and sufficiently small $\epsilon > 0$, the quantum sample complexity for estimating the quantum $q$-Tsallis entropy of a quantum state to within additive error $\epsilon$ is $\Omega\rbra{1/{\epsilon}}$. 
\end{theorem}

\begin{proof}
    Consider the hypothesis testing problem where the given quantum state $\rho$ is promised to be either $\rho_0$ or $\rho_\epsilon$, each with equal probability. Specifically, the states are defined as
    \[
    \forall x\in[0,1], \quad \rho_x \coloneqq \rbra{1-x} \ketbra{0}{0} + x \ketbra{1}{1}.
    \]
    For sufficiently small $\epsilon > 0$, we know that $\abs{ \Sq\rbra{\rho_\epsilon} - \Sq\rbra{\rho_0} } = \Omega\rbra{\epsilon}$, as shown in the proof of \cref{thm:Tsallis-query-complexity-lower-bound-largeQ}. 
    Now, assume that there is a quantum estimator for $\Sq\rbra{\rho}$ to within additive error $\Theta\rbra{\epsilon}$ with sample complexity $\mathsf{S}$. 
    This estimator can then be used to distinguish these two states $\rho_0$ and $\rho_\epsilon$ with success probability $p_{\text{succ}} \geq 2/3$. 
    On the other hand, by \cref{lemma:Holevo-Helstrom-bound}, we have 
    \[
    p_{\text{succ}} \leq \frac 1 2 + \frac 1 2 \td\rbra*{\rho_0^{\otimes \mathsf{S}}, \rho_\epsilon^{\otimes \mathsf{S}}}.
    \]
    By applying the Fuchs–van de Graaf inequalities \cite[Theorem 1]{FvdG99}, we have
    \[
    \td\rbra*{\rho_0^{\otimes \mathsf{S}}, \rho_\epsilon^{\otimes \mathsf{S}}} \leq \sqrt{1 - \mathrm{F}\rbra*{\rho_0^{\otimes \mathsf{S}}, \rho_\epsilon^{\otimes \mathsf{S}}}^2},
    \]
    where $\mathrm{F}\rbra{\rho, \sigma} = \tr\rbra{\sqrt{\sqrt{\sigma}\rho\sqrt{\sigma}}}$ is the fidelity of quantum states. 
    A direct calculation shows that $\mathrm{F}\rbra{\rho_0, \rho_\epsilon} = \sqrt{1-\epsilon}$, which gives that
    \[
    p_{\text{succ}} \leq \frac 1 2 + \frac 1 2 \sqrt{1 - \rbra{1-\epsilon}^{\mathsf{S}}}.
    \]
    By combining this with the condition $p_{\text{succ}} \geq 2/3$, we conclude that $\mathsf{S} = \Omega\rbra{1/\epsilon}$.
\end{proof}

\begin{theorem}[Sample complexity lower bound for estimating quantum Tsallis entropy with $q>1$ near $1$]
    \label{thm:Tsallis-sample-complexity-lower-bound-smallQ}
    For any $q \in \big(1,1\!+\!\frac{1}{n-1}\big]$, there exists a mixed quantum state $\rho$ of sufficiently large rank $r$ such that the quantum sample complexity for estimating $\Sq(\rho)$ is $\Omega(r^{0.51-c})$ for any constant $c > 0$. 
\end{theorem}

Notably, \Cref{remark:tau-dependence-lower-bound} on the $\tau$-dependence in the lower bound also applies to \Cref{thm:Tsallis-sample-complexity-lower-bound-smallQ}. Moreover, the proof strategy of \Cref{thm:Tsallis-sample-complexity-lower-bound-smallQ} is similar to that of \Cref{thm:Tsallis-query-complexity-lower-bound-smallQ}, as both rely on the direct product lemma for the trace distance~\cite[Lemma 8]{Wat02}.\footnote{This inequalities can also be derived using the polarization lemma for the measured quantum triangular discrimination, specifically combining Theorem 3.3 and Lemma 4.11 in~\cite{Liu23}.}

\begin{proof}[Proof of \Cref{thm:Tsallis-sample-complexity-lower-bound-smallQ}]
    By \Cref{lemma:sample-complexity-lower-bound-QSD} with $\epsilon=1/2$, there exists an $\hat{n}$-qubit state $\hat{\rho}$ of rank $\hat{r}\geq 2$ and the corresponding ``uniform'' state $\hat{\rho}_{\ttU}$ of rank $r$ on the same support as $\hat{\rho}$ such that the quantum sample complexity to decide whether $\td\big(\hat{\rho},\hat{\rho}_{\ttU}\big)$ is at least $1/2$ or exactly $0$ is $\Omega(\hat{r})$. 
    We apply the direct product lemma~\cite[Lemma 8]{Wat02} to the states $\hat{\rho}$ and $\hat{\rho}_{\ttU}$. 
    Let $\rho \coloneqq \hat{\rho}^{\otimes \hat{r}^{k}}$ and $\rho_{\ttU} \coloneqq \hat{\rho}_{\ttU}^{\otimes \hat{r}^{k}}$ be the resulting states, where $k$ is a parameter to be determined later. 
    Then, for any constant $k > \frac{\tau}{1-\tau}$ with $\tau = 0.49$ and for sufficiently large $\hat{r}$, the following holds:
    \begin{align*}
        \td\rbra*{\hat{\rho}, \hat{\rho}_{\ttU}} &\geq 1/2  &\Rightarrow \quad \td\rbra*{\rho, \rho_{\ttU}} &\geq 1-\exp\rbra{-\hat{r}^{k}/8} \geq 1-2^{-r^{\tau}},\\
        \td\rbra*{\hat{\rho}, \hat{\rho}_{\ttU}} &= 0  &\Rightarrow \quad \td\rbra*{\rho, \rho_{\ttU}} &\leq \hat{r}^{k} \cdot 0 = 0 \leq 2^{-r^{\tau}}.
    \end{align*}
    
    As a consequence, the sample complexity of deciding whether $\td(\rho,\rho_{\ttU})$ is at least $1-2^{-r^{\tau}}$ or at most $2^{-r^{\tau}}$ is $\Omega(r^{\frac{1}{1+k}})$, where $r \coloneqq \hat{r}\cdot \hat{r}^{k} = \hat{r}^{1+k}$. For any $q \in \big(1, 1+\frac{1}{n-1}\big] \subseteq \big(1, 1+\frac{1}{r-1}\big]$, utilizing the reduction from \QSD{} to \TsallisQED{} (\Cref{lemma:reduction-QSD-TsallisQED}) with parameters from \Cref{thm:TsallisQED-QSZK-hardness}, there are two corresponding states $\rho'_0$ and $\rho'_1$ of rank at most $2r$ such that the quantum sample complexity for deciding whether $\Sq(\rho'_0)-\Sq(\rho'_1)$ is at least $1/400$ or at most $-1/400$ is $\Omega(r^{\frac{1}{1+k}}) = \Omega\rbra{r^{1-\tau-c}} = \Omega\rbra{r^{0.51-c}}$ for any constant $c > 0$. 
    Therefore, estimating $\Sq(\rho'_b)$ for $b \in \binset$ to within additive error $1/800$ requires at least the same number of copies of $\rho$. 
\end{proof}

%%%%%%%%%%%%%%%%%%%%%%%%%%%%%%%%%%%%%%%%%%%

\section*{Acknowledgments}
\noindent
A preliminary version of this article, under the same title, appeared in the Proceedings of the 2025 Annual ACM-SIAM Symposium on Discrete Algorithms (SODA 2025)~\cite{LW25SODA}, and was also included in Yupan Liu's PhD thesis~\cite{Liu25}.

The authors express their gratitude to the anonymous reviewers for their valuable comments, particularly for suggesting the inclusion of further explanation of the importance of estimating $\Sq(\rho)$ for non-integer $q>1$, and for identifying an error in the proof of \Cref{lemma:Tsallis-binary-entropy-lower-bound} concerning the evaluation of the endpoints, as well as a gap in the reasoning establishing \Cref{eq:Tsallis-binary-entropy-upper-bound-inner-condition} in the proof of \Cref{lemma:Tsallis-binary-entropy-upper-bound} in an earlier version.
The authors also thank Kean Chen for pointing out a parameter misuse in \cref{thm:tr-power-constant-queries,thm:tr-power-constant-samples} in a previous version of this paper. 

This work was supported in part by the Ministry of Education, Culture, Sports, Science and Technology (MEXT) Quantum Leap Flagship Program (MEXT Q-LEAP) under Grant \mbox{JPMXS0120319794}. 
The work of Yupan Liu was also supported in part by JST, the establishment of University fellowships towards the creation of science technology innovation, under Grant \mbox{JPMJFS2125}, and in part by funding from the Swiss State Secretariat for Education, Research and Innovation (SERI). 
The work of Qisheng Wang was also supported in part by startup funding from Shanghai Jiao Tong University and in part by the Engineering and Physical Sciences Research Council under Grant \mbox{EP/X026167/1}. 

%%%%%%%%%%%%%%%%%%%%%%%%%%%%%%%%%%%%%%%%%%%

\bibliographystyle{alphaurl}
\bibliography{main}

\newcommand{\etalchar}[1]{$^{#1}$}
\newcommand{\prelimVersion}[2]{Preliminary version in \textit{\MakeUppercase{#1} #2}}\newcommand{\ECCC}[3]{\href{https://eccc.weizmann.ac.il/report/#1#2/#3/}{\texttt{ECCC:TR{#2}-{#3}}}}\newcommand{\IACR}[1]{\href{https://eprint.iacr.org/#1}{\texttt{IACR ePrint:#1}}}\newcommand\hyphen{-}\DeclareRobustCommand{\dutchPrefix}[2]{#2}\providecommand{\dutchPrefix}[2]{#2}\renewcommand{\dutchPrefix}[2]{#2}
\begin{thebibliography}{{\dutchPrefix{Apeldoorn}{v}}ACGN23}

\bibitem[AISW20]{AISW20}
Jayadev Acharya, Ibrahim Issa, Nirmal~V. Shende, and Aaron~B. Wagner.
\newblock Estimating quantum entropy.
\newblock {\em IEEE Journal on Selected Areas in Information Theory}, 1(2):454--468, 2020.
\newblock \prelimVersion{ISIT}{2019}.
\newblock \href {https://arxiv.org/abs/1711.00814} {\path{arXiv:1711.00814}}, \href {https://doi.org/10.1109/JSAIT.2020.3015235} {\path{doi:10.1109/JSAIT.2020.3015235}}.

\bibitem[AJL09]{AJL09}
Dorit Aharonov, Vaughan Jones, and Zeph Landau.
\newblock A polynomial quantum algorithm for approximating the {Jones} polynomial.
\newblock {\em Algorithmica}, 55(3):395--421, 2009.
\newblock \prelimVersion{STOC}{2006}.
\newblock \href {https://arxiv.org/abs/quant-ph/0511096} {\path{arXiv:quant-ph/0511096}}, \href {https://doi.org/10.1007/s00453-008-9168-0} {\path{doi:10.1007/s00453-008-9168-0}}.

\bibitem[ALL22]{ALL22}
Anurag Anshu, Zeph Landau, and Yunchao Liu.
\newblock Distributed quantum inner product estimation.
\newblock In {\em Proceedings of the 54th Annual ACM SIGACT Symposium on Theory of Computing}, pages 44--51, 2022.
\newblock \href {https://arxiv.org/abs/2111.03273} {\path{arXiv:2111.03273}}, \href {https://doi.org/10.1145/3519935.3519974} {\path{doi:10.1145/3519935.3519974}}.

\bibitem[Amb14]{Ambainis14}
Andris Ambainis.
\newblock On physical problems that are slightly more difficult than $\mathsf{QMA}$.
\newblock In {\em 2014 IEEE 29th Conference on Computational Complexity (CCC)}, pages 32--43. IEEE, 2014.
\newblock \href {https://arxiv.org/abs/1312.4758} {\path{arXiv:1312.4758}}, \href {https://doi.org/10.1109/ccc.2014.12} {\path{doi:10.1109/ccc.2014.12}}.

\bibitem[AOST17]{AOST17}
Jayadev Acharya, Alon Orlitsky, Ananda~Theertha Suresh, and Himanshu Tyagi.
\newblock Estimating {Renyi} entropy of discrete distributions.
\newblock {\em IEEE Transactions on Information Theory}, 63(1):38--56, 2017.
\newblock \href {https://arxiv.org/abs/1408.1000} {\path{arXiv:1408.1000}}, \href {https://doi.org/10.1109/TIT.2016.2620435} {\path{doi:10.1109/TIT.2016.2620435}}.

\bibitem[{\dutchPrefix{Apeldoorn}{v}}ACGN23]{vACGN23}
Joran {\dutchPrefix{Apeldoorn}{v}}an~Apeldoorn, Arjan Cornelissen, Andr{\'{a}}s Gily{\'{e}}n, and Giacomo Nannicini.
\newblock Quantum tomography using state-preparation unitaries.
\newblock In {\em Proceedings of the 2023 Annual ACM-SIAM Symposium on Discrete Algorithms (SODA)}, pages 1265--1318, 2023.
\newblock \href {https://arxiv.org/abs/2207.08800} {\path{arXiv:2207.08800}}, \href {https://doi.org/10.1137/1.9781611977554.ch47} {\path{doi:10.1137/1.9781611977554.ch47}}.

\bibitem[AS17]{AS17}
Guillaume Aubrun and Stanis{\l}aw~J Szarek.
\newblock {\em Alice and Bob Meet Banach: The Interface of Asymptotic Geometric Analysis and Quantum Information Theory}, volume 223 of {\em Mathematical Surveys and Monographs}.
\newblock American Mathematical Society, 2017.
\newblock \href {https://doi.org/10.1090/surv/223} {\path{doi:10.1090/surv/223}}.

\bibitem[AS25]{AS24}
Srinivasan Arunachalam and Louis Schatzki.
\newblock Generalized inner product estimation with limited quantum communication.
\newblock In {\em Proceedings of the 42nd International Symposium on Theoretical Aspects of Computer Science ({STACS} 2025)}, volume 327 of {\em LIPIcs}, pages 11:1--11:17. Schloss Dagstuhl - Leibniz-Zentrum f{\"{u}}r Informatik, 2025.
\newblock \href {https://arxiv.org/abs/2410.12684} {\path{arXiv:2410.12684}}, \href {https://doi.org/10.4230/LIPICS.STACS.2025.11} {\path{doi:10.4230/LIPICS.STACS.2025.11}}.

\bibitem[BASTS10]{BASTS10}
Avraham Ben-Aroya, Oded Schwartz, and Amnon Ta-Shma.
\newblock Quantum expanders: motivation and construction.
\newblock {\em Theory of Computing}, 6(3):47--79, 2010.
\newblock \prelimVersion{CCC}{2008}.
\newblock \href {https://doi.org/10.4086/toc.2010.v006a003} {\path{doi:10.4086/toc.2010.v006a003}}.

\bibitem[Bau11]{Bau11}
Bernhard Baumgartner.
\newblock An inequality for the trace of matrix products, using absolute values.
\newblock ArXiv e-prints, 2011.
\newblock \href {https://arxiv.org/abs/1106.6189} {\path{arXiv:1106.6189}}.

\bibitem[BBBV97]{BBBV97}
Charles~H. Bennett, Ethan Bernstein, Gilles Brassard, and Umesh Vazirani.
\newblock Strengths and weaknesses of quantum computing.
\newblock {\em SIAM Journal on Computing}, 26(5):1510--1523, 1997.
\newblock \href {https://arxiv.org/abs/quant-ph/9701001} {\path{arXiv:quant-ph/9701001}}, \href {https://doi.org/10.1137/S0097539796300933} {\path{doi:10.1137/S0097539796300933}}.

\bibitem[BCH{\etalchar{+}}19]{BCHTV19}
Adam Bouland, Lijie Chen, Dhiraj Holden, Justin Thaler, and Prashant~Nalini Vasudevan.
\newblock On the power of statistical zero knowledge.
\newblock {\em SIAM Journal on Computing}, 49(4):{FOCS17\hyphen1}--{FOCS17\hyphen58}, 2019.
\newblock \prelimVersion{FOCS}{2017}.
\newblock \href {https://arxiv.org/abs/1609.02888} {\path{arXiv:1609.02888}}, \href {https://doi.org/10.1137/17M1161749} {\path{doi:10.1137/17M1161749}}.

\bibitem[BCW{\dutchPrefix{Wolf}{d}}W01]{BCWdW01}
Harry Buhrman, Richard Cleve, John Watrous, and Ronald {\dutchPrefix{Wolf}{d}}e~Wolf.
\newblock Quantum fingerprinting.
\newblock {\em Physical Review Letters}, 87(16):167902, 2001.
\newblock \href {https://arxiv.org/abs/quant-ph/0102001} {\path{arXiv:quant-ph/0102001}}, \href {https://doi.org/10.1103/PhysRevLett.87.167902} {\path{doi:10.1103/PhysRevLett.87.167902}}.

\bibitem[BDRV19]{BDRV19}
Itay Berman, Akshay Degwekar, Ron~D Rothblum, and Prashant~Nalini Vasudevan.
\newblock Statistical difference beyond the polarizing regime.
\newblock In {\em Proceedings of the 17th International Conference on Theory of Cryptography Conference}, pages 311--332. Springer, 2019.
\newblock \ECCC{20}{19}{038}.
\newblock \href {https://doi.org/10.1007/978-3-030-36033-7\_12} {\path{doi:10.1007/978-3-030-36033-7\_12}}.

\bibitem[Bec02]{Beck02}
Christian Beck.
\newblock Generalized statistical mechanics and fully developed turbulence.
\newblock {\em Physica A: Statistical Mechanics and its Applications}, 306:189--198, 2002.
\newblock \href {https://arxiv.org/abs/cond-mat/0110073} {\path{arXiv:cond-mat/0110073}}, \href {https://doi.org/10.1016/s0378-4371(02)00497-1} {\path{doi:10.1016/s0378-4371(02)00497-1}}.

\bibitem[Bel19]{Belovs19}
Aleksandrs Belovs.
\newblock Quantum algorithms for classical probability distributions.
\newblock In {\em Proceedings of the 27th Annual European Symposium on Algorithms, {ESA} 2019}, volume 144 of {\em LIPIcs}, pages 16:1--16:11. Schloss Dagstuhl - Leibniz-Zentrum f{\"{u}}r Informatik, 2019.
\newblock \href {https://arxiv.org/abs/1904.02192} {\path{arXiv:1904.02192}}, \href {https://doi.org/10.4230/LIPICS.ESA.2019.16} {\path{doi:10.4230/LIPICS.ESA.2019.16}}.

\bibitem[Ber14]{Bernstein14}
Serge Bernstein.
\newblock Sur la meilleure approximation de $|x|$ par des polynomes de degr{\'e}s donn{\'e}s.
\newblock {\em Acta Mathematica}, 37(1):1--57, 1914.
\newblock \href {https://doi.org/10.1007/BF02401828} {\path{doi:10.1007/BF02401828}}.

\bibitem[Ber38]{Bernstein38}
Serge Bernstein.
\newblock Sur la meilleure approximation de $|x|^p$ par des polynômes de degrés très élevés.
\newblock {\em Izvestiya Akademii Nauk SSSR. Seriya Matematicheskaya}, 2(2):169--190, 1938.
\newblock URL: \url{https://www.mathnet.ru/eng/im3513}.

\bibitem[BH09]{BH09}
Jop Bri{\"e}t and Peter Harremo{\"e}s.
\newblock Properties of classical and quantum {Jensen-Shannon} divergence.
\newblock {\em Physical Review A}, 79(5):052311, 2009.
\newblock \href {https://arxiv.org/abs/0806.4472} {\path{arXiv:0806.4472}}, \href {https://doi.org/10.1103/PhysRevA.79.052311} {\path{doi:10.1103/PhysRevA.79.052311}}.

\bibitem[BHH11]{BHH11}
Sergey Bravyi, Aram~W. Harrow, and Avinatan Hassidim.
\newblock Quantum algorithms for testing properties of distributions.
\newblock {\em IEEE Transactions on Information Theory}, 57(6):3971--3981, 2011.
\newblock \prelimVersion{STACS}{2010}.
\newblock \href {https://arxiv.org/abs/0907.3920} {\path{arXiv:0907.3920}}, \href {https://doi.org/10.1109/TIT.2011.2134250} {\path{doi:10.1109/TIT.2011.2134250}}.

\bibitem[BHMT02]{BHMT02}
Gilles Brassard, Peter H{\o}yer, Michele Mosca, and Alain Tapp.
\newblock Quantum amplitude amplification and estimation.
\newblock In Samuel~J. Lomonaco, Jr. and Howard~E. Brandt, editors, {\em Quantum Computation and Information}, volume 305 of {\em Contemporary Mathematics}, pages 53--74. AMS, 2002.
\newblock \href {https://arxiv.org/abs/quant-ph/0005055} {\path{arXiv:quant-ph/0005055}}, \href {https://doi.org/10.1090/conm/305/05215} {\path{doi:10.1090/conm/305/05215}}.

\bibitem[BKT20]{BKT20}
Mark Bun, Robin Kothari, and Justin Thaler.
\newblock The polynomial method strikes back: tight quantum query bounds via dual polynomials.
\newblock {\em Theory of Computing}, 16(10):1--71, 2020.
\newblock \prelimVersion{STOC}{2018}.
\newblock \href {https://arxiv.org/abs/1710.09079} {\path{arXiv:1710.09079}}, \href {https://doi.org/10.4086/toc.2020.v016a010} {\path{doi:10.4086/toc.2020.v016a010}}.

\bibitem[BOW19]{BOW19}
Costin B{\u{a}}descu, Ryan O'Donnell, and John Wright.
\newblock Quantum state certification.
\newblock In {\em Proceedings of the 51st Annual ACM SIGACT Symposium on Theory of Computing}, pages 503--514, 2019.
\newblock \href {https://arxiv.org/abs/1708.06002} {\path{arXiv:1708.06002}}, \href {https://doi.org/10.1145/3313276.3316344} {\path{doi:10.1145/3313276.3316344}}.

\bibitem[BR82]{BR82}
Jacob Burbea and Calyampudi~Radhakrishna Rao.
\newblock On the convexity of some divergence measures based on entropy functions.
\newblock {\em IEEE Transactions on Information Theory}, 28(3):489--495, 1982.
\newblock \href {https://doi.org/10.1109/tit.1982.1056497} {\path{doi:10.1109/tit.1982.1056497}}.

\bibitem[Can20]{Canonne20}
Cl{\'e}ment~L. Canonne.
\newblock A survey on distribution testing: your data is big. but is it blue?
\newblock In {\em Theory of Computing Library}, number~9 in Graduate Surveys, pages 1--100. University of Chicago, 2020.
\newblock \ECCC{20}{15}{063}.
\newblock \href {https://doi.org/10.4086/toc.gs.2020.009} {\path{doi:10.4086/toc.gs.2020.009}}.

\bibitem[CCKV08]{CCKV08}
Andr{\'e} Chailloux, Dragos~Florin Ciocan, Iordanis Kerenidis, and Salil Vadhan.
\newblock Interactive and noninteractive zero knowledge are equivalent in the help model.
\newblock In {\em Proceedings of the Fifth Theory of Cryptography Conference}, pages 501--534. Springer, 2008.
\newblock \IACR{2007/467}.
\newblock \href {https://doi.org/10.1007/978-3-540-78524-8_28} {\path{doi:10.1007/978-3-540-78524-8_28}}.

\bibitem[CFM{\dutchPrefix{Wolf}{d}}W10]{CFMdW10}
Sourav Chakraborty, Eldar Fischer, Arie Matsliah, and Ronald {\dutchPrefix{Wolf}{d}}e~Wolf.
\newblock New results on quantum property testing.
\newblock In {\em IARCS Annual Conference on Foundations of Software Technology and Theoretical Computer Science (FSTTCS 2010)}, volume~8 of {\em LIPIcs}, pages 145--156. Schloss Dagstuhl - Leibniz-Zentrum f{\"{u}}r Informatik, 2010.
\newblock \href {https://arxiv.org/abs/1005.0523} {\path{arXiv:1005.0523}}, \href {https://doi.org/10.4230/LIPICS.FSTTCS.2010.145} {\path{doi:10.4230/LIPICS.FSTTCS.2010.145}}.

\bibitem[CLW20]{CLW20}
Anirban~N. Chowdhury, Guang~Hao Low, and Nathan Wiebe.
\newblock A variational quantum algorithm for preparing quantum {Gibbs} states.
\newblock ArXiv e-prints, 2020.
\newblock \href {https://arxiv.org/abs/2002.00055} {\path{arXiv:2002.00055}}.

\bibitem[CLW26]{CLW26}
Kean Chen, Yupan Liu, and Qisheng Wang.
\newblock Trace estimation of quantum state powers: Sample complexity and computational hardness, 2026.
\newblock \href {https://arxiv.org/abs/2505.09563v2} {\path{arXiv:2505.09563v2}}.

\bibitem[CM18]{CM18}
Chris Cade and Ashley Montanaro.
\newblock The quantum complexity of computing {Schatten} $p$-norms.
\newblock In {\em Proceedings of the 13th Conference on the Theory of Quantum Computation, Communication and Cryptography}, volume 111 of {\em LIPIcs}, pages 4:1--4:20. Schloss Dagstuhl - Leibniz-Zentrum f{\"{u}}r Informatik, 2018.
\newblock \href {https://arxiv.org/abs/1706.09279} {\path{arXiv:1706.09279}}, \href {https://doi.org/10.4230/LIPICS.TQC.2018.4} {\path{doi:10.4230/LIPICS.TQC.2018.4}}.

\bibitem[CT14]{CT14}
Richard~Y. Chen and Joel~A. Tropp.
\newblock Subadditivity of matrix $\phi$-entropy and concentration of random matrices.
\newblock {\em Electronic Journal of Probability}, 19(27):1--30, 2014.
\newblock \href {https://arxiv.org/abs/1308.2952} {\path{arXiv:1308.2952}}, \href {https://doi.org/10.1214/ejp.v19-2964} {\path{doi:10.1214/ejp.v19-2964}}.

\bibitem[CW25]{CW25}
Kean Chen and Qisheng Wang.
\newblock Improved sample upper and lower bounds for trace estimation of quantum state powers.
\newblock In {\em The 38th Annual Conference on Learning Theory}, volume 291 of {\em Proceedings of Machine Learning Research}, pages 1008--1028. {PMLR}, 2025.
\newblock URL: \url{https://proceedings.mlr.press/v291/chen25d.html}.

\bibitem[CWLY23]{CWLY23}
Kean Chen, Qisheng Wang, Peixun Long, and Mingsheng Ying.
\newblock Unitarity estimation for quantum channels.
\newblock {\em IEEE Transactions on Information Theory}, 69(8):5116--5134, 2023.
\newblock \href {https://arxiv.org/abs/2212.09319} {\path{arXiv:2212.09319}}, \href {https://doi.org/10.1109/TIT.2023.3263645} {\path{doi:10.1109/TIT.2023.3263645}}.

\bibitem[CWYZ25]{CWYZ25}
Kean Chen, Qisheng Wang, Zhan Yu, and Zhicheng Zhang.
\newblock Simultaneous estimation of nonlinear functionals of a quantum state.
\newblock ArXiv e-prints, 2025.
\newblock \href {https://arxiv.org/abs/2505.16715} {\path{arXiv:2505.16715}}.

\bibitem[Dar70]{Daroczy70}
Zolt{\'a}n Dar{\'o}czy.
\newblock Generalized information functions.
\newblock {\em Information and Control}, 16(1):36--51, 1970.
\newblock \href {https://doi.org/10.1016/s0019-9958(70)80040-7} {\path{doi:10.1016/s0019-9958(70)80040-7}}.

\bibitem[EAO{\etalchar{+}}02]{EAO+02}
Artur~K. Ekert, Carolina~Moura Alves, Daniel K.~L. Oi, Micha{\l} Horodecki, Pawe{\l} Horodecki, and Leong~Chuan Kwek.
\newblock Direct estimations of linear and nonlinear functionals of a quantum state.
\newblock {\em Physical Review Letters}, 88(21):217901, 2002.
\newblock \href {https://arxiv.org/abs/quant-ph/0203016} {\path{arXiv:quant-ph/0203016}}, \href {https://doi.org/10.1103/physrevlett.88.217901} {\path{doi:10.1103/physrevlett.88.217901}}.

\bibitem[F{\dutchPrefix{Graaf}{v}}dG99]{FvdG99}
Christopher~A. Fuchs and Jeroen {\dutchPrefix{Graaf}{v}}an~de Graaf.
\newblock Cryptographic distinguishability measures for quantum-mechanical states.
\newblock {\em IEEE Transactions on Information Theory}, 45(4):1216--1227, 1999.
\newblock \href {https://arxiv.org/abs/quant-ph/9712042} {\path{arXiv:quant-ph/9712042}}, \href {https://doi.org/10.1109/18.761271} {\path{doi:10.1109/18.761271}}.

\bibitem[Fur05]{Furuichi05}
Shigeru Furuichi.
\newblock On uniqueness theorems for {Tsallis} entropy and {Tsallis} relative entropy.
\newblock {\em IEEE Transactions on Information Theory}, 51(10):3638--3645, 2005.
\newblock \href {https://arxiv.org/abs/cond-mat/0410270} {\path{arXiv:cond-mat/0410270}}, \href {https://doi.org/10.1109/tit.2005.855606} {\path{doi:10.1109/tit.2005.855606}}.

\bibitem[FYK04]{FYK04}
Shigeru Furuichi, Kenjiro Yanagi, and Ken Kuriyama.
\newblock Fundamental properties of {Tsallis} relative entropy.
\newblock {\em Journal of Mathematical Physics}, 45(12):4868--4877, 2004.
\newblock \href {https://arxiv.org/abs/cond-mat/0406178} {\path{arXiv:cond-mat/0406178}}, \href {https://doi.org/10.1063/1.1805729} {\path{doi:10.1063/1.1805729}}.

\bibitem[FYK07]{FYK07}
Shigeru Furuichi, Kenjiro Yanagi, and Ken Kuriyama.
\newblock A generalized {Fannes'} inequality.
\newblock {\em Journal of Inequalities in Pure and Applied Mathematics}, 8(1):5, 2007.
\newblock URL: \url{https://www.emis.de/journals/JIPAM/article818.html?sid=818}, \href {https://arxiv.org/abs/1001.1390} {\path{arXiv:1001.1390}}.

\bibitem[Gan02]{Ganzburg02}
Michael~I. Ganzburg.
\newblock The {Bernstein} constant and polynomial interpolation at the {Chebyshev} nodes.
\newblock {\em Journal of Approximation Theory}, 119(2):193--213, 2002.
\newblock \href {https://doi.org/10.1006/jath.2002.3729} {\path{doi:10.1006/jath.2002.3729}}.

\bibitem[GH20]{GH20}
Alexandru Gheorghiu and Matty~J. Hoban.
\newblock Estimating the entropy of shallow circuit outputs is hard.
\newblock ArXiv e-prints, 2020.
\newblock \href {https://arxiv.org/abs/2002.12814} {\path{arXiv:2002.12814}}.

\bibitem[GHS21]{GHS21}
Tom Gur, Min-Hsiu Hsieh, and Sathyawageeswar Subramanian.
\newblock Sublinear quantum algorithms for estimating von {Neumann} entropy.
\newblock ArXiv e-prints, 2021.
\newblock \href {https://arxiv.org/abs/2111.11139} {\path{arXiv:2111.11139}}.

\bibitem[GHYZ26]{GHYZ24}
Weiyuan Gong, Jonas Haferkamp, Qi~Ye, and Zhihan Zhang.
\newblock Sample complexity of purity and inner product estimation.
\newblock {\em Physics Review Research}, pages~--, 2026.
\newblock \href {https://arxiv.org/abs/2410.12712} {\path{arXiv:2410.12712}}, \href {https://doi.org/10.1103/kh2j-jkn6} {\path{doi:10.1103/kh2j-jkn6}}.

\bibitem[GL20]{GL20}
Andr{\'a}s Gily{\'e}n and Tongyang Li.
\newblock Distributional property testing in a quantum world.
\newblock In {\em 11th Innovations in Theoretical Computer Science Conference (ITCS 2020)}, volume 151 of {\em LIPIcs}, pages 25:1--25:19. Schloss Dagstuhl - Leibniz-Zentrum f{\"{u}}r Informatik, 2020.
\newblock \href {https://arxiv.org/abs/1902.00814} {\path{arXiv:1902.00814}}, \href {https://doi.org/10.4230/LIPIcs.ITCS.2020.25} {\path{doi:10.4230/LIPIcs.ITCS.2020.25}}.

\bibitem[Gol08]{Goldreich08}
Oded Goldreich.
\newblock {\em Computational Complexity: A Conceptual Perspective}.
\newblock Cambridge University Press, 2008.
\newblock \href {https://doi.org/10.1017/CBO9780511804106} {\path{doi:10.1017/CBO9780511804106}}.

\bibitem[Gol17]{Goldreich17}
Oded Goldreich.
\newblock {\em Introduction to Property Testing}.
\newblock Cambridge University Press, 2017.
\newblock \href {https://doi.org/10.1017/9781108135252} {\path{doi:10.1017/9781108135252}}.

\bibitem[GP22]{GP22}
Andr\'{a}s {Gily\'{e}n} and Alexander Poremba.
\newblock Improved quantum algorithms for fidelity estimation.
\newblock ArXiv e-prints, 2022.
\newblock \href {https://arxiv.org/abs/2203.15993} {\path{arXiv:2203.15993}}.

\bibitem[GSLW19]{GSLW19}
Andr\'{a}s Gily\'{e}n, Yuan Su, Guang~Hao Low, and Nathan Wiebe.
\newblock Quantum singular value transformation and beyond: exponential improvements for quantum matrix arithmetics.
\newblock In {\em Proceedings of the 51st Annual ACM SIGACT Symposium on Theory of Computing}, pages 193--204, 2019.
\newblock \href {https://arxiv.org/abs/1806.01838} {\path{arXiv:1806.01838}}, \href {https://doi.org/10.1145/3313276.3316366} {\path{doi:10.1145/3313276.3316366}}.

\bibitem[Hay25]{Hayashi24}
Masahito Hayashi.
\newblock Measuring quantum relative entropy with finite-size effect.
\newblock {\em Quantum}, 9:1725, 2025.
\newblock \href {https://arxiv.org/abs/2406.17299} {\path{arXiv:2406.17299}}, \href {https://doi.org/10.22331/q-2025-05-05-1725} {\path{doi:10.22331/q-2025-05-05-1725}}.

\bibitem[HC67]{HC67}
Jan Havrda and Franti{\v{s}}ek Charv{\'a}t.
\newblock Quantification method of classification processes. concept of structural $a$-entropy.
\newblock {\em Kybernetika}, 3(1):30--35, 1967.
\newblock URL: \url{https://eudml.org/doc/28681}.

\bibitem[Hel67]{Hel67}
Carl~W. Helstrom.
\newblock Detection theory and quantum mechanics.
\newblock {\em Information and Control}, 10(3):254--291, 1967.
\newblock \href {https://doi.org/10.1016/S0019-9958(67)90302-6} {\path{doi:10.1016/S0019-9958(67)90302-6}}.

\bibitem[HHJ{\etalchar{+}}17]{HHJ+17}
Jeongwan Haah, Aram~W. Harrow, Zhengfeng Ji, Xiaodi Wu, and Nengkun Yu.
\newblock Sample-optimal tomography of quantum states.
\newblock {\em IEEE Transactions on Information Theory}, 63(9):5628--5641, 2017.
\newblock \prelimVersion{STOC}{2016}.
\newblock \href {https://arxiv.org/abs/1508.01797} {\path{arXiv:1508.01797}}, \href {https://doi.org/10.1109/TIT.2017.2719044} {\path{doi:10.1109/TIT.2017.2719044}}.

\bibitem[Hoe63]{Hoe63}
Wassily Hoeffding.
\newblock Probability inequalities for sums of bounded random variables.
\newblock {\em Journal of the American Statistical Association}, 58(301):13--30, 1963.
\newblock \href {https://doi.org/10.1080/01621459.1963.10500830} {\path{doi:10.1080/01621459.1963.10500830}}.

\bibitem[Hol73a]{Holevo73}
Alexander~S Holevo.
\newblock Bounds for the quantity of information transmitted by a quantum communication channel.
\newblock {\em Problemy Peredachi Informatsii}, 9(3):3--11, 1973.
\newblock URL: \url{https://www.mathnet.ru/eng/ppi903}.

\bibitem[Hol73b]{Hol73}
Alexander~S. Holevo.
\newblock Statistical decision theory for quantum systems.
\newblock {\em Journal of Multivariate Analysis}, 3(4):337--394, 1973.
\newblock \href {https://doi.org/10.1016/0047-259X(73)90028-6} {\path{doi:10.1016/0047-259X(73)90028-6}}.

\bibitem[JMDA21]{JMA21}
Reza~Asgharzadeh Jelodar, Hossein Mehri-Dehnavi, and Hamzeh Agahi.
\newblock Some properties of {Tsallis} and {Tsallis--Lin} quantum relative entropies.
\newblock {\em Physica A: Statistical Mechanics and its Applications}, 567:125719, 2021.
\newblock \href {https://doi.org/10.1016/j.physa.2020.125719} {\path{doi:10.1016/j.physa.2020.125719}}.

\bibitem[JUW09]{JUW09}
Rahul Jain, Sarvagya Upadhyay, and John Watrous.
\newblock Two-message quantum interactive proofs are in $\mathsf{PSPACE}$.
\newblock In {\em Proceedings of the 50th Annual IEEE Symposium on Foundations of Computer Science}, pages 534--543. IEEE, 2009.
\newblock \href {https://arxiv.org/abs/0905.1300} {\path{arXiv:0905.1300}}, \href {https://doi.org/10.1109/focs.2009.30} {\path{doi:10.1109/focs.2009.30}}.

\bibitem[JVHW15]{JVHW15}
Jiantao Jiao, Kartik Venkat, Yanjun Han, and Tsachy Weissman.
\newblock Minimax estimation of functionals of discrete distributions.
\newblock {\em IEEE Transactions on Information Theory}, 61(5):2835--2885, 2015.
\newblock \href {https://arxiv.org/abs/1406.6956} {\path{arXiv:1406.6956}}, \href {https://doi.org/10.1109/TIT.2015.2412945} {\path{doi:10.1109/TIT.2015.2412945}}.

\bibitem[JVHW17]{JVHW17}
Jiantao Jiao, Kartik Venkat, Yanjun Han, and Tsachy Weissman.
\newblock Maximum likelihood estimation of functionals of discrete distributions.
\newblock {\em IEEE Transactions on Information Theory}, 63(10):6774--6798, 2017.
\newblock \href {https://arxiv.org/abs/1406.6959} {\path{arXiv:1406.6959}}, \href {https://doi.org/10.1109/TIT.2017.2733537} {\path{doi:10.1109/TIT.2017.2733537}}.

\bibitem[Kim16]{Kim16}
Jeong~San Kim.
\newblock Tsallis entropy and general polygamy of multiparty quantum entanglement in arbitrary dimensions.
\newblock {\em Physical Review A}, 94(6):062338, 2016.
\newblock \href {https://arxiv.org/abs/1612.04480} {\path{arXiv:1612.04480}}, \href {https://doi.org/10.1103/physreva.94.062338} {\path{doi:10.1103/physreva.94.062338}}.

\bibitem[Kit95]{Kitaev95}
Alexei~Yu. Kitaev.
\newblock Quantum measurements and the {Abelian} stabilizer problem.
\newblock ArXiv e-prints, 1995.
\newblock \href {https://arxiv.org/abs/quant-ph/9511026} {\path{arXiv:quant-ph/9511026}}.

\bibitem[KLGN19]{KLGN19}
Hirotada Kobayashi, François Le~Gall, and Harumichi Nishimura.
\newblock Generalized quantum {Arthur--Merlin} games.
\newblock {\em SIAM Journal on Computing}, 48(3):865--902, 2019.
\newblock \prelimVersion{CCC}{2015}.
\newblock \href {https://arxiv.org/abs/1312.4673} {\path{arXiv:1312.4673}}, \href {https://doi.org/10.1137/17m1160173} {\path{doi:10.1137/17m1160173}}.

\bibitem[KLL{\etalchar{+}}17]{KLL+17}
Shelby Kimmel, Cedric Yen-Yu Lin, Guang~Hao Low, Maris Ozols, and Theodore~J. Yoder.
\newblock Hamiltonian simulation with optimal sample complexity.
\newblock {\em npj Quantum Information}, 3(1):1--7, 2017.
\newblock \href {https://arxiv.org/abs/1608.00281} {\path{arXiv:1608.00281}}, \href {https://doi.org/10.1038/s41534-017-0013-7} {\path{doi:10.1038/s41534-017-0013-7}}.

\bibitem[Kob03]{Kobayashi03}
Hirotada Kobayashi.
\newblock Non-interactive quantum perfect and statistical zero-knowledge.
\newblock In {\em Proceedings of the 14th International Symposium on Algorithms and Computation}, pages 178--188. Springer, 2003.
\newblock \href {https://arxiv.org/abs/quant-ph/0207158} {\path{arXiv:quant-ph/0207158}}, \href {https://doi.org/10.1007/978-3-540-24587-2_20} {\path{doi:10.1007/978-3-540-24587-2_20}}.

\bibitem[KS16]{KS16}
Yasuhito Kawano and Hiroshi Sekigawa.
\newblock Quantum {Fourier} transform over symmetric groups --- improved result.
\newblock {\em Journal of Symbolic Computation}, 75:219--243, 2016.
\newblock \href {https://doi.org/10.1016/j.jsc.2015.11.016} {\path{doi:10.1016/j.jsc.2015.11.016}}.

\bibitem[LC19]{LC19}
Guang~Hao Low and Isaac~L. Chuang.
\newblock Hamiltonian simulation by qubitization.
\newblock {\em Quantum}, 3:163, 2019.
\newblock \href {https://arxiv.org/abs/1707.05391} {\path{arXiv:1707.05391}}, \href {https://doi.org/10.22331/q-2019-07-12-163} {\path{doi:10.22331/q-2019-07-12-163}}.

\bibitem[LGDC24]{LGDC24}
Zhenhuan Liu, Weiyuan Gong, Zhenyu Du, and Zhenyu Cai.
\newblock Exponential separations between quantum learning with and without purification.
\newblock ArXiv e-prints, 2024.
\newblock \href {https://arxiv.org/abs/2410.17718} {\path{arXiv:2410.17718}}.

\bibitem[Lie73]{Lieb73}
Elliott~H. Lieb.
\newblock Convex trace functions and the {Wigner-Yanase-Dyson} conjecture.
\newblock {\em Advances in Mathematics}, 11(3):267--288, 1973.
\newblock \href {https://doi.org/10.1016/0001-8708(73)90011-x} {\path{doi:10.1016/0001-8708(73)90011-x}}.

\bibitem[Lin75]{Lindblad75}
G{\"o}ran Lindblad.
\newblock Completely positive maps and entropy inequalities.
\newblock {\em Communications in Mathematical Physics}, 40:147--151, 1975.
\newblock \href {https://doi.org/10.1007/bf01609396} {\path{doi:10.1007/bf01609396}}.

\bibitem[Lin91]{Lin91}
Jianhua Lin.
\newblock Divergence measures based on the {Shannon} entropy.
\newblock {\em IEEE Transactions on Information Theory}, 37(1):145--151, 1991.
\newblock \href {https://doi.org/10.1109/18.61115} {\path{doi:10.1109/18.61115}}.

\bibitem[Liu25a]{Liu25}
Yupan Liu.
\newblock {\em Complexity-theoretic perspectives on quantum state testing}.
\newblock PhD thesis, Nagoya University, 2025.
\newblock URL: \url{https://nagoya.repo.nii.ac.jp/records/2012662}.

\bibitem[Liu25b]{Liu23}
Yupan Liu.
\newblock Quantum state testing beyond the polarizing regime and quantum triangular discrimination.
\newblock {\em Computational Complexity}, 34(11):1--67, 2025.
\newblock \href {https://arxiv.org/abs/2303.01952} {\path{arXiv:2303.01952}}, \href {https://doi.org/10.1007/s00037-025-00273-8} {\path{doi:10.1007/s00037-025-00273-8}}.

\bibitem[Liu26]{Liu26}
Yupan Liu.
\newblock Computational hardness of estimating quantum entropies via binary entropy bounds.
\newblock In {\em 43rd International Symposium on Theoretical Aspects of Computer Science (STACS 2026)}, volume 364 of {\em Leibniz International Proceedings in Informatics (LIPIcs)}, pages 21:1--21:23. Schloss Dagstuhl -- Leibniz-Zentrum f{\"u}r Informatik, 2026.
\newblock \href {https://arxiv.org/abs/2601.03734} {\path{arXiv:2601.03734}}, \href {https://doi.org/10.4230/LIPIcs.STACS.2026.66} {\path{doi:10.4230/LIPIcs.STACS.2026.66}}.

\bibitem[LLW26]{LGLW23}
Fran{\c{c}}ois {Le Gall}, Yupan Liu, and Qisheng Wang.
\newblock Space-bounded quantum state testing via space-efficient quantum singular value transformation.
\newblock {\em \emph{To appear in} computational complexity}, 2026.
\newblock \href {https://arxiv.org/abs/2308.05079} {\path{arXiv:2308.05079}}, \href {https://doi.org/10.1007/s00037-025-00284-5} {\path{doi:10.1007/s00037-025-00284-5}}.

\bibitem[LMR14]{LMR14}
Seth Lloyd, Masoud Mohseni, and Patrick Rebentrost.
\newblock Quantum principal component analysis.
\newblock {\em Nature Physics}, 10(9):631--633, 2014.
\newblock \href {https://arxiv.org/abs/1307.0401} {\path{arXiv:1307.0401}}, \href {https://doi.org/10.1038/nphys3029} {\path{doi:10.1038/nphys3029}}.

\bibitem[LS20]{LS20}
Alessandro Luongo and Changpeng Shao.
\newblock Quantum algorithms for spectral sums.
\newblock ArXiv e-prints, 2020.
\newblock \href {https://arxiv.org/abs/2011.06475} {\path{arXiv:2011.06475}}.

\bibitem[LW19]{LW19}
Tongyang Li and Xiaodi Wu.
\newblock Quantum query complexity of entropy estimation.
\newblock {\em IEEE Transactions on Information Theory}, 65(5):2899--2921, 2019.
\newblock \href {https://arxiv.org/abs/1710.06025} {\path{arXiv:1710.06025}}, \href {https://doi.org/10.1109/TIT.2018.2883306} {\path{doi:10.1109/TIT.2018.2883306}}.

\bibitem[LW25a]{LW25Lalpha}
Yupan Liu and Qisheng Wang.
\newblock On estimating the quantum $\ell_{\alpha}$ distance.
\newblock In {\em Proceedings of the 33rd Annual European Symposium on Algorithms ({ESA} 2025)}, volume 351 of {\em LIPIcs}, pages 105:1--105:20. Schloss Dagstuhl - Leibniz-Zentrum f{\"{u}}r Informatik, 2025.
\newblock \href {https://arxiv.org/abs/2505.00457} {\path{arXiv:2505.00457}}, \href {https://doi.org/10.4230/LIPIcs.ESA.2025.105} {\path{doi:10.4230/LIPIcs.ESA.2025.105}}.

\bibitem[LW25b]{LW25SODA}
Yupan Liu and Qisheng Wang.
\newblock On estimating the trace of quantum state powers.
\newblock In {\em Proceedings of the 2025 Annual ACM-SIAM Symposium on Discrete Algorithms (SODA)}, pages 947--993. SIAM, 2025.
\newblock \href {https://arxiv.org/abs/2410.13559} {\path{arXiv:2410.13559}}, \href {https://doi.org/10.1137/1.9781611978322.28} {\path{doi:10.1137/1.9781611978322.28}}.

\bibitem[LWL24]{LWL24}
Jingquan Luo, Qisheng Wang, and Lvzhou Li.
\newblock Succinct quantum testers for closeness and $k$-wise uniformity of probability distributions.
\newblock {\em IEEE Transactions on Information Theory}, 70(7):5092--5103, 2024.
\newblock \href {https://arxiv.org/abs/2304.12916} {\path{arXiv:2304.12916}}, \href {https://doi.org/10.1109/TIT.2024.3393756} {\path{doi:10.1109/TIT.2024.3393756}}.

\bibitem[LWWZ25]{LWWZ25}
Nana Liu, Qisheng Wang, Mark~M. Wilde, and Zhicheng Zhang.
\newblock Quantum algorithms for matrix geometric means.
\newblock {\em npj Quantum Information}, 11:101, 2025.
\newblock \href {https://arxiv.org/abs/2405.00673} {\path{arXiv:2405.00673}}, \href {https://doi.org/10.1038/s41534-025-00973-7} {\path{doi:10.1038/s41534-025-00973-7}}.

\bibitem[MLP05]{MLP05}
Ana~P. Majtey, Pedro~W. Lamberti, and Domingo~P. Prato.
\newblock {Jensen-Shannon} divergence as a measure of distinguishability between mixed quantum states.
\newblock {\em Physical Review A}, 72(5):052310, 2005.
\newblock \href {https://arxiv.org/abs/quant-ph/0508138} {\path{arXiv:quant-ph/0508138}}, \href {https://doi.org/10.1103/PhysRevA.72.052310} {\path{doi:10.1103/PhysRevA.72.052310}}.

\bibitem[Mon15]{Mon15}
Ashley Montanaro.
\newblock Quantum speedup of {Monte} {Carlo} methods.
\newblock {\em Proceedings of the Royal Society A}, 471(2181):20150301, 2015.
\newblock \href {https://arxiv.org/abs/1504.06987} {\path{arXiv:1504.06987}}, \href {https://doi.org/10.1098/rspa.2015.0301} {\path{doi:10.1098/rspa.2015.0301}}.

\bibitem[M{\dutchPrefix{Wolf}{d}}W16]{MdW16}
Ashley Montanaro and Ronald {\dutchPrefix{Wolf}{d}}e~Wolf.
\newblock A survey of quantum property testing.
\newblock In {\em Theory of Computing Library}, number~7 in Graduate Surveys, pages 1--81. University of Chicago, 2016.
\newblock \href {https://arxiv.org/abs/1310.2035} {\path{arXiv:1310.2035}}, \href {https://doi.org/10.4086/toc.gs.2016.007} {\path{doi:10.4086/toc.gs.2016.007}}.

\bibitem[NC10]{NC10}
Michael~A. Nielsen and Isaac~L. Chuang.
\newblock {\em Quantum Computation and Quantum Information}.
\newblock Cambridge University Press, 2010.
\newblock \href {https://doi.org/10.1017/CBO9780511976667} {\path{doi:10.1017/CBO9780511976667}}.

\bibitem[Ngu23]{Nguyen23}
Quynh~T. Nguyen.
\newblock The mixed {Schur} transform: efficient quantum circuit and applications.
\newblock ArXiv e-prints, 2023.
\newblock \href {https://arxiv.org/abs/2310.01613} {\path{arXiv:2310.01613}}.

\bibitem[NRTM25]{NRTM25}
Ryotaro Niwa, Zane~Marius Rossi, Philip Taranto, and Mio Murao.
\newblock Singular value transformation for unknown quantum channels.
\newblock ArXiv e-prints, 2025.
\newblock \href {https://arxiv.org/abs/2506.24112} {\path{arXiv:2506.24112}}.

\bibitem[OW16]{OW16}
Ryan O'Donnell and John Wright.
\newblock Efficient quantum tomography.
\newblock In {\em Proceedings of the 48th Annual ACM Symposium on Theory of Computing}, pages 899--912, 2016.
\newblock \href {https://arxiv.org/abs/1508.01907} {\path{arXiv:1508.01907}}, \href {https://doi.org/10.1145/2897518.2897544} {\path{doi:10.1145/2897518.2897544}}.

\bibitem[OW21]{OW21}
Ryan O'Donnell and John Wright.
\newblock Quantum spectrum testing.
\newblock {\em Communications in Mathematical Physics}, 387(1):1--75, 2021.
\newblock \prelimVersion{STOC}{2015}.
\newblock \href {https://arxiv.org/abs/1501.05028} {\path{arXiv:1501.05028}}, \href {https://doi.org/10.1007/s00220-021-04180-1} {\path{doi:10.1007/s00220-021-04180-1}}.

\bibitem[Pet07]{Petz07}
D{\'e}nes Petz.
\newblock {\em Quantum Information Theory and Quantum Statistics}.
\newblock Springer, 2007.
\newblock \href {https://doi.org/10.1007/978-3-540-74636-2} {\path{doi:10.1007/978-3-540-74636-2}}.

\bibitem[QKW24]{QKW22}
Yihui Quek, Eneet Kaur, and Mark~M. Wilde.
\newblock Multivariate trace estimation in constant quantum depth.
\newblock {\em Quantum}, 8:1220, 2024.
\newblock \href {https://arxiv.org/abs/2206.15405} {\path{arXiv:2206.15405}}, \href {https://doi.org/10.22331/Q-2024-01-10-1220} {\path{doi:10.22331/Q-2024-01-10-1220}}.

\bibitem[Rag95]{Raggio95}
Guido~A. Raggio.
\newblock Properties of $q$-entropies.
\newblock {\em Journal of Mathematical Physics}, 36(9):4785--4791, 1995.
\newblock \href {https://doi.org/10.1063/1.530920} {\path{doi:10.1063/1.530920}}.

\bibitem[Ras11]{Rastegin11}
Alexey~E. Rastegin.
\newblock Some general properties of unified entropies.
\newblock {\em Journal of Statistical Physics}, 143:1120--1135, 2011.
\newblock \href {https://arxiv.org/abs/1012.5356} {\path{arXiv:1012.5356}}, \href {https://doi.org/10.1007/s10955-011-0231-x} {\path{doi:10.1007/s10955-011-0231-x}}.

\bibitem[RASW23]{RASW23}
Soorya Rethinasamy, Rochisha Agarwal, Kunal Sharma, and Mark~M. Wilde.
\newblock Estimating distinguishability measures on quantum computers.
\newblock {\em Physical Review A}, 108(1):012409, 2023.
\newblock \href {https://arxiv.org/abs/2108.08406} {\path{arXiv:2108.08406}}, \href {https://doi.org/10.1103/PhysRevA.108.012409} {\path{doi:10.1103/PhysRevA.108.012409}}.

\bibitem[Riv90]{Rivlin90}
Theodore~J. Rivlin.
\newblock {\em Chebyshev Polynomials: From Approximation Theory to Algebra and Number Theory}.
\newblock Courier Dover Publications, 1990.

\bibitem[Roc70]{Rockafellar70}
Ralph~Tyrell Rockafellar.
\newblock {\em Convex Analysis}.
\newblock Princeton University Press, 1970.
\newblock \href {https://doi.org/10.1515/9781400873173} {\path{doi:10.1515/9781400873173}}.

\bibitem[Rus22]{Ruskai22}
Mary~Beth Ruskai.
\newblock Yet another proof of the joint convexity of relative entropy.
\newblock {\em Letters in Mathematical Physics}, 112(4):81, 2022.
\newblock \href {https://arxiv.org/abs/2112.13763} {\path{arXiv:2112.13763}}, \href {https://doi.org/10.1007/s11005-022-01562-x} {\path{doi:10.1007/s11005-022-01562-x}}.

\bibitem[SH21]{SH21}
Sathyawageeswar Subramanian and Min-Hsiu Hsieh.
\newblock Quantum algorithm for estimating $\alpha$-{Renyi} entropies of quantum states.
\newblock {\em Physical Review A}, 104(2):022428, 2021.
\newblock \href {https://arxiv.org/abs/1908.05251} {\path{arXiv:1908.05251}}, \href {https://doi.org/10.1103/PhysRevA.104.022428} {\path{doi:10.1103/PhysRevA.104.022428}}.

\bibitem[SLLJ25]{SLLJ24}
Myeongjin Shin, Junseo Lee, Seungwoo Lee, and Kabgyun Jeong.
\newblock Resource-efficient algorithm for estimating the trace of quantum state powers.
\newblock {\em Quantum}, 9:1832, 2025.
\newblock \href {https://arxiv.org/abs/2408.00314} {\path{arXiv:2408.00314}}, \href {https://doi.org/10.22331/q-2025-08-27-1832} {\path{doi:10.22331/q-2025-08-27-1832}}.

\bibitem[Sra21]{Sra21}
Suvrit Sra.
\newblock Metrics induced by {Jensen-Shannon} and related divergences on positive definite matrices.
\newblock {\em Linear Algebra and its Applications}, 616:125--138, 2021.
\newblock \href {https://arxiv.org/abs/1911.02643} {\path{arXiv:1911.02643}}, \href {https://doi.org/10.1016/j.laa.2020.12.023} {\path{doi:10.1016/j.laa.2020.12.023}}.

\bibitem[SV03]{SV97}
Amit Sahai and Salil Vadhan.
\newblock A complete problem for statistical zero knowledge.
\newblock {\em Journal of the ACM}, 50(2):196--249, 2003.
\newblock \prelimVersion{FOCS}{1997}. \ECCC{20}{00}{084}.
\newblock \href {https://doi.org/10.1145/636865.636868} {\path{doi:10.1145/636865.636868}}.

\bibitem[Tim63]{Timan63}
Aleksandr~F. Timan.
\newblock {\em Theory of Approximation of Functions of a Real Variable}, volume~34 of {\em International Series of Monographs on Pure and Applied Mathematics}.
\newblock Pergamon Press, 1963.
\newblock \href {https://doi.org/10.1016/c2013-0-05307-8} {\path{doi:10.1016/c2013-0-05307-8}}.

\bibitem[Top00]{Top00}
Flemming Tops{\o}e.
\newblock Some inequalities for information divergence and related measures of discrimination.
\newblock {\em IEEE Transactions on Information Theory}, 46(4):1602--1609, 2000.
\newblock \href {https://doi.org/10.1109/18.850703} {\path{doi:10.1109/18.850703}}.

\bibitem[Top01]{Topsoe01}
Flemming Tops{\o}e.
\newblock Bounds for entropy and divergence for distributions over a two-element set.
\newblock {\em Journal of Inequalities in Pure and Applied Mathematics}, 2(2), 2001.
\newblock URL: \url{https://eudml.org/doc/122035}.

\bibitem[Tsa88]{Tsallis88}
Constantino Tsallis.
\newblock Possible generalization of {Boltzmann-Gibbs} statistics.
\newblock {\em Journal of Statistical Physics}, 52:479--487, 1988.
\newblock \href {https://doi.org/10.1007/bf01016429} {\path{doi:10.1007/bf01016429}}.

\bibitem[Tsa01]{Tsallis01}
Constantino Tsallis.
\newblock {\em Nonextensive Statistical Mechanics and Its Applications}, chapter I. Nonextensive Statistical Mechanics and Thermodynamics: Historical Background and Present Status, page 3–98.
\newblock Springer, 2001.
\newblock \href {https://doi.org/10.1007/3-540-40919-x_1} {\path{doi:10.1007/3-540-40919-x_1}}.

\bibitem[Uhl77]{Uhlmann77}
A.~Uhlmann.
\newblock Relative entropy and the {Wigner-Yanase-Dyson-Lieb} concavity in an interpolation theory.
\newblock {\em Communications in Mathematical Physics}, 54:21--32, 1977.
\newblock \href {https://doi.org/10.1007/BF01609834} {\path{doi:10.1007/BF01609834}}.

\bibitem[Vad99]{Vad99}
Salil~Pravin Vadhan.
\newblock {\em A Study of Statistical Zero-Knowledge Proofs}.
\newblock Phd thesis, Massachusetts Institute of Technology, 1999.
\newblock URL: \url{https://people.seas.harvard.edu/~salil/research/phdthesis.pdf}.

\bibitem[Vaj70]{Vajda70}
Igor Vajda.
\newblock Note on discrimination information and variation.
\newblock {\em IEEE Transactions on Information Theory}, 16(6):771--773, 1970.
\newblock \href {https://doi.org/10.1109/TIT.1970.1054557} {\path{doi:10.1109/TIT.1970.1054557}}.

\bibitem[Vir19]{Virosztek19}
D{\'a}niel Virosztek.
\newblock Jointly convex quantum {Jensen} divergences.
\newblock {\em Linear Algebra and its Applications}, 576:67--78, 2019.
\newblock \href {https://arxiv.org/abs/1712.05324} {\path{arXiv:1712.05324}}, \href {https://doi.org/10.1016/j.laa.2018.03.002} {\path{doi:10.1016/j.laa.2018.03.002}}.

\bibitem[Vir21]{Virosztek21}
D{\'a}niel Virosztek.
\newblock The metric property of the quantum {Jensen-Shannon} divergence.
\newblock {\em Advances in Mathematics}, 380:107595, 2021.
\newblock \href {https://arxiv.org/abs/1910.10447} {\path{arXiv:1910.10447}}, \href {https://doi.org/10.1016/j.aim.2021.107595} {\path{doi:10.1016/j.aim.2021.107595}}.

\bibitem[Wan24]{Wang24pureQSD}
Qisheng Wang.
\newblock Optimal trace distance and fidelity estimations for pure quantum states.
\newblock {\em IEEE Transactions on Information Theory}, 70(12):8791--8805, 2024.
\newblock \href {https://arxiv.org/abs/2408.16655} {\path{arXiv:2408.16655}}, \href {https://doi.org/10.1109/TIT.2024.3447915} {\path{doi:10.1109/TIT.2024.3447915}}.

\bibitem[Wan25]{Wan25}
Qisheng Wang.
\newblock Information-theoretic lower bounds for approximating monomials via optimal quantum {Tsallis} entropy estimation.
\newblock ArXiv preprint, 2025.
\newblock \href {https://arxiv.org/abs/2509.03496} {\path{arXiv:2509.03496}}.

\bibitem[Wat02]{Wat02}
John Watrous.
\newblock Limits on the power of quantum statistical zero-knowledge.
\newblock In {\em Proceedings of the 43rd Annual IEEE Symposium on Foundations of Computer Science}, pages 459--468. IEEE, 2002.
\newblock \href {https://arxiv.org/abs/quant-ph/0202111} {\path{arXiv:quant-ph/0202111}}, \href {https://doi.org/10.1109/SFCS.2002.1181970} {\path{doi:10.1109/SFCS.2002.1181970}}.

\bibitem[Wat09]{Wat09}
John Watrous.
\newblock Zero-knowledge against quantum attacks.
\newblock {\em SIAM Journal on Computing}, 39(1):25--58, 2009.
\newblock \prelimVersion{STOC}{2006}.
\newblock \href {https://arxiv.org/abs/quant-ph/0511020} {\path{arXiv:quant-ph/0511020}}, \href {https://doi.org/10.1137/060670997} {\path{doi:10.1137/060670997}}.

\bibitem[WGL{\etalchar{+}}24]{WGL+22}
Qisheng Wang, Ji~Guan, Junyi Liu, Zhicheng Zhang, and Mingsheng Ying.
\newblock New quantum algorithms for computing quantum entropies and distances.
\newblock {\em IEEE Transactions on Information Theory}, 70(8):5653--5680, 2024.
\newblock \href {https://arxiv.org/abs/2203.13522} {\path{arXiv:2203.13522}}, \href {https://doi.org/10.1109/TIT.2024.3399014} {\path{doi:10.1109/TIT.2024.3399014}}.

\bibitem[Wil13]{Wil13}
Mark~M. Wilde.
\newblock {\em Quantum Information Theory}.
\newblock Cambridge University Press, 2013.
\newblock \href {https://doi.org/10.1017/9781316809976} {\path{doi:10.1017/9781316809976}}.

\bibitem[WSP{\etalchar{+}}24]{WSP+24}
Rafael Wagner, Zohar {Schwartzman-Nowik}, Ismael~L Paiva, Amit Te’eni, Antonio {Ruiz-Molero}, Rui~Soares Barbosa, Eliahu Cohen, and Ernesto~F Galv{\~a}o.
\newblock Quantum circuits for measuring weak values, {Kirkwood--Dirac} quasiprobability distributions, and state spectra.
\newblock {\em Quantum Science and Technology}, 9(1):015030, 2024.
\newblock \href {https://arxiv.org/abs/2302.00705} {\path{arXiv:2302.00705}}, \href {https://doi.org/10.1088/2058-9565/ad124c} {\path{doi:10.1088/2058-9565/ad124c}}.

\bibitem[WY16]{WY16}
Yihong Wu and Pengkun Yang.
\newblock Minimax rates of entropy estimation on large alphabets via best polynomial approximation.
\newblock {\em IEEE Transactions on Information Theory}, 62(6):3702--3720, 2016.
\newblock \href {https://arxiv.org/abs/1407.0381} {\path{arXiv:1407.0381}}, \href {https://doi.org/10.1109/TIT.2016.2548468} {\path{doi:10.1109/TIT.2016.2548468}}.

\bibitem[WZ24]{WZ24}
Qisheng Wang and Zhicheng Zhang.
\newblock Fast quantum algorithms for trace distance estimation.
\newblock {\em IEEE Transactions on Information Theory}, 70(4):2720--2733, 2024.
\newblock \href {https://arxiv.org/abs/2301.06783} {\path{arXiv:2301.06783}}, \href {https://doi.org/10.1109/TIT.2023.3321121} {\path{doi:10.1109/TIT.2023.3321121}}.

\bibitem[WZ25a]{WZ23b}
Qisheng Wang and Zhicheng Zhang.
\newblock Quantum lower bounds by sample-to-query lifting.
\newblock {\em SIAM Journal on Computing}, 54(5):1294--1334, 2025.
\newblock \href {https://arxiv.org/abs/2308.01794} {\path{arXiv:2308.01794}}, \href {https://doi.org/10.1137/24M1638616} {\path{doi:10.1137/24M1638616}}.

\bibitem[WZ25b]{WZ24b}
Qisheng Wang and Zhicheng Zhang.
\newblock Time-efficient quantum entropy estimator via samplizer.
\newblock {\em IEEE Transactions on Information Theory}, 71(12):9569--9599, 2025.
\newblock \prelimVersion{ESA}{2024}.
\newblock \href {https://doi.org/10.1109/TIT.2025.3576137} {\path{doi:10.1109/TIT.2025.3576137}}.

\bibitem[WZ26]{WZ24c}
Qisheng Wang and Zhicheng Zhang.
\newblock Sample-optimal quantum estimators for pure-state trace distance and fidelity via samplizer.
\newblock In {\em Proceedings of the 53nd International Colloquium on Automata, Languages, and Programming (ICALP 2026)}, volume 374 of {\em LIPIcs}, pages 154:1--154:21. Schloss Dagstuhl - Leibniz-Zentrum f{\"{u}}r Informatik, 2026.
\newblock \href {https://arxiv.org/abs/2410.21201} {\path{arXiv:2410.21201}}, \href {https://doi.org/10.4230/LIPIcs.ICALP.2026.154} {\path{doi:10.4230/LIPIcs.ICALP.2026.154}}.

\bibitem[WZC{\etalchar{+}}23]{WZC+23}
Qisheng Wang, Zhicheng Zhang, Kean Chen, Ji~Guan, Wang Fang, Junyi Liu, and Mingsheng Ying.
\newblock Quantum algorithm for fidelity estimation.
\newblock {\em IEEE Transactions on Information Theory}, 69(1):273--282, 2023.
\newblock \href {https://arxiv.org/abs/2103.09076} {\path{arXiv:2103.09076}}, \href {https://doi.org/10.1109/TIT.2022.3203985} {\path{doi:10.1109/TIT.2022.3203985}}.

\bibitem[WZL24]{WZL24}
Xinzhao Wang, Shengyu Zhang, and Tongyang Li.
\newblock A quantum algorithm framework for discrete probability distributions with applications to {R\'{e}nyi} entropy estimation.
\newblock {\em IEEE Transactions on Information Theory}, 70(5):3399--3426, 2024.
\newblock \href {https://arxiv.org/abs/2212.01571} {\path{arXiv:2212.01571}}, \href {https://doi.org/10.1109/TIT.2024.3382037} {\path{doi:10.1109/TIT.2024.3382037}}.

\bibitem[WZW23]{WZW23}
Youle Wang, Benchi Zhao, and Xin Wang.
\newblock Quantum algorithms for estimating quantum entropies.
\newblock {\em Physical Review Applied}, 19(4):044041, 2023.
\newblock \href {https://arxiv.org/abs/2203.02386} {\path{arXiv:2203.02386}}, \href {https://doi.org/10.1103/PhysRevApplied.19.044041} {\path{doi:10.1103/PhysRevApplied.19.044041}}.

\bibitem[Yam02]{Yamano02}
Takuya Yamano.
\newblock Some properties of $q$-logarithm and $q$-exponential functions in tsallis statistics.
\newblock {\em Physica A: Statistical Mechanics and its Applications}, 305(3-4):486--496, 2002.
\newblock \href {https://doi.org/10.1016/s0378-4371(01)00567-2} {\path{doi:10.1016/s0378-4371(01)00567-2}}.

\bibitem[Zha07]{Zhang07}
Zhengmin Zhang.
\newblock Uniform estimates on the {Tsallis} entropies.
\newblock {\em Letters in Mathematical Physics}, 80:171--181, 2007.
\newblock \href {https://doi.org/10.1007/s11005-007-0155-1} {\path{doi:10.1007/s11005-007-0155-1}}.

\end{thebibliography}

\appendix

\section{Omitted proofs}

\subsection{Omitted proof in \texorpdfstring{\Cref{sec:properties-QJTq}}{Section 4}}
\label{subsec:omitted-proofs-Tsallis-properties}
\TsallisEntropyUpperBoundSolution*

\begin{proof}
    We begin by noting that $\Hq(p) = \frac{1}{q-1} \rbra*{1 - \sum_{i\in[N]} p(i)^q}$ is concave (\Cref{lemma:Tsallis-entropy-properties}) for any fixed $q>1$. Consequently, an optimal solution $p_{\max}$ to the optimization problem specified in \Cref{eq:Tsallis-entropy-max} has a particular form. Specifically, $p_{\max}$ is one of probability distributions $p^{(k)}$ for integer $k\in [\floor*{N(1-\gamma)}]$ defined in \Cref{eq:Tsallis-max-solution-candiate} with a maximum Tsallis entropy:\footnote{It is easy to verify that $\frac{1}{N} - \frac{\gamma}{N-k} \geq 0$ holds if and only if $k \leq N(1-\gamma)$ holds. } 
    \begin{equation}
        \label{eq:Tsallis-max-solution-candiate}
        \Hq(p_{\max}) = \max_{k \in [\floor*{N(1-\gamma)}]} \Hq\rbra*{p^{(k)}},
        \text{ where }
        p^{(k)}(i) \coloneqq \begin{cases}
        \frac{1}{N}+\frac{\gamma}{k}, &\text{if } i\in[k]\\
        \frac{1}{N} - \frac{\gamma}{N-k}, &\text{otherwise}
        \end{cases}.
    \end{equation}

    Plugging \Cref{eq:Tsallis-max-solution-candiate} into \Cref{eq:Tsallis-entropy-max}, it suffices to solve the following optimization problem with $q>1$:
    \begin{mini}{}{F_q(N,k,\gamma) \coloneqq \sum_{i\in[N]} p(i)^q = k \cdot \rbra*{\frac{1}{N} + \frac{\gamma}{k}}^q + (N-k) \cdot \rbra*{\frac{1}{N} - \frac{\gamma}{N-k}}^q}{}{}{\label{eq:Tsallis-entropy-max-largeDist}}{}
        \addConstraint{1/q}{ \leq \gamma \leq 1-1/N}{}
        \addConstraint{1}{\leq k \leq \floor*{N(1-\gamma)}}{}
        \addConstraint{k, N}{\in \bbZ_+}{}
    \end{mini}

    To establish that \Cref{eq:Tsallis-entropy-max-solution} is an optimal solution to \Cref{eq:Tsallis-entropy-max-largeDist}, it remains to show that the objective function $F_q(N,k,\gamma)$ is monotonically non-increasing in $k$ for $N$, $\gamma$, and $q>1$ satisfying the constraints in \Cref{eq:Tsallis-entropy-max-largeDist}. Equivalently, it needs to be shown that $\frac{\partial}{\partial k} F_q(N,k,\gamma) \leq 0$ for $1/q \leq \gamma \leq 1-1/N$ and $1 \leq k \leq \floor*{N(1-\gamma)}$, specifically:
    \begin{equation}
        \label{eq:Tsallis-max-solution-partialDerivative}
        \frac{\partial}{\partial k} F_q(N,k,\gamma) =
        \frac{(k-\gamma  N (q-1)) \left(\frac{\gamma }{k}+\frac{1}{N}\right)^q}{k+\gamma  N}+\frac{\left(\frac{1}{N} - \frac{\gamma }{N-k}\right)^q (\gamma  N (q-1)+N-k)}{\gamma N - (N-k)} \leq 0.
    \end{equation}

    Since it is evident that $\frac{\gamma}{k}+\frac{1}{N} \geq 0$, $k + \gamma N \geq 0$, and $k \leq \floor*{N(1-\gamma)} \leq N(1-\gamma)$, we can deduce \Cref{eq:Tsallis-max-solution-partialDerivative} by combining the following inequalities:
    \begin{align*}
        k-\gamma  N (q-1) & \leq N(1-\gamma) -\gamma  N (q-1) = N(1-q\gamma) \leq 0,\\
        \frac{1}{N} - \frac{\gamma}{N-k} &\geq \frac{1}{N} - \frac{\gamma}{N-N(1-\gamma)} = 0,\\
        \gamma N (q-1) + N-k & \geq N (q-1) + N- N(1-\gamma) = Nq\gamma \geq N >0,\\
        \gamma N - (N-k) & \leq N - (N-N(1-\gamma)) = 0. 
    \end{align*}
    Here, the first and the third line hold also due to $\gamma \geq 1/q$. This completes the proof.
\end{proof}

\subsection{Omitted proof in \texorpdfstring{\Cref{sec:hardness-via-QJTq-reductions}}{Section 5}}
\label{subsec:omitted-proofs-QJTq-reductions}

\reductionQSCMMtoTsallisQEAsmallTD*

\begin{proof}
    We begin by defining $f_1(n) \coloneqq 2^{-n} + \frac{1-2^{\frac{n}{1-n}}}{n}$, $f_2(n) \coloneqq 2^{\frac{n^2}{1-n}} (n-1)$, and $f_3(n) \coloneqq \frac{n}{4} \rbra*{1- 2^{\frac{1}{1-n}}}$ such that $g_1(n) = f_1(n) + f_2(n) + f_3(n)$. We then prove the first item separately:
    \begin{itemize}[leftmargin=2em, topsep=0.33em, itemsep=0.33em, parsep=0.33em] 
        \item For $f_1(n)$, since $2^{\frac{n}{1-n}} = 2^{-\rbra*{1+\frac{1}{n-1}}}$, we know that $f_1(n)$ is monotonically decreasing for $n\geq 2$, and thus, $f_1(n) \geq \lim_{n\rightarrow \infty} f_1(n) = 0$ for $n \geq 2$. 
        \item For $f_2(n)$, noting that $\frac{\dd}{\dn} f_2(n) = \frac{2^{n^2/(1-n)}}{n-1}  \rbra*{-\log(2) n^2 + (1+\log(2))n -1}$, we obtain that $f_2(n)$ is monotonically decreasing for $n \geq 3 > \frac{1+2\log(2) + \sqrt{1+4 \log(2)^2}}{2 \log(2)} \approx 2.9544$, and consequently, $f_2(n) \geq \lim_{n \rightarrow \infty} g_2(n) = 0$ for $n \geq 3$. 
        \item For $f_3(n)$, it suffices to show that $2^{1/(1-n)} \leq 1$ for $n \geq 3$. Since $2^{1/(1-n)}$ is monotonically increasing for $n \geq 3$, we prove the first item by noting that $2^{1/(1-n)} \leq \lim_{n\rightarrow \infty} 2^{1/(1-n)} = 1$. 
    \end{itemize}   
    
    For $g_2(n)$, noting that $\frac{\dd}{\dn} g_3(n) = \frac{2^{\frac{1}{1-n}} \log (2)}{(n-1)^2}+2^{-n} (n \log(2)-1)$ and $n \log(2) \geq 1$ for $n \geq 2$, we obtain that $g_3(n)$ is monotonically increasing for $n \geq 2$. 
    
    For $g_3(n)$, since $2^{-1/x}$ is monotonically increasing for $x \geq 1$, we have $g_3(n) \geq \frac{1}{4} n (n^{1/n}-2^{-1/n})$. It remains to show that $\tilde{g}_3(n) \coloneqq \frac{1}{4} n (n^{1/n}-2^{-1/n})$ are monotonically increasing for $n\geq 3$, namely:
    \begin{equation}
        \label{eq:reduction-QSCMM-TsallisQEA-smallTD-mono}
        \frac{\dd}{\dn} \tilde{g}_3(n) = \frac{1}{4n} \underbrace{ \rbra*{n^{1/n} - 2^{-1/n}\log(2)}}_{f_4(n)} + \underbrace{\frac{1}{4n} \rbra*{n^{1/n} n - n^{1/n} \log(n) - n 2^{-1/n}}}_{f_5(n)} \geq 0.
    \end{equation}
    
    Noting that $\frac{\dd}{\dn} f_4(n) = -\frac{1}{n^2}\rbra*{n^{1/n} (\log (n)-1)+2^{-1/n} \log ^2(2)} < 0$ for $n > e$, namely $f_4(n)$ is monotonically decreasing for $n\geq 3$, we obtain that $\frac{f_4(n)}{4n} \geq \frac{1}{4n} \lim_{n\rightarrow\infty} f_4(n) = \frac{1}{4n} > 0$. 
    Let $f_6(n) \coloneqq \rbra*{\frac{1}{2n}}^{1/n}$. Notice that $\frac{\dd}{\dn} f_6(n) = 2^{-1/n} \left(\frac{1}{n}\right)^{\frac{1}{n}+2} \rbra*{-\log\big(\frac{1}{n}\big)-1+\log (2)} \geq 0$ for $n \geq e/2$, we have that $f_6(n) =\rbra*{\frac{1}{2n}}^{1/n} \geq f_6(2) =1/2$ for $n \geq 2$. Consequently, we can derive that: 
    \[ \rbra*{\frac{1}{n}}^{\frac{1}{n}} \frac{\dd f_5(n)}{\dn} = \frac{ (\log (n)-1) \log (n) - \rbra*{\frac{1}{2n}}^{\frac{1}{n}}  n \log (2)}{4n^3} \leq \frac{1}{4n^2} \rbra*{\frac{ (\log (n)-1) \log (n)}{n} - \frac{\log(2)}{2}} < 0.\]
    Here, the last inequality follows by assuming $f_7(n) \coloneqq \frac{\log(n) \rbra*{\log(n)-1}}{n} < \frac{\log(2)}{2}$.
    A direct calculation implies that $\frac{\dd}{\dn} f_7(n)=-\frac{1}{n^2}\rbra*{(\log (n)-3) \log (n)+1}=0$ have two zeros at $n = \exp\rbra{\frac{3 \pm \sqrt{5}}{2}}$. Therefore, we establish \Cref{eq:reduction-QSCMM-TsallisQEA-smallTD-mono} by noticing \[f_7(n) \leq \max\cbra*{f_7\rbra*{\frac{3 - \sqrt{5}}{2}}, f_7\rbra*{\frac{3 + \sqrt{5}}{2}}} < \frac{\log(2)}{2}. \qedhere\]
\end{proof}

\end{document}